\documentclass[review]{elsarticle}
\usepackage{geometry}
\usepackage{lineno,hyperref}
\modulolinenumbers[5]

\makeatletter
 \def\ps@pprintTitle{%
      \let\@oddhead\@empty
      \let\@evenhead\@empty
      \def\@oddfoot{\footnotesize\itshape
       \hfill\today}%
      \let\@evenfoot\@oddfoot}
\makeatother

\biboptions{authoryear}
\usepackage{booktabs}
\usepackage{longtable}
\usepackage{xcolor}
\usepackage{array,multirow,graphicx}
\usepackage{float}
\usepackage{makecell}
\usepackage[T1]{fontenc}
\usepackage[utf8]{inputenc}
\usepackage{enumitem}

\usepackage{amsmath}
\usepackage{amsthm}
\usepackage{amsfonts}
\usepackage{bm}
\usepackage[font = footnotesize]{caption}
\usepackage{subcaption}
\usepackage{setspace}
\usepackage[bottom,hang,flushmargin]{footmisc}

\theoremstyle{definition}

\newtheorem{theorem}{Theorem}

\usepackage{tikz}
\usetikzlibrary{arrows,shapes,positioning,shadows,trees, calc}
\usetikzlibrary{matrix,fit,arrows.meta, backgrounds, decorations.pathreplacing}

\tikzset{
  basic/.style  = {draw, text width=2cm, drop shadow, font=\sffamily, rectangle},
  root/.style   = {basic, rounded corners=2pt, thin, align=center,
                   fill=green!30},
  level 2/.style = {basic, rounded corners=6pt, thin,align=center, fill=green!60,
                   text width=4em},
  level 3/.style = {basic, thin, align=left, fill=pink!60, text width=1.5em}
}

\newcommand{\relation}[3]
{
	\draw (#3.south) -- +(0,-#1) -| ($ (#2.north) $)
}

\newcommand{\relationDashed}[3]
{
	\draw[dashed, line width=0.3mm] (#3.south) -- +(0,-#1) -| ($ (#2.north) $)
}

\input{./preamble/mathdef.tex}

\begin{document}

\begin{frontmatter}

\title{Enhancements in cross-temporal forecast reconciliation, with an application to solar irradiance forecasts}


\author[mymainaddress]{Tommaso Di Fonzo\corref{mycorrespondingauthor}}
\cortext[mycorrespondingauthor]{Corresponding author}
\ead{difonzo@stat.unipd.it}

\author[mymainaddress]{Daniele Girolimetto}

\address[mymainaddress]{Department of Statistical Sciences, University of Padua, \\ Via C. Battisti 241, 35121 Padova (Italy)}

\begin{abstract}
In recent works by \cite{Yang2017cs, Yang2017te}, 
and \cite{Yagli2019}, 
geographical, 
temporal, and sequential deterministic 
reconciliation of hierarchical photovoltaic (PV) power generation have been considered for a simulated PV dataset in California. 
In the first two cases, the reconciliations are carried out in spatial and temporal domains separately.  To further improve forecasting accuracy, in the third case these two reconciliation approaches are sequentially applied.
During the replication of the forecasting experiment,
some issues emerged about non-negativity and coherence (in space and/or in time) of the sequentially reconciled forecasts.
Furthermore, while the accuracy improvement of the considered approaches over the benchmark persistence forecasts is clearly visible at any data granularity, we argue that an even better performance may be obtained by a thorough exploitation of cross-temporal hierarchies.
To this end, in this paper the cross-temporal point forecast reconciliation approach is applied to generate non-negative, fully coherent (both in space and time) 
forecasts. In particular, (i) some useful relationships between two-step, iterative and simultaneous cross-temporal reconciliation procedures are for the first time established, (ii) non-negativity issues of the final reconciled forecasts are discussed and correctly dealt with in a simple and effective way, and (iii) the most recent cross-temporal reconciliation approaches proposed in literature are adopted.
The normalised Root Mean Square Error is used to measure forecasting accuracy, and a statistical multiple comparison procedure is performed to rank the approaches. Besides assuring 
full coherence, and non-negativity of the reconciled forecasts, the results show that for the considered dataset, cross-temporal forecast reconciliation significantly improves on 
the sequential procedures proposed by \cite{Yagli2019}, at any cross-sectional level of the hierarchy and for any temporal granularity.

\end{abstract}

\begin{keyword}
Forecasting\sep Cross-temporal forecast reconciliation\sep Sequential and iterative approaches\sep Non-negative forecasts\sep Photovoltaic power generation
\end{keyword}

\end{frontmatter}

\newpage

\section{Introduction}
\label{sec:1}

Traditional electricity relies heavily on fossil fuels such as coal and natural gas. Not only are they bad for the environment, but they are also limited resources.
Net-zero emissions by 2050 are crucial to achieve the core Paris Agreement\footnote{The Paris Agreement is a legally binding international treaty on climate change. 
Its goal is to limit global warming to well below 2, preferably to 1.5 degrees Celsius, compared to pre-industrial levels.	To achieve this long-term temperature goal, countries aim to reach global peaking of greenhouse gas emissions as soon as possible to achieve a climate neutral world (i.e., with net-zero greenhouse gas emissions) by mid-century.
}
goals of a global average temperature rise of 1.5 degrees Celsius (\citealp{UN2015}), and this in turn can only be achieved if global greenhouse gas emissions are halved by the end of this decade (\citealp{EC2019}, \citealp{UN2022}).
Solar power is one of the crucial production methods in the move to clean energy, and as economies of scale drive prices down, its importance will undoubtedly increase.
The deployment of solar-power generation is causing total installed capacity to increase at a very high pace. 
Eurostat reports that the European Union added 18,224.8 MW of net capacity in 2020, compared to its 16,146.9 MW increase in 2019, registering a growth of 12.9\%. At the end of 2020, the EU's photovoltaic base stood at 136,136.6 MW, which is a 15\% year-on-year increase (\citealp{Eurobserver2022}, p.~14).

Using solar resource as a stable source of energy is not an easy task. Estimating the solar energy potential is a key target to ensure its management in a reliable and efficient way for its integration into an electrical power grid.
 The prediction of solar irradiation despite its variability is particularly important as it is a precondition for (i) the management of the solar photovoltaic production through storage systems to reduce the impact of the intermittent nature of the solar resource, and (ii) the integration of the solar resource into a power grid in order to meet the local energy needs and to cope with the load fluctuations.
Understanding of the need for short, mid or long term prediction (e.g., 1h, 6h or a day ahead forecasting) 
is growing as utilities and grid operators gain experience in dealing with solar-power sources.
Increasing spatial and temporal resolution of the available forecasting models
would enable grid operators to better forecast how much solar energy will be added to the grid. These efforts will improve the management of solar power’s variability and uncertainty, enabling its more reliable and cost-effective integration onto the grid.


Solar forecasting is a fast-growing sub-domain of energy forecasting (\citealp{Yang2020verification}, p.~20).
We agree with the claim of (\citealp{Yang2022}, p.~7) that a ``common misconception is that the novelty in solar forecasting should be solely revolved around forecasting methodology. Indeed, forecasting methodology is an important aspect, but it is never the only one''.
Nevertheless, we think that a clear assessment of the available forecasting procedures may help in moving forward the fronteer of knowledge of this fundamental topic, generating beneficial effects on the activities of practitioners.

A major goal for solar forecasting is to provide information on future photovoltaic (PV) power generation at different locations, time scales, and horizons to power system operators (\citealp{Yang2022}).
In recent works by \cite{Yang2017cs, Yang2017te}, 
and \cite{Yagli2019}, 
geographical, 
temporal, and sequential deterministic 
reconciliation of hierarchical PV power generation have been considered for a simulated PV dataset in California. 
In the first two cases, the reconciliations are carried out in spatial and temporal domains separately.  To further improve prediction accuracy, in the third case these two reconciliation approaches are sequentially applied.
During the replication of the forecasting experiment\footnote{The forecasting experiment grounds on the documentation files and data made available by \cite{Yang2017te}.}, some issues emerged about non-negativity and coherency (in space and/or in time) of the sequentially reconciled forecasts.
Furthermore, 
while the accuracy improvement of the considered approaches over the benchmark persistence forecasts is clearly visible at any data granularity, we think that an even better performance may be obtained by a thorough exploitation of cross-temporal hierarchies.

Cross-temporal forecast reconciliation is a rather recent sub-domain of the general theme `forecast reconciliation' (\citealp{Fliedner2001}, \citealp{Athanasopoulos2009}, \citealp{Hyndman2011}). The idea of exploiting the aggregation relationships valid both in space (cross-sectional coherency), and for different time granularities (temporal coherency) to improve the forecast accuracy of base forecasts for a hierarchical time series, was discussed in the groundbreaking papers by \cite{Kourentzes2019}, \cite{Yagli2019}, \cite{Spiliotis2020b}, and \cite{Punia2020}. The relevance of the topic has been further confirmed by (i) the ISF 2020 keynote speech by prof. \cite{HyndmanISF2020}, (ii) the entry `cross-temporal hierarchies' in the encyclopedic review on the theory and practice of forecasting by \cite{Petropoulos2022}\footnote{The entry was written by N. Kourentzes.}, as well as (iii) the section entitled `cross-temporal hierarchies' in a recent review paper on aggregation and hierarchical approaches for demand forecasting in supply chains (\citealp{Babai2022}).
These important contributions paved the way for the findings of \cite{DiFonzoGiro2021}, 
who showed the potentiality of a number of new cross-temporal reconciliation approaches, and discussed fundamental feasibility issues, offering 
insights to the practitioner wishing to evaluate the effort-to-benefit ratio of using this forecasting device\footnote{All the forecast reconciliation procedures considered in \cite{DiFonzoGiro2021}, and in this paper, are available in the \texttt{R} package \texttt{FoReco} (\citealp{FoReco2022}).}.

In this paper, cross-temporal point forecast reconciliation is applied to generate non-negative, fully coherent (both in space and time) 
forecasts of PV generated power. In particular, (i) some useful relationships between two-step, iterative and simultaneous cross-temporal reconciliation procedures are for the first time established, (ii) non-negativity issues of the final reconciled forecasts are discussed and correctly dealt with in a simple and effective way, and (iii) the most recent cross-temporal reconconciliation approaches proposed in literature are adopted. The iterative and simultaneous approaches by \cite{DiFonzoGiro2021}, and the heuristic cross-temporal procedure proposed by \cite{Kourentzes2019} are applied to the base forecasts with forecast horizon of 1 day, of PV generated power at different time granularities (1 hour to 1 day), of a hierarchy consisting of 324 series along 3 levels. The normalised Root Mean Square Error is used to measure forecasting accuracy, and a statistical multiple comparison procedure is performed to rank the approaches.

The paper is organized as follows. The deterministic (point) cross-temporal forecast reconciliation framework is described in section \ref{sec:ctrecap}, and in section \ref{sec:techdigr} some useful connections between apparently different approaches are shown.
The forecasting experiment of \cite{Yagli2019} is replicated and discussed in section \ref{sec:repass}, and the performance of the newly proposed forecasting approaches is presented in section \ref{extanalysis}. Conclusions follow in section \ref{sec:conclusion}.

\newpage

\section{Cross-temporal point forecast reconciliation}
\label{sec:ctrecap}

\subsection{Problem definition}
\label{subsec:probdef}
To begin with, consider the very simple example of a two-level cross-sectional hierarchy, where the top variable $X$ is equal to the sum of two bottom series\footnote{In this paper, we consider only genuine hierarchical/grouped time series, that share the same top- and bottom-level variables. The treatment of a general linearly constrained multiple time series is discussed in \cite{DiFonzoGiro2021}.}, $W$ and $Z$. Further, assume that the highest time frequency the variables are observed at is quarterly, which means that by simple non-overlapping temporal aggregation of quarterly time series, semi-annual and annual time series may be obtained as well.
Figure \ref{toyquarterlyhierarchy} gives a visual representation of such cross-temporal hierarchy for a time cycle of 1 year. 

\begin{figure}[ht]
	\centering
	\resizebox{0.49\linewidth}{!}{
			\begin{tikzpicture}[baseline=(current  bounding  box.center),
			every node/.append style={shape=ellipse,
				draw=black},
			minimum width=1.2cm,
			minimum height=1.2cm]

			\node at (0, 0) (5Q1){$w^{[1]}_1$};
			\node at (1.5, 0) (5Q2){$w^{[1]}_2$};
			\node at (3, 0) (5Q3){$w^{[1]}_3$};
			\node at (4.5, 0) (5Q4){$w^{[1]}_4$};
			\node at (0.75, 1.8) (5SA1){$w^{[2]}_1$};
			\node at (3.75, 1.8) (5SA2){$w^{[2]}_2$};
			\node at (2.25, 3.6) (5A){$w^{[4]}_1$};
			\relation{0.2}{5Q1}{5SA1};
			\relation{0.2}{5Q2}{5SA1};
			\relation{0.2}{5Q3}{5SA2};
			\relation{0.2}{5Q4}{5SA2};
			\relation{0.2}{5SA1}{5A};
			\relation{0.2}{5SA2}{5A};
			\node[draw=none, align=center] at (0.25,3.6) {\Large $W$};
			
			\node at (8, 0) (6Q1){$z^{[1]}_1$};
			\node at (9.5, 0) (6Q2){$z^{[1]}_2$};
			\node at (11, 0) (6Q3){$z^{[1]}_3$};
			\node at (12.5, 0) (6Q4){$z^{[1]}_4$};
			\node at (8.75, 1.8) (6SA1){$z^{[2]}_1$};
			\node at (11.75, 1.8) (6SA2){$z^{[2]}_2$};
			\node at (10.25, 3.6) (6A){$z^{[4]}_1$};
			\relation{0.2}{6Q1}{6SA1};
			\relation{0.2}{6Q2}{6SA1};
			\relation{0.2}{6Q3}{6SA2};
			\relation{0.2}{6Q4}{6SA2};
			\relation{0.2}{6SA1}{6A};
			\relation{0.2}{6SA2}{6A};
			\node[draw=none, align=center] at (8.25,3.6) {\Large $Z$};
			
			\node at (4, 6.5) (7Q1){$x^{[1]}_1$};
			\node at (5.5, 6.5) (7Q2){$x^{[1]}_2$};
			\node at (7, 6.5) (7Q3){$x^{[1]}_3$};
			\node at (8.5, 6.5) (7Q4){$x^{[1]}_4$};
			\node at (4.75, 8.3) (7SA1){$x^{[2]}_1$};
			\node at (7.75, 8.3) (7SA2){$x^{[2]}_2$};
			\node at (6.25, 10.1) (7A){$x^{[4]}_1$};
			\relation{0.2}{7Q1}{7SA1};
			\relation{0.2}{7Q2}{7SA1};
			\relation{0.2}{7Q3}{7SA2};
			\relation{0.2}{7Q4}{7SA2};
			\relation{0.2}{7SA1}{7A};
			\relation{0.2}{7SA2}{7A};
			\node[draw=none, align=center] at (4.25,10.1) {\Large $X$};
			
			\node[draw=black, shape=rectangle, dashed, line width=0.3mm, fit={(5Q1) (5Q2) (5Q3) (5Q4) (5SA1) (5SA2) (5A)},
			minimum width=7cm,
			minimum height=5.5cm](U1){};
			\node[draw=black, shape=rectangle, dashed, line width=0.3mm, fit={(6Q1) (6Q2) (6Q3) (6Q4) (6SA1) (6SA2) (6A)},
			minimum width=7cm,
			minimum height=5.5cm](U2){};
			\node[draw=black, shape=rectangle, dashed, line width=0.3mm, fit={(7Q1) (7Q2) (7Q3) (7Q4) (7SA1) (7SA2) (7A)},
			minimum width=7cm,
			minimum height=5.5cm](T){};
			\relationDashed{0.3}{U1}{T};
			\relationDashed{0.3}{U2}{T};
		\end{tikzpicture}}\hfill\vline\hfill\resizebox{0.49\linewidth}{!}{\begin{tikzpicture}[baseline=(current  bounding  box.center),
	every node/.append style={shape=ellipse,
		draw=black},
	minimum width=1.2cm,
	minimum height=1.2cm]

	\node at (0, 0) (5Q1){$w^{[1]}_1$};
	\node at (1.5, 0) (5Q2){$w^{[1]}_2$};
	\node at (3, 0) (5Q3){$w^{[1]}_3$};
	\node at (4.5, 0) (5Q4){$w^{[1]}_4$};
	\node at (2.25, 1.8) (5A){$w^{[4]}_1$};
	\relation{0.2}{5Q1}{5A};
	\relation{0.2}{5Q2}{5A};
	\relation{0.2}{5Q3}{5A};
	\relation{0.2}{5Q4}{5A};
	\node[draw=none, align=center] at (0.25,1.8) {\Large $W$};
	
	\node at (8, 0) (6Q1){$z^{[1]}_1$};
	\node at (9.5, 0) (6Q2){$z^{[1]}_2$};
	\node at (11, 0) (6Q3){$z^{[1]}_3$};
	\node at (12.5, 0) (6Q4){$z^{[1]}_4$};
	\node at (10.25, 1.8) (6A){$z^{[4]}_1$};
	\relation{0.2}{6Q1}{6A};
	\relation{0.2}{6Q2}{6A};
	\relation{0.2}{6Q3}{6A};
	\relation{0.2}{6Q4}{6A};
	\node[draw=none, align=center] at (8.25,1.8) {\Large $Z$};
	
	\node at (4, 4.5) (7Q1){$x^{[1]}_1$};
			\node at (5.5, 4.5) (7Q2){$x^{[1]}_2$};
			\node at (7, 4.5) (7Q3){$x^{[1]}_3$};
			\node at (8.5, 4.5) (7Q4){$x^{[1]}_4$};
	\node at (6.25, 6.3) (7A){$x^{[4]}_1$};
	\relation{0.2}{7Q1}{7A};
	\relation{0.2}{7Q2}{7A};
	\relation{0.2}{7Q3}{7A};
	\relation{0.2}{7Q4}{7A};
	\node[draw=none, align=center] at (4.25,6.3) {\Large $X$};
	
	\node[draw=black, shape=rectangle, dashed, line width=0.3mm, fit={(5Q1) (5Q2) (5Q3) (5Q4) (5A)},
	minimum width=7cm,
	minimum height=3.5cm](U1){};
	\node[draw=black, shape=rectangle, dashed, line width=0.3mm, fit={(6Q1) (6Q2) (6Q3) (6Q4) (6A)},
	minimum width=7cm,
	minimum height=3.5cm](U2){};
	\node[draw=black, shape=rectangle, dashed, line width=0.3mm, fit={(7Q1) (7Q2) (7Q3) (7Q4) (7A)},
	minimum width=7cm,
	minimum height=3.5cm](T){};
	\relationDashed{0.3}{U1}{T};
	\relationDashed{0.3}{U2}{T};
	\end{tikzpicture}
	}
	\caption{Complete (left) and reduced (right) cross-temporal hierarchies for a quarterly two-level hierarchical time series.}
        \label{toyquarterlyhierarchy}
\end{figure}
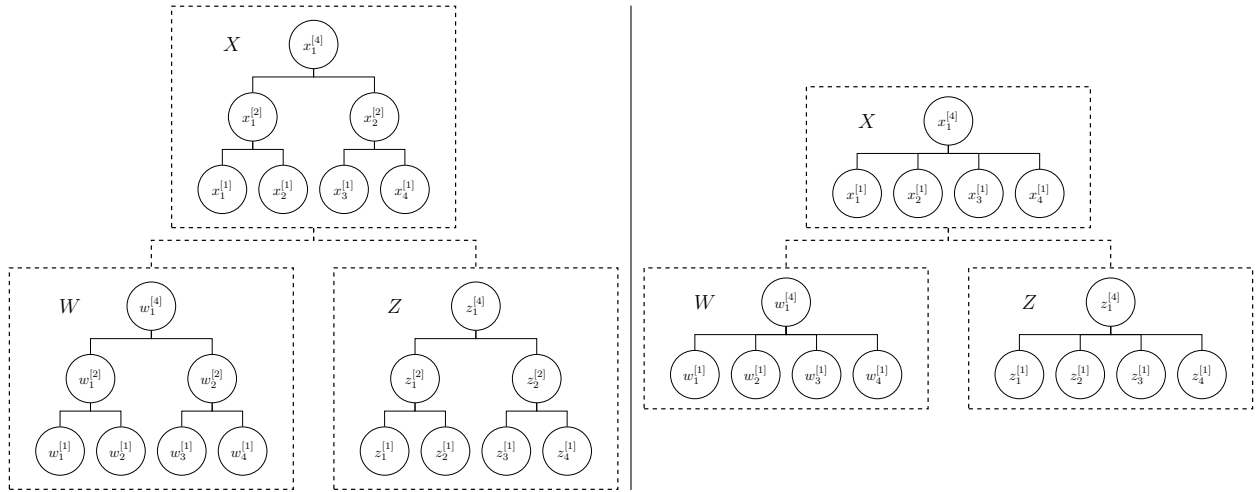

\noindent While the left panel of figure \ref{toyquarterlyhierarchy} shows the complete cross-temporal hierarchy, consisting of all aggregated temporal granularities which may be defined starting from quarterly data (i.e, semi-annual, and annual), in the right panel is represented its reduced version, where only the highest (quarterly) and the lowest (annual) time granularities are respectively considered.
The square boxes in the figure denote the nodes of a two-level cross-sectional (contemporaneous) hierarchy, while the circles denote the nodes of the temporal hierarchies.
In agreement with a standard notation in the temporal forecast reconciliation literature (\citealp{Athanasopoulos2017}, \citealp{Yang2017te}), the superscript $[k]$ denotes the temporal aggregation order for each time granularity, i.e. annual ($k=4$), semi-annual ($k=2$), and quarterly ($k=1$).
The cross-sectional hierarchy is described by the aggregation relationship $X = W + Z$, which is valid for any temporal aggregation order $k \in \mathcal{K}=\{4,2,1\}$ (i.e., $x_{\tau}^{[k]} = w_{\tau}^{[k]} + z_{\tau}^{[k]}$, $\tau = 1,\ldots, 4/k$).
Assuming $v$ alternatively equal to $x, w, z$, the temporal hierarchies
describing the relationships between different time granularities of a single time series, may be expressed as
$$
\begin{array}{rcll}
v^{[4]}_1 & = & v^{[2]}_1 + v^{[2]}_2 & \text{{\small the annual value is the sum of the two semi-annual values},} \\
v^{[2]}_1 & = & v^{[1]}_1 + v^{[1]}_2 & \text{{\small the first half-year value is the sum of the first two quarters' values},} \\
v^{[2]}_2 & = & v^{[1]}_3 + v^{[1]}_4 & \text{{\small the second half-year value is the sum of the last two quarters' values},} \\
\multicolumn{3}{l}{\text{and thus}}\\
v^{[4]}_1 & = & v^{[1]}_1 + v^{[1]}_2 + v^{[1]}_3 + v^{[1]}_4 & \text{{\small the annual value is the sum of the four quarterly values}.}
\end{array}
$$


All the relationships so far can be expressed in more compact form, using matrix notation.
Denote by
$\yvet_{\tau}^{[k]} = \begin{bmatrix} x_{\tau}^{[k]} \; w_{\tau}^{[k]} \; z_{\tau}^{[k]} \end{bmatrix}'$
the $(3 \times 1)$ vector of the observations for temporal granularity $k \in \mathcal{K}$ of the variables forming the cross-sectional hierarchy at time $\tau=1,\ldots,4/k$. The cross-sectional aggregation relationships can be described as follows:
\[
x_{\tau}^{[k]} = \Cvet \begin{bmatrix} w_{\tau}^{[k]} \\ z_{\tau}^{[k]} \end{bmatrix}, \quad
\yvet_{\tau}^{[k]} = \Svet \begin{bmatrix} w_{\tau}^{[k]} \\ z_{\tau}^{[k]} \end{bmatrix}, \quad
\Uvet' \yvet_{\tau}^{[k]} = 0, \qquad
\tau=1,\ldots,\displaystyle\frac{4}{k}, \quad
k \in \mathcal{K},
\]
where $\Cvet$ is the cross-sectional aggregation matrix, $\Svet$ is the cross-sectional summing matrix, and $\Uvet'$ is the zero-constraints matrix expressing the cross-sectional constraints in homogeneous form, respectively given by: 
\begin{small}
\[
\Cvet = \begin{bmatrix}
	1 & 1 \\
\end{bmatrix} ,\quad
\Svet = \begin{bmatrix}
	\Cvet \\ \Ivet_2
\end{bmatrix} =
\begin{bmatrix}
1 & 1 \\
1 & 0 \\
0 & 1
\end{bmatrix} , \quad
\Uvet' = \begin{bmatrix}
	\Ivet_1 \; \; -\Cvet
\end{bmatrix} = 
\begin{bmatrix}
1 & -1 & -1 \\
\end{bmatrix} ,
\]
\end{small}

\noindent with $\Ivet_{l}$ denoting the identity matrix of order $l$.
The complete temporal aggregation relationships linking the values of a single variable (say $V$) at different time granularities, may in turn be expressed through matrices
\begin{small}
\[
\Kvet = \begin{bmatrix}
	1 & 1 & 1 & 1 \\
	1 & 1 & 0 & 0 \\
	0 & 0 & 1 & 1
\end{bmatrix} ,\quad
\Rvet = \begin{bmatrix}
	\Kvet \\ \Ivet_4
\end{bmatrix} =
\begin{bmatrix}
	1 & 1 & 1 & 1 \\
	1 & 1 & 0 & 0 \\
	0 & 0 & 1 & 1 \\
	1 & 0 & 0 & 0 \\
	0 & 1 & 0 & 0 \\
	0 & 0 & 1 & 0 \\
	0 & 0 & 0 & 1
\end{bmatrix} , \quad
\Zvet' = \begin{bmatrix}
	\Ivet_3 \; \; -\Kvet
\end{bmatrix} = 
\begin{bmatrix}
	1 & 0 & 0 & -1 & -1 & -1 & -1 \\
	0 & 1 & 0 & -1 & -1 &  0 &  0 \\
	0 & 0 & 1 &  0 &  0 & -1 & -1
\end{bmatrix} ,
\]
\end{small}

\noindent where $\Kvet$ is the temporal aggregation matrix, $\Rvet$ is the temporal summing matrix, and $\Zvet'$ is the temporal zero-constraints matrix expressing the temporal constraints in homogeneous form. It follows that:
\begin{small}
$$
\begin{bmatrix}
	v^{[4]}_1  \\ v^{[2]}_1 \\ v^{[2]}_2
\end{bmatrix} =
\Kvet\begin{bmatrix}
	v^{[1]}_1 \\ v^{[1]}_2 \\ v^{[1]}_3 \\ v^{[1]}_4
\end{bmatrix} =
\Kvet \vvet^{[1]}, \quad
\vvet = \begin{bmatrix}
	v^{[4]}_1  \\ v^{[2]}_1 \\ v^{[2]}_2 \\
	v^{[1]}_1 \\ v^{[1]}_2 \\ v^{[1]}_3 \\ v^{[1]}_4
\end{bmatrix} =
\Rvet \vvet^{[1]}, \quad
\Zvet'\vvet=\Zerovet_{(3 \times 1)} .
$$
\end{small}

\noindent Reduced temporal hierarchies can be obtained by simply eliminating the appropriate rows\footnote{This option is available in the \texttt{R} package \texttt{FoReco} (\citealp{FoReco2022}).} from matrix $\Kvet$.

To simultaneously consider cross-sectional and temporal aggregation relationships, all the nodes of the complete cross-temporal hierarchy in figure \ref{toyquarterlyhierarchy} can be expressed in terms of the quarterly time series $w_{\tau}^{[1]}$ and $z_{\tau}^{[1]}$, $\tau=1,\ldots,4$, according to the \textit{structural} representation:
\begin{equation}
	\label{ctstructuralform}
	\yvet = \Fvet\bvet^{[1]},
\end{equation}
that is
\begin{small}
	\[
\underbrace{\begin{bmatrix}
		x_1^{[4]} \\ x_1^{[2]} \\ x_2^{[2]} \\
		x_1^{[1]} \\ x_2^{[1]} \\ x_3^{[1]} \\ x_4^{[1]} \\
		w_1^{[4]} \\ w_1^{[2]} \\ w_2^{[2]} \\
		w_1^{[1]} \\ w_2^{[1]} \\ w_3^{[1]} \\ w_4^{[1]} \\
		z_1^{[4]} \\ z_1^{[2]} \\ z_2^{[2]} \\
		z_1^{[1]} \\ z_2^{[1]} \\ z_3^{[1]} \\ z_4^{[1]} \\
	\end{bmatrix}}_{\yvet}
	=
\underbrace{\begin{bmatrix}
		1  &  1  &  1  &  1  &  1  &  1  &  1  &  1\\
		1  &  1  &  0  &  0  &  1  &  1  &  0  &  0\\
		0  &  0  &  1  &  1  &  0  &  0  &  1  &  1\\
		1  &  0  &  0  &  0  &  1  &  0  &  0  &  0\\
		0  &  1  &  0  &  0  &  0  &  1  &  0  &  0\\
		0  &  0  &  1  &  0  &  0  &  0  &  1  &  0\\
		0  &  0  &  0  &  1  &  0  &  0  &  0  &  1\\
		1  &  1  &  1  &  1  &  0  &  0  &  0  &  0\\
		1  &  1  &  0  &  0  &  0  &  0  &  0  &  0\\
		0  &  0  &  1  &  1  &  0  &  0  &  0  &  0\\
		1  &  0  &  0  &  0  &  0  &  0  &  0  &  0\\
		0  &  1  &  0  &  0  &  0  &  0  &  0  &  0\\
		0  &  0  &  1  &  0  &  0  &  0  &  0  &  0\\
		0  &  0  &  0  &  1  &  0  &  0  &  0  &  0\\
		0  &  0  &  0  &  0  &  1  &  1  &  1  &  1\\
		0  &  0  &  0  &  0  &  1  &  1  &  0  &  0\\
		0  &  0  &  0  &  0  &  0  &  0  &  1  &  1\\
		0  &  0  &  0  &  0  &  1  &  0  &  0  &  0\\
		0  &  0  &  0  &  0  &  0  &  1  &  0  &  0\\
		0  &  0  &  0  &  0  &  0  &  0  &  1  &  0\\
		0  &  0  &  0  &  0  &  0  &  0  &  0  &  1
	\end{bmatrix}}_{\Fvet}
	\underbrace{\begin{bmatrix}
		w_1^{[1]} \\ w_2^{[1]} \\ w_3^{[1]} \\ w_4^{[1]} \\
		z_1^{[1]} \\ z_2^{[1]} \\ z_3^{[1]} \\ z_4^{[1]} \\
	\end{bmatrix}}_{\bvet^{[1]}} ,
	\]
\end{small}

\noindent where $\yvet = \begin{bmatrix}	\xvet' \; \wvet' \; \zvet' \end{bmatrix}'$ is the vector containing the data for all variables at any temporal granularity,
$\bvet^{[1]} = \begin{bmatrix} \wvet^{[1]\prime} \; \zvet^{[1]\prime} \end{bmatrix}'$
is the vector of the high-frequency bottom time series, and $\Fvet$ is the cross-temporal summing matrix mapping $\bvet^{[1]}$ into $\yvet$.
Expression (\ref{ctstructuralform}) is the natural extension of the cross-sectional structural representation firstly shown by \cite{Athanasopoulos2009}. It relates the observations at the upper levels of both cross-sectional and temporal hierarchies, to the high-frequency bottom time series of the cross-sectional hierarchy, which are the `very' bottom time series in a cross-temporal hierarchy (\citealp{DiFonzoGiro2021}).
 
Besides the number of variables forming the cross sectional hierarchy ($n=3$ in the above example), two crucial aspects affecting the dimension of matrix $\Fvet$ are (i) the temporal frequency of the highest-frequency granularity ($k=1$), and (ii) the amount of temporal granularities taken into account in the temporal hierarchy. For example, if one is interested in coherently forecasting hourly time series within a day-cycle, the complete cross-temporal summing matrix $\Rvet$ defining all infra-day temporal granularities ($\mathcal{K}=\{24,12,8,6,4,3,2,1\}$) has dimension $(60 \times 24)$, and is equal to
$$
\Rvet = \begin{bmatrix}
	\textbf{1}_{24} & \Ivet_2 \otimes \textbf{1}_{12} & \Ivet_3 \otimes \textbf{1}_{8} & \Ivet_4 \otimes \textbf{1}_{6} & \Ivet_6 \otimes \textbf{1}_{4} & \Ivet_8 \otimes \textbf{1}_{3} & \Ivet_{12} \otimes \textbf{1}_{2} & \Ivet_{24}
\end{bmatrix}' .
$$

Matrix $\Fvet$ is thus a large and sparse matrix\footnote{Sparse matrices require less memory than dense matrices, and allow some computations to be more efficient
	(\citealp{Paige1982}, \citealp{Davis2006}, \citealp{Matrix2022}).}
 of dimension $(60n \times 24n_b)$, where $n$ is the total number of series, and $n_b$ is the number of bottom time series in the cross-sectional hierarchy, respectively (in the above example, $n=3$ and $n_b=2$). Just to give an idea, the total number of variables in the dataset analyzed in this paper (see section \ref{sec:repass}) is $n=324$, with $n_b=318$ bottom time series, thus matrix $\Fvet$ has dimension ($19{,}440 \times 7{,}632$). However, if the interest in forecasting at certain time granularities is low, this dimensonality issue may be mitigated by considering only part of the temporal granularities between the highest and lowest temporal frequencies\footnote{Possible losses in the forecasting accuracy of the reconciled forecasts according to reduced temporal hierarchies should however be evaluated. This issue is currently under study.}.
For example, if one considers only hourly and daily forecasts,  the reduced $\Fvet$ is a $(8{,}100 \times 7{,}632)$ matrix, with a decrease of about 58\% in the amount of matrix entries wrt its complete counterpart.

In this framework, by extending the seminal idea by \cite{Hyndman2011}, a forecast reconciliation problem arises when, for the nodes of a cross-temporal hierarchy, a set of base forecasts - however obtained, and usually not aggregate consistent either in space and/or in time - 
are wished to be revised to fulfill the coherency relationships in space and time valid for the target data.
The purpose is to improve the accuracy of the initial forecasts by combining forecasts at different aggregation levels in space and time, and by incorporating in the final forecasts the information given by cross-sectional and temporal constraints.

\subsection{Notation}
\label{subsec:notation}
Suppose we want to forecast a $n$-variate high-frequency hierarchical time series $\left\{\yvet_t^{[1]}\right\}_{t=1}^T$, with forecast horizon equal to the seasonal cycle $m$, (e.g., month per year, $m=12$, quarter per year, $m=4$, hour per day, $m=24$), or a multiple thereof. Given a factor $k$ of $m$, we may consider a number of temporally aggregated versions of each component  of $\yvet_t^{[1]}$, given by the non-overlapping sums of $k$ successive values, each having seasonal period equal to $M_k=m/k$. To avoid ragged-edge data, we assume that the total number of observations involved in the non-overlaping aggregation is a multiple of $m$, and define $N$ the number of the lowest-frequency series observations, i.e. $N = T/m$.
Let $\mathcal{K}$ be the set of $p$ factors of $m$, in descending order, $\mathcal{K}=\{k_{p}, k_{p-1}, \ldots, k_2, k_1\}$, where $k_p=m$ and $k_1=1$, and define $k^* = \displaystyle\sum_{j=2}^{p} k_j$.

Following \cite{DiFonzoGiro2021}, denote $\Yvet_{N+h} \equiv \Yvet$ the $\left[n \times (k^*+m)\right]$ matrix of the target forecasts for any temporal granularity, with low-frequency temporal horizon $h$, given by:
\[
\Yvet = \left[ 	\Yvet^{[m]} \; \Yvet^{[k_{p-1}]} \ldots  \Yvet^{[k_2]} \; \Yvet^{[1]} \right] =
\left[\begin{array}{c} \Avet \\ \Bvet \end{array}\right] =
\left[\begin{array}{ccccc}
	\Avet^{[m]} & \Avet^{[k_{p-1}]} & \ldots  & \Avet^{[k_2]} & \Avet^{[1]} \\
	\Bvet^{[m]} & \Bvet^{[k_{p-1}]} & \ldots  & \Bvet^{[k_2]} & \Bvet^{[1]}
\end{array}
\right] ,
\]
where $m$ is the highest available sampling frequency per seasonal cycle (i.e., max. order of temporal aggregation). 
Each matrix
$\Yvet^{[k]} = \left[\begin{array}{c} \Avet^{[k]} \\ \Bvet^{[k]} \end{array}\right]$, $k \in \mathcal{K}$, contains the order-$k$ temporal aggregates of the $n_a$ cross-sectional upper time series $(\Avet^{[k]})$, and of the $n_b$ cross-sectional bottom time series $(\Bvet^{[k]})$, respectively, with $n=n_a + n_b$.
Accordingly, we define the matrix of base forecasts $\widehat{\Yvet}$ as:
\[
\widehat{\Yvet} = \left[ \widehat{\Yvet}^{[m]} \; \widehat{\Yvet}^{[k_{p-1}]} \ldots  \widehat{\Yvet}^{[k_2]} \; \widehat{\Yvet}^{[1]} \right] =
\left[\begin{array}{ccccc}
	\widehat{\Avet}^{[m]} & \widehat{\Avet}^{[k_{p-1}]} & \ldots  & \widehat{\Avet}^{[k_2]} & \widehat{\Avet}^{[1]} \\
	\widehat{\Bvet}^{[m]} & \widehat{\Bvet}^{[k_{p-1}]} & \ldots  & \widehat{\Bvet}^{[k_2]} & \widehat{\Bvet}^{[1]}
\end{array}
\right] .
\]
While the target forecasts are expected to be aggregate-consistent both in time and space, the base forecasts are in general cross-sectionally and/or temporally incoherent, that is:
\[
\Uvet' \Yvet = \Zerovet_{[n_a \times (k^*+m)]} \quad \text{and} \quad
\Zvet' \Yvet' = \Zerovet_{[k^* \times n]} , \quad \text{while} \quad
\Uvet' \widehat{\Yvet} \ne \Zerovet_{[n_a \times (k^*+m)]} \quad \text{and/or} \quad
\Zvet' \widehat{\Yvet}' \ne \Zerovet_{[k^* \times n]} ,
\]
where $\Uvet' = \begin{bmatrix}	\Ivet_{n_a} \; -\Cvet \end{bmatrix}$ and
$\Zvet' = \begin{bmatrix}	\Ivet_{k^*} \; -\Kvet \end{bmatrix}$
are zero-constraints matrices associated to the cross-sectional and temporal constraints, respectively.
Working with a single dimension (either sectional, or temporal), we may consider the multivariate generalization of the structural representations for cross-sectional (\citealp{Athanasopoulos2009}), and temporal (\citealp{Athanasopoulos2017}), hierarchies, respectively:

\noindent \textit{Cross-sectional structural representation of $k^*+m$ hierarchical time series}
\[
\yvet_{\tau}^{[k]} = \Svet\bvet_{\tau}^{[k]}, \quad \tau=1,\ldots, \frac{m}{k}, \quad \rightarrow \quad \Yvet^{[k]} = \Svet\Bvet^{[k]} \quad \quad k \in \mathcal{K},
\]
that is, in compact form
\[
\Yvet = \Svet\Bvet \quad \rightarrow \quad
\yvet = \left(\Svet \otimes \Ivet_{k^* + m}\right)\bvet,
\]
where 
$\yvet = \text{vec}\left(\Yvet'\right)$, and $\bvet = \text{vec}\left(\Bvet'\right)$.

\noindent \textit{Temporal hierarchies structural representation for $n$ individual time series}

\[
\begin{array}{c}
	\avet_{i} = \Rvet\avet_i^{[1]}, \quad i=1, \ldots, n_a \\
	\bvet_{j} = \Rvet\bvet_j^{[1]}, \quad j=1, \ldots, n_b \\
\end{array} \quad 
\rightarrow \quad 
\begin{array}{c}
	\Avet = \Rvet\Avet^{[1]\prime}\\
	\Bvet = \Rvet\Bvet^{[1]\prime} \\
\end{array},
\]
that is
\[
\Yvet' = \Rvet \Yvet^{[1]\prime}
\quad \rightarrow\quad 
\yvet = \left(\Ivet_n \otimes \Rvet\right)\yvet^{[1]} ,
\]
where $\yvet^{[1]} = \text{vec}\left({\Yvet^{[1]}}'\right)$,
and $\Rvet$ is the temporal summing matrix (\citealp{DiFonzoGiro2021}).

\subsection{Cross-temporal bottom-up reconciliation}
\label{sec:ctbu}
Bottom-up is an old and classic approach in the forecast reconciliation literature (\citealp{Dunn1976}, \citealp{Dangerfield1992}).
This approach simply consists in obtaining the upper-level series' forecasts by summing-up the base forecasts of the bottom level series in the hierarchy.
While its representation is rather straightforward when only a single dimension (either space or time) is involved, it is instructive considering how `bottom-up' works with cross-temporal hierarchies.
The cross-temporal bottom-up (ct$(bu)$) reconciliation of $k^*+m$ hierarchical time series' base forecasts at different time granularities, 
may be represented as:
\begin{equation}
	\label{ct-bu_formula_pap}
	\text{vec}\left(\widetilde{\Yvet}_{\text{ct}(bu)}'\right) = \Fvet \text{vec}\left(\widehat{\Bvet}^{[1]\prime}\right)
	\quad \leftrightarrow  \quad
	\widetilde{\yvet}_{\text{ct}(bu)} = 
	\Fvet \widehat{\bvet}^{[1]} ,
\end{equation}
where $\Fvet = \Svet \otimes \Rvet$ is the cross-temporal summing matrix, with $\otimes$ denoting the Kronecker product,
and
$\widehat{\bvet}^{[1]} = \text{vec}\left(\widehat{\Bvet}^{[1]\prime}\right)$ is the vector containing the base forecasts of the high-frequency bottom time series. 
In \ref{sec:ctbu-appendix} it is shown that 
the cross-temporal bottom-up reconciliation can be tought of as a two-step sequential reconciliation approach, where either cross-sectional reconciliation of the high-frequency bottom time series base forecasts is followed by temporal reconciliation, or \textit{vice-versa}.
This observation opens the way to `partly bottom-up' cross-temporal reconciliation approaches, where forecasts of the $n$ time series for different time granularities, and aggregation coherent only along a single dimension, 
are subsequently cross-temporally reconciled \textit{via} simple bottom-up according to the other dimension.
We call these cross-temporal forecast reconciliation approaches either ct$(rec_{te},bu_{cs}$), or ct$(rec_{cs},bu_{te}$),
where `$rec_{te}$' and `$rec_{cs}$' denote a generic forecast reconciliation approach in time and in space, respectively.

\subsection{Regression-based cross-temporal reconciliation}
\label{sec:regbased}

Let us consider the multivariate regression model 
$$
\widehat{\Yvet} = \Yvet + \Evet ,
$$
where the involved matrices have each dimension $[n \times (k^\ast + m)]$ and contain, respectively, the base ($\widehat{\Yvet}$) and the target forecasts ($\Yvet$), and the coherency errors ($\Evet$) for the $n$ component variables of the hierarchical time series of interest. 
Consider now the vectorized version of the model, that is
\begin{equation}
\label{eq:vec1}
	\mathrm{vec}\left(\widehat{\Yvet}'\right) = \mathrm{vec}\left(\Yvet'\right) + \mathrm{vec}\left(\Evet'\right) \quad \Leftrightarrow \quad \widehat{\yvet} = \yvet + \etavet ,
\end{equation}
where $\etavet = \mathrm{vec}\left(\Evet'\right)$
is the cross-temporal reconciliation error with zero mean and p.d. covariance matrix $\Omegavet_{ct}$.
Assuming $\Omegavet_{ct}$ known,
the optimal combination reconciled forecasts $\widetilde{\yvet}_{\text{oct}} = \mathrm{vec}\left(\widetilde{\Yvet}_{\text{oct}}'\right)$ are found by linearly constrained minimization of the generalized least squares (GLS) objective function
\begin{equation}
	\label{glsobj}
	\left(\yvet-\widehat{\yvet}\right)'\Omegavet_{ct}^{-1}\left(\yvet-\widehat{\yvet}\right) \qquad \mathrm{s.t.} \; \Hvet'\yvet = \Zerovet ,
\end{equation}
where 
$$
\Hvet'\yvet = \left[\begin{array}{c}
	\Uvet^\ast \\
	\Ivet_n \otimes \Zvet
\end{array}\right]\yvet = \Zerovet_{[n(k^\ast + m)\times 1]}
$$
is a full row-rank cross-temporal zero-constraint matrix, with $\Uvet^\ast = \left[\Zerovet_{(n_a m \times nk^\ast)} \; \Ivet_m \otimes\Uvet'\right]\Pvet$, and $\Pvet$ is the $[n(k^\ast + m) \times n(k^\ast + m)]$ commutation matrix, such that $\Pvet \mathrm{vec}(\Yvet) = \yvet$ (\citealp{DiFonzoGiro2021}).

The cross-temporally reconciled forecasts according to the projection approach solution (\citealp{Byron1978}; see also \citealp{VanErven2015}, \citealp{Wickramasuriya2019}, \citealp{Panagiotelis2021}, \citealp{DiFonzoGiro2021}) are given by 
\begin{equation}
\label{octctproj}
\widetilde{\yvet}_{\text{oct}} = \left[\Ivet_{n(k^\ast + m)} - \Omegavet_{ct}\Hvet\left(\Hvet'\Omegavet_{ct}\Hvet\right)^{-1}\Hvet'\right]\widehat{\yvet} = \Mvet_{ct} \widehat{\yvet}.
\end{equation}
Alternatively, they may be obtained according to the structural approach developed by \cite{Hyndman2011} for the cross-sectional framework:
$$
\widehat{\yvet} = \Fvet \betavet + \etavet ,
$$
where 
$\betavet = E[\bvet^{[1]}|\mathcal{I}_{T}]$ is the mean of the high-frequency bottom-level values conditional to $\mathcal{I}_{T}$, the available information up to time $T$, $\bvet^{[1]}=\mathrm{vec}\left(\Bvet^{[1]\prime}\right)$. 
The structural approach forecast reconciliation formula is
\begin{equation}
	\label{octctstruc}
\widetilde{\yvet}_{\text{oct}} = \Fvet \left(\Fvet'\Omegavet^{-1}_{ct}\Fvet\right)^{-1} \Fvet'\Omegavet^{-1}_{ct}\widehat{\yvet} = \Fvet \Gvet_{ct} \widehat{\yvet},
\end{equation}
with $\Gvet_{ct} = \left(\Fvet'\Omegavet^{-1}_{ct}\Fvet\right)^{-1} \Fvet'\Omegavet^{-1}_{ct}$.
\cite{DiFonzoGiro2021} considered
the following approximations for the cross-temporal covariance matrix (`$oct$' stands for `optimal cross-temporal'):
\begin{itemize}[nosep, leftmargin=!, labelwidth=\widthof{ oct$(bdsam)$ -}, align=right]
	\item[oct$(ols)$ -] identity: $\Omegavet_{ct} = \Ivet_{n(k^*+m)}$
	\item[oct$(struc)$ -] structural: $\Omegavet_{ct} = \mathrm{diag}(\Fvet \mathbf{1}_{mn_b})$
	\item[oct$(wlsv)$ -] series variance scaling: $\Omegavet_{ct} = \reallywidehat{\Omegavet}_{ct,wlsv}$, that is a straightforward extension of the series variance scaling matrix presented by \cite{Athanasopoulos2017} in the temporal framework
	\item[oct$(bdshr)$ -] block-diagonal shrunk cross-covariance scaling: $\Omegavet_{ct} = \Pvet\reallywidehat{\Wvet}^{BD}_{ct,shr}\Pvet'$
	\item[oct$(bdsam)$ -] block-diagonal cross-covariance scaling: $\Omegavet_{ct} = \Pvet\reallywidehat{\Wvet}^{BD}_{ct,sam}\Pvet'$
	\item[oct$(shr)$ -] MinT-shr:   $\Omegavet_{ct} = \hat{\lambda}\widehat{\Omegavet}_{ct,D} + (1-\hat{\lambda})\widehat{\Omegavet}_{ct}$
	\item[oct$(sam)$ -] MinT-sam:  $\Omegavet_{ct} = \reallywidehat{\Omegavet}_{ct}$
\end{itemize}
where the symbol $\odot$ denotes the Hadamard product, $\hat{\lambda}$ is an estimated shrinkage coefficient (\citealp{Ledoit2004}), $\reallywidehat{\Omegavet}_{ct,D} = \Ivet_{n(k^\ast + m)} \odot \widehat{\Omegavet}_{ct}$, and $\reallywidehat{\Omegavet}_{ct}$ is the covariance matrix of the cross-temporal one-step ahead in-sample forecast errors.
The cross-sectional point forecast reconciliation formula
is obtained by assuming $m = 1$ (which implies $k^*=0$, and $\Omegavet_{ct} = \Wvet$ is an ($n \times n$) p.d. matrix):
$$
\widetilde{\yvet} = \left[\Ivet_{n} - \Wvet\Uvet\left(\Uvet'\Wvet\Uvet\right)^{-1}\Uvet'\right]\widehat{\yvet} = \Mvet_{cs} \widehat{\yvet},
$$
or, equivalently (\citealp{Hyndman2011}, \citealp{Hyndman2016}, \citealp{Wickramasuriya2019}),
$$
\widetilde{\yvet} = \Svet \left(\Svet'\Wvet^{-1}\Svet\right)^{-1} \Svet'\Wvet^{-1}\widehat{\yvet} = \Svet \Gvet_{cs} \widehat{\yvet},
$$
with $\Gvet_{cs} = \left(\Svet'\Wvet^{-1}\Svet\right)^{-1} \Svet'\Wvet^{-1}$.
The reconciled forecasts through temporal hierarchies for a single time series (\citealp{Athanasopoulos2017}) are in turn obtained by setting $n=1$ (i.e., $n_a=0$ and $n_b=1$, and $\Omegavet_{ct} = \Omegavet$ is a ($k^*+m \times k^*+m$) p.d. matrix):
$$
\widetilde{\yvet} = \left[\Ivet_{(k^\ast + m)} - \Omegavet\Zvet\left(\Zvet'\Omegavet\Zvet\right)^{-1}\Zvet'\right]\widehat{\yvet} = \Mvet_{te} \widehat{\yvet},
$$
or, equivalently,
$$
\widetilde{\yvet} = \Rvet \left(\Rvet'\Omegavet^{-1}\Rvet\right)^{-1} \Rvet'\Omegavet^{-1}\widehat{\yvet} = \Rvet \Gvet_{te} \widehat{\yvet},
$$
with $\Gvet_{te} = \left(\Rvet'\Omegavet^{-1}\Rvet\right)^{-1} \Rvet'\Omegavet^{-1}$.
Table \ref{tab:cs_te_app_cv} presents some approximations for the cross-sectional and the temporal covariance matrices. Other alternatives for temporal reconciliation, exploiting possible information in the residuals’ autocorrelation, can be found in \cite{Nystrup2020} and \cite{DiFonzoGiro2021}.
\begin{table}[htb]
  \caption{Approximations for the cross-sectional (\citealp{Hyndman2011}, \citealp{Hyndman2016}, \citealp{Wickramasuriya2019}, \citealp{DiFonzoGiro2021}) and temporal (\citealp{Athanasopoulos2017}, \citealp{DiFonzoGiro2021}) covariance matrix to be used in a reconciliation approach$^*$.}
  \label{tab:cs_te_app_cv}
\centering
  \begin{tabular}{>{\raggedleft\arraybackslash}m{0.15\linewidth}|>{\centering\arraybackslash}m{0.35\linewidth}|>{\centering\arraybackslash}m{0.35\linewidth}}
  \toprule
    & \textbf{Cross-sectional framework} & \textbf{Temporal framework} \\
    \midrule
    identity & cs$(ols)$: $\Wvet = \Ivet_n$ & te$(ols)$: $\Omegavet = \Ivet_{k^\ast + m}$\\[0.1cm]
    structural & cs$(struc)$: $\Wvet = \mathrm{diag}(\Svet \mathbf{1}_{nb})$ & te$(struc)$: $\Omegavet = \mathrm{diag}(\Rvet \mathbf{1}_{m})$\\[0.1cm]
    series variance & cs$(wls)$: $\Wvet = \widehat{\Wvet}_D = \Ivet_n \odot \widehat{\Wvet}$ & te$(wlsv)$: $\Omegavet = \widehat{\Omegavet}_{wlsv}$\\[0.1cm]
    MinT-shr & cs$(shr)$: $\Wvet = \hat{\lambda}\widehat{\Wvet}_D + (1-\hat{\lambda})\widehat{\Wvet}$ & te$(shr)$: $\Omegavet = \hat{\lambda}\widehat{\Omegavet}_D + (1-\hat{\lambda})\widehat{\Omegavet}$\\[0.1cm]
    MinT-sam & cs$(sam)$: $\Wvet = \widehat{\Wvet}$ & te$(sam)$: $\Omegavet = \widehat{\Omegavet}$ \\
    \bottomrule \addlinespace[0.1cm]
    \multicolumn{3}{p{0.9\linewidth}}{\footnotesize $^*$ $\widehat{\Wvet}$ ($\widehat{\Omegavet}$) is the covariance matrix of the cross-sectional (temporal) one-step ahead in-sample forecast errors, $\widehat{\Omegavet}_{wlsv}$ is a diagonal matrix ‘‘which contains estimates of the in-sample one-step-ahead error variances across each level’’ (\citealp{Athanasopoulos2017}, p. 64), and $\widehat{\Omegavet}_D = \Ivet_{k^\ast + m} \odot \widehat{\Omegavet}$.}
  \end{tabular}
\end{table}


\subsection{Heuristic and iterative cross-temporal reconciliation}
\label{sec:heuite}
\cite{Kourentzes2019} proposed an ensemble forecasting procedure (denoted KA), that exploits the simple averaging of different forecasts. It consists in the following steps (for further details, see \citealp{DiFonzoGiro2021}):
\begin{description}[nosep]
	\item[KA-Step 1] compute the temporally reconciled forecasts for each variable $i \in \{1, \ldots, n\}$, and arrange them in the $[n \times (k + m)]$ matrix $\widetilde{\Yvet}_{\text{te}}$;
	\item[KA-Step 2] starting from $\widetilde{\Yvet}_{\text{te}}$, compute the time-by-time cross-sectional reconciled forecasts for all the temporal aggregation levels ($\widetilde{\Yvet}_{\text{cs}}$), and collect all the $(n\times n)$ projection matrices used to reconcile forecasts of $k$-level temporally aggregated time series, $\Mvet_{cs}^{[k]}$, $k \in \mathcal{K}$;
	\item[KA-Step 3] transform the step 1 forecasts once more, by computing time-by-time cross-sectional reconciled forecasts for all temporal aggregation levels using the $(n \times n)$ matrix $\overline{\Mvet}$, given by the average of the matrices $\Mvet_{\text{cs}}^{[k]}$:
	$$
	\widetilde{\Yvet}_{\text{KA}} = \left(\frac{1}{p}\sum_{k \in \mathcal{K}} \Mvet_{\text{cs}}^{[k]}\right)\widetilde{\Yvet}_{\text{te}}= \overline{\Mvet} \widetilde{\Yvet}_{\text{te}} .
	$$ 
\end{description}

\cite{DiFonzoGiro2021} presented an iterative approach, that produces cross-temporally reconciled forecasts by alternating forecast reconciliation along one dimension (cross-sectional or temporal), based on the first two steps of the KA approach. The iteration $j\geq 1$ can be described as follows:
\begin{description}[nosep]
	\item[Step 1] compute the temporally reconciled forecasts ($\widetilde{\Yvet}_{\text{te}}^{(j)}$) for each variable $i \in \{1, \ldots, n\}$ of $\widetilde{\Yvet}_{\text{cs}}^{(j-1)}$;
	\item[Step 2] compute the time-by-time cross-sectional reconciled forecasts ($\widetilde{\Yvet}^{(j)}_{\text{cs}}$) for all the temporal aggregation levels of $\widetilde{\Yvet}_{\text{te}}^{(j)}$.
\end{description}
At $j=0$, the starting values are given by $\widetilde{\Yvet}_{\text{cs}}^{(0)}=\widehat{\Yvet}$, and the iterates end when the entries of matrix $\Dvet_{\text{te}}=\Zvet'\widetilde{\Yvet}^{(j)\prime}_{\text{cs}}$, containing all the temporal discrepancies, are small enough according to a suitable convergence criterion\footnote{Denoting $d_{\text{te}}$ a scalar measure of the gross temporal discrepancies contained in the matrix $\Dvet_{\text{te}}$, the convergence is achieved when $d_{\text{te}} < \delta$, where $\delta$ is a positive tolerance value. \cite{DiFonzoGiro2021} propose to use
$d_{\text{te}} = \left\Vert  \Dvet_{\text{te}}\right\Vert_{1}$, where $||\Xvet||_{1} = \sum_{i,l}|x_{i,l}|$. When the number of series and the time granularities increase, this criterion may be too demanding, 
resulting in more iterations than necessary to converge. An alternative choice, avalable in \texttt{FoReco} (\citealp{FoReco2022}), and used in this paper, is setting $d_{\text{te}} = \left\Vert  \Dvet_{\text{te}}\right\Vert_{\infty}$, where $||\Xvet||_{\infty} = \max |x_{i,l}|$.}.
The order of the dimensions to be reconciled in the two steps within each iteration is arbitrary (on this point, see \citealp{DiFonzoGiro2021}): in the description above, \textit{temporal-then-cross-sectional} reconciliation is iteratively performed (ite$(rec_{\text{te}},rec_{\text{cs}})$), otherwise the order may be reversed, getting \textit{cross-sectional-then-temporal} reconciliation  (ite$(rec_{\text{cs}},rec_{\text{te}})$).

\section{Cross-temporal coherency of sequential approaches and some remarkable equivalences}
\label{sec:techdigr}

In order to exploit both cross-sectional and temporal hierarchies, \cite{Yagli2019} consider two sequential reconciliation approaches: Spatial-then-Temporal-Reconciliation (STR), and Temporal-then-Spatial-Reconciliation (TSR).
In the former case, cross-sectional reconciliation of the base forecasts 
is performed first for any temporal granularity, followed by the temporal reconciliation of the individual series' forecasts. In the latter case, the order of application of the two reconciliation approaches is reversed, temporal reconciliation being performed first, followed by cross-sectional reconciliation.
In general, one would expect that the forecasts obtained this way are different (i.e.,
$\widetilde{\Yvet}_{TSR} \ne \widetilde{\Yvet}_{TSR}$),
and either cross-sectionally (TSR) or temporally (STR), but not cross-temporally, reconciled:
$$
\Uvet'\widetilde{\Yvet}_{TSR} = \Zerovet_{[n_a \times (k^*+m)]}, \quad
\Zvet'\widetilde{\Yvet}_{TSR} \ne \Zerovet_{(k^* \times n)}, \qquad
\Uvet'\widetilde{\Yvet}_{STR} \ne \Zerovet_{[n_a \times (k^*+m)]}, \quad
\Zvet'\widetilde{\Yvet}_{STR} = \Zerovet_{(k^* \times n)}.
$$

At this regard, we have found an interesting result, here shown as Theorem \ref{thm:ite}, according to which, if both covariance matrices used in either steps are constant across levels and time granularities, the final result (i) does not depend on the order of application of the uni-dimensional reconciliation phases, and (ii) is equivalent to that obtained through an optimal combination approach using a separable covariance matrix, i.e. with a Kronecker product structure \citep{Genton2007, Werner2008, Velu2017}.

\begin{theorem}\label{thm:ite}
	Let $\Wvet^{[k]}$, $k \in \mathcal{K}$, be the cross-sectional hierarchy error covariance matrix, and $\Omegavet_i$, $i = 1,\ldots,n$, the $i$-th series temporal hierarchy error covariance matrix. If $\Wvet^{[k]} = \Wvet$, $\forall k \in \mathcal{K}$, and $\Omegavet_i = \Omegavet$, $\forall i = 1,\ldots,n$, then:
	\begin{enumerate}[nosep]
		\item the iterative procedure reduces to a single (two-step) iteration to obtain cross-temporal reconciled forecasts. Furthermore,
		$$\widetilde{\Yvet}_{\text{tcs}} = \widetilde{\Yvet}_{\text{cst}} = \widetilde{\Yvet}_{\text{seq}},$$
		where $\widetilde{\Yvet}_{\text{tcs}}$ ($\widetilde{\Yvet}_{\text{cst}}$) is the [$n\times (k^\ast+m)$] matrix of the \textit{temporal-then-cross-sectional} (\textit{cross-sectional-then-temporal}) reconciled forecasts;
		\item Denoting $\widetilde{\Yvet}_{\text{oct}}$  the optimal (in least squares sense) combination cross-temporal reconciliation approach with cross-temporal covariance matrix  
		$\Omegavet_{\text{ct}} = \Wvet \otimes \Omegavet$, it is:
		$$
		\widetilde{\Yvet}_{\text{seq}} = \widetilde{\Yvet}_{\text{oct}}.
		$$
	\end{enumerate}
\end{theorem}
\begin{proof}
	\ref{sec:proof}.
\end{proof}

It is worth noting that, when only constant error covariance matrices are involved in both unidimensional steps, the iterative cross-temporal reconciliation approach proposed by \cite{DiFonzoGiro2021} reduces to the sequential procedure proposed by \cite{Yagli2019}, and gives a unique result, $\widetilde{\Yvet}_{\text{seq}}$, which does not depend on the order of application of the reconciliation phases\footnote{It is therefore surprising that in \cite{Yagli2019}, where only constant matrices are used in the reconciliation approaches, the results for the STR procedures are different from those for the TSR counterparts. We guess this depends on the way the unidimensional reconciled forecasts are computed, as we deduced from the \texttt{R} scripts made available by \cite{Yang2017te}.}.
In addition,
it does exist a simple simultaneous, optimal (in least-squares sense) reconciliation approach equivalent to the 
above sequential forecast reconciliation procedure. 
Finally, it should be noted that \textit{under the constant covariance matrices assumption} of the theorem, the heuristic cross-temporal reconciliation approach by \cite{Kourentzes2019} is equivalent to a sequential approach, since in this case the final averaging phase is not needed (more details can be found in \ref{sec:proof}, after the proof of the theorem).
All these results can be summarized as follows:
\[
\begin{array}{l}
\Wvet^{[k]} = \Wvet, \forall k \in \mathcal{K} \\
\Omegavet_i = \Omegavet, i=1,\ldots,n\\
\Omegavet_{ct} = \Wvet \otimes \Omegavet
\end{array} \quad \rightarrow \quad
\underbrace{\widetilde{\Yvet}_{\text{KA}(tcs)} = \widetilde{\Yvet}_{\text{KA}(cst)}}_{\parbox[t]{3cm}{\scriptsize\centering\cite{Kourentzes2019}}} =
\underbrace{\widetilde{\Yvet}_{\text{seq}}}_{\parbox[t]{1.75cm}{\scriptsize\centering\cite{Yagli2019}}} =
\underbrace{\widetilde{\Yvet}_{\text{ite}(tcs)} = \widetilde{\Yvet}_{\text{ite}(cst)} = \widetilde{\Yvet}_{\text{oct}}}_{\parbox[t]{3cm}{\scriptsize\centering\cite{DiFonzoGiro2021}}} 
\]

\noindent Using the covariance matrices in Table \ref{tab:cs_te_app_cv}, we may thus establish some useful equivalences between the cross-temporal reconciliation approaches considered in \cite{DiFonzoGiro2021}, and identify new ones:
\begin{enumerate}[nosep]
	\item oct$(ols)$: $\Omegavet_{\text{ct}} = \Ivet_{n(k^\ast+m)}$
	
	 equivalent to seq/KA/ite$(ols_{te},ols_{cs})$, and seq/KA/ite$(ols_{cs},ols_{te})$,
	 with $\Wvet = \Ivet_n$, $\Omegavet = \Ivet_{k^\ast+m}$.

	\item oct$(struc)$: $\Omegavet_{\text{ct}} = \mathrm{diag}\left(\Fvet\mathbf{1}_{mn_b}\right)$
	
	equivalent to
	seq/KA/ite$(struc_{te},struc_{cs})$, and seq/KA/ite$(struc_{cs},struc_{te})$,
	 with $\Wvet = \mathrm{diag}\left(\Svet\mathbf{1}_{n_b}\right)$, $\Omegavet = \mathrm{diag}\left(\Rvet\mathbf{1}_{m}\right)$.
	 
	\item oct$(ols_{cs},struc_{te})$: $\Omegavet_{ct} = \mathrm{diag}\left[(\Ivet_n \otimes \Rvet)\mathbf{1}_{mn_b}\right]$
	
	equivalent to
	seq/KA/ite$(ols_{cs},struc_{te})$,
	with $\Wvet = \Ivet_n$, $\Omegavet = \mathrm{diag}\left(\Rvet\mathbf{1}_{m}\right)$.
	
	\item oct$(struc_{cs},ols_{te})$: $\Omegavet_{ct} = \mathrm{diag}\left[(\Svet \otimes \Ivet_{k^\ast+m})\mathbf{1}_{mn_b}\right]$
	
	equivalent to
	seq/KA/ite$(struc_{cs},ols_{te})$,
	with $\Wvet = \mathrm{diag}\left(\Svet\mathbf{1}_{n_b}\right)$, $\Omegavet = \Ivet_{k^\ast+m}$.
\end{enumerate}

\noindent All the procedures considered so far make use of very simple covariance matrices, not using the information coming from the in-sample forecast errors. At this regard, another interesting result is that the
iterative reconciliation procedure where both cross-sectional and temporal reconciliation are performed using diagonal covariance matrices, computed using the one-step-ahead in-sample forecast errors (i.e., cs$(wls)$ and te$(wlsv)$), `converges' to the optimal combination reconciliation approach oct$(wlsv)$ (\citealp{DiFonzoGiro2021}), irrespective of the order of application of the unidimensional reconciliation steps\footnote{For the time being, we only have found strong empirical evidence of the correctness of this result, as shown in figure \ref{fig:conv_plot}, and are currently involved in finding a general proof.}. In other terms:
\[
\widetilde{\Yvet}_{\mathrm{ite}(wlsv_{te},wls_{cs})} \simeq \widetilde{\Yvet}_{\mathrm{ite}(wls_{cs},wlsv_{te})} \simeq \widetilde{\Yvet}_{\mathrm{oct}(wlsv)} ,
\]
where the quality of the approximation solely depends on the convergence criterion: the lower the tolerance value $\delta$ (see section \ref{sec:heuite}), the better the approximation will be.
An empirical support to this result is given by figure \ref{fig:conv_plot}, showing the Frobenius norm\footnote{The Frobenius norm of a real-valued matrix $\Xvet$ is defined as the square root of the sum of the squares of its elements (\citealp{Golub2013}, p. 71): $||\Xvet||_{F}=\sqrt{\displaystyle\sum_{i,j}x_{i,j}^2}$.}
 of the difference between the matrices of the reconciled forecasts using ite$(wlsv_{te},wls_{cs})$ with decreasing tolerance values $\delta$, and oct$(wlsv)$, in the 350 replications of the forecasting experiment described in section \ref{sec:repass}.
 
 In summary, when the considered diagonal covariance matrices are used in cross-sectional and temporal reconciliation phases, the iterative cross-temporal reconciliation approach is equivalent to a specific optimal combination cross-temporal approach. We argue that in practical application it is convenient to adopt the optimal combination version of a cross-temporal procedure, rather than its iterative counterpart, because the numerical results have better accuracy, and the computing time is often lower, mostly if the programming codes exploit sparse matrices tools (\citealp{Davis2006}).
 

\begin{figure}[t]
	\centering
	\includegraphics[width=\linewidth]{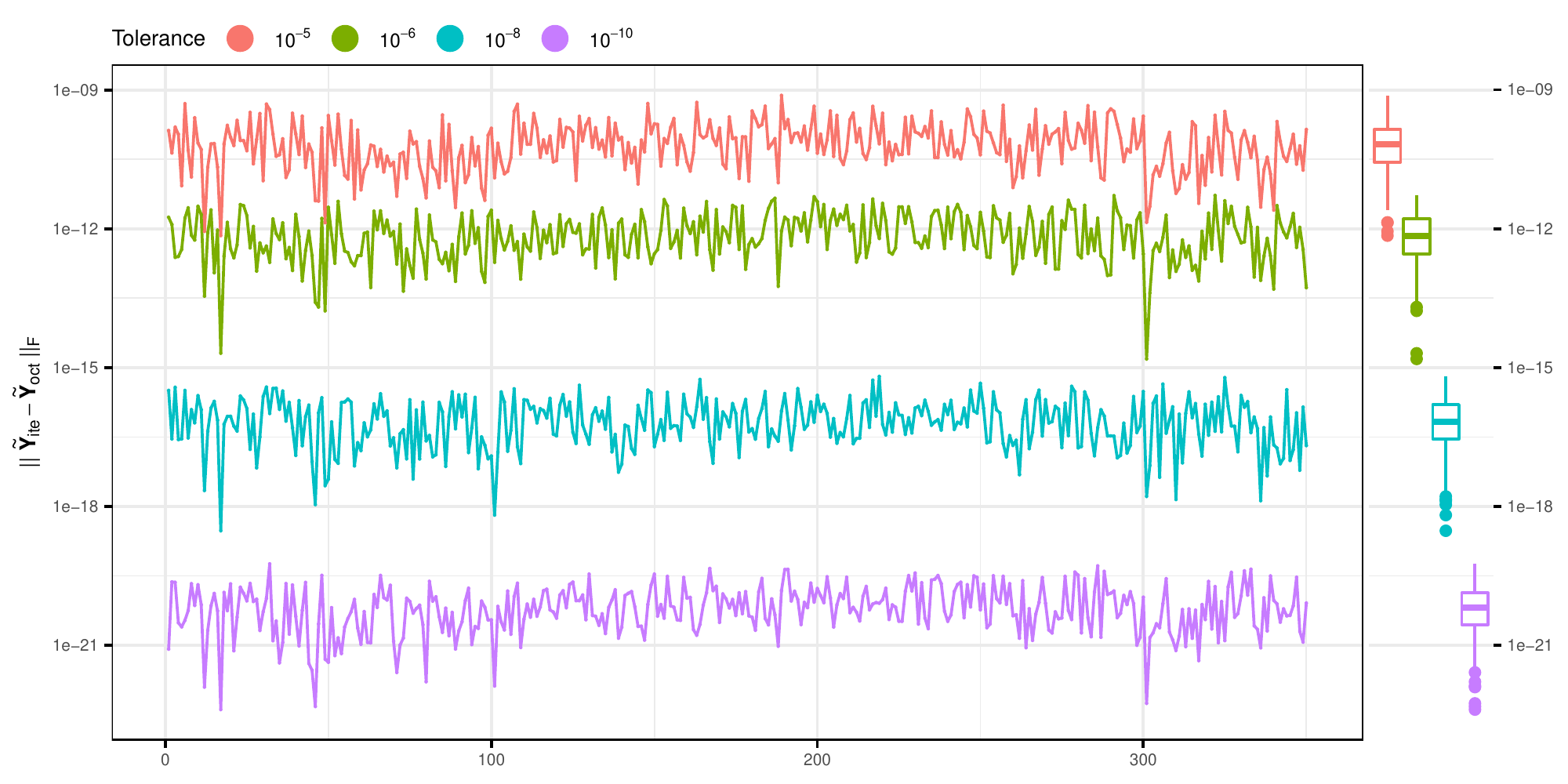}
	\caption{Frobenius norm of the difference between the matrices of the reconciled forecasts using ite$(wlsv_{te},wls_{cs})$ and oct$(wlsv)$, 
	with different tolerance value $\delta$. 350 replications of the forecasting experiment described in section \ref{sec:repass}.}
	\label{fig:conv_plot}
\end{figure}


\section{Replication and assessment of the forecasting experiment of Yagli et al. (2019)}
\label{sec:repass}

The dataset used in this study, called PV324, is the same used 
by \cite{Yang2017cs,Yang2017te}, and \cite{Yagli2019}. It refers to 318 simulated PV plants in California, whose hourly irradiation data are organized in three levels (figure \ref{Fig:PV324hierarchy}):
\begin{itemize}[nosep]
	\item $\mathcal{L}_0$: 1 time series for the Independent System Operator (ISO), given by the sum of the 318 plant series;
	\item $\mathcal{L}_1$: 5 time series for the Transmission Zones (TZ), each given by the sum of 27, 73, 101, 86, and 31 plants, respectively;
	\item $\mathcal{L}_2$: 318 bottom time series at plant level (P).
\end{itemize}

\begin{figure}[H]
	\centering
	\resizebox{\linewidth}{!}{
		\begin{tikzpicture} [baseline=(current  bounding  box.center),
			every node/.append style={shape=circle,
				draw=black,
				minimum size=1cm}
			]
			\node at (0, 0) (A1){P$_{1}$};
			\node at (1.2, 0) (A2){$...$};
			\node at (2.4, 0) (A3){P$_{27}$};
			\node at (1.2, 1.5) (A){TZ$_1$};

			\relation{0.2}{A1}{A};
			\relation{0.2}{A2}{A};
			\relation{0.2}{A3}{A};
			
			\node at (3.6, 0) (B1){P$_{28}$};
			\node at (4.8, 0) (B2){$...$};
			\node at (6, 0) (B3){P$_{100}$};
			\node at (4.8, 1.5) (B){TZ$_2$};

			\relation{0.2}{B1}{B};
			\relation{0.2}{B2}{B};
			\relation{0.2}{B3}{B};
			
			\node at (7.2, 0) (C1){P$_{101}$};
			\node at (8.4, 0) (C2){$...$};
			\node at (9.6, 0) (C3){P$_{201}$};
			\node at (8.4, 1.5) (C){TZ$_3$};

			\relation{0.2}{C1}{C};
			\relation{0.2}{C2}{C};
			\relation{0.2}{C3}{C};
			
			\node at (10.8, 0) (D1){P$_{202}$};
			\node at (12, 0) (D2){$...$};
			\node at (13.2, 0) (D3){P$_{287}$};
			\node at (12, 1.5) (D){TZ$_4$};

			\relation{0.2}{D1}{D};
			\relation{0.2}{D2}{D};
			\relation{0.2}{D3}{D};
			
			\node at (14.4, 0) (E1){P$_{288}$};
			\node at (15.6, 0) (E2){$...$};
			\node at (16.8, 0) (E3){P$_{318}$};
			\node at (15.6, 1.5) (E){TZ$_5$};

			\relation{0.2}{E1}{E};
			\relation{0.2}{E2}{E};
			\relation{0.2}{E3}{E};
			
			\node at (8.4, 3) (T){ISO};
			\relation{0.2}{A}{T};
			\relation{0.2}{B}{T};
			\relation{0.2}{C}{T};
			\relation{0.2}{D}{T};
			\relation{0.2}{E}{T};
			
			\node[draw = none] at (18, 3) (L0){$\mathcal{L}_0$};
			\node[draw = none, color =red] at (18, 1.5) (L1){$\mathcal{L}_1$};
			\node[draw = none, color = blue] at (18, 0) (L2){$\mathcal{L}_2$};
			
			\node[shape=rectangle, opacity=0.1, rounded corners, inner sep=6pt, fill = blue, fit=(A1.south west)(L2.north east)](Bcs){};
			
			
			
			\node[shape=rectangle, opacity=0.1, rounded corners, inner sep=6pt, fill = red, fit=(A.south west)(L1.north east)](Bcs){};
			\node[shape=rectangle, opacity=0.1, rounded corners, inner sep=6pt, fill = gray, fit=(T.south west)(L0.north east)](Bcs){};
		
			\draw [decorate,decoration={brace,amplitude=10pt,mirror, raise = 0.5em}] (A1.south west) -- (A3.south east) node[draw=none, midway,yshift=-2.5em]{27 ts};
			\draw [decorate,decoration={brace,amplitude=10pt,mirror, raise = 0.5em}] (B1.south west) -- (B3.south east) node[draw=none, midway,yshift=-2.5em]{73 ts};
			\draw [decorate,decoration={brace,amplitude=10pt,mirror, raise = 0.5em}] (C1.south west) -- (C3.south east) node[draw=none, midway,yshift=-2.5em]{101 ts};
			\draw [decorate,decoration={brace,amplitude=10pt,mirror, raise = 0.5em}] (D1.south west) -- (D3.south east) node[draw=none, midway,yshift=-2.5em]{86 ts};
			\draw [decorate,decoration={brace,amplitude=10pt,mirror, raise = 0.5em}] (E1.south west) -- (E3.south east) node[draw=none, midway,yshift=-2.5em]{31 ts};
		\end{tikzpicture}}
		\vspace{-0.6cm}
	\caption{PV324 hierarchy}
	\label{Fig:PV324hierarchy}
\end{figure}
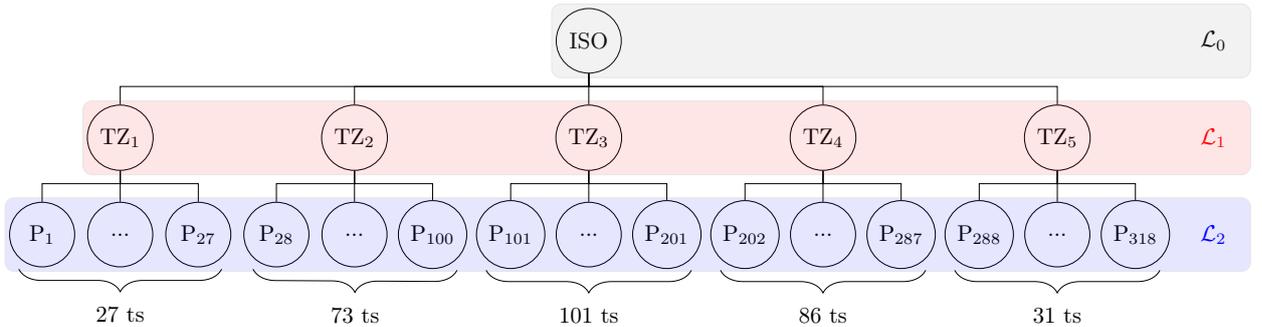

\noindent Following \cite{Yang2017te} and \cite{Yagli2019}, we perform a forecasting experiment with fixed length window of 14 days (i.e., 336 hours), forecast horizon of two days, and forecasting evaluation taking into account only the day-2 forecasts. These settings are coherent with the forecast operational submission requirements of CAISO, the public corporation managing power grid operations in California (\citealp{Makarov2011}, \citealp{Kleissl2013}).
For the 318 hourly time series, numerical weather prediction (NWP) forecasts generated by 3TIER (\citealp{3TIER2010}) are used as base forecasts. All the remaining base forecasts, for the six $\mathcal{L}_0$ and $\mathcal{L}_1$ time series at any time granularity $k \in \mathcal{K}$, and for the $\mathcal{L}_2$ 2-3-4-6-8-12-24 hours time series, are computed using the automatic ETS forecasting procedure of the \texttt{R}-package \texttt{forecast} (\citealp{forecast2021}), not controlling for possible negative forecasts.
Furthermore, day-ahead persistence is used as the reference model (PERS):
$$
\widehat{y}_{T+h|T,\; PERS}^{[1]} = y_{T+h-48}^{[1]}.
$$
Given the lead time of 48h, the day-ahead persistence takes the measurements made at day -2 as the forecasts for the operating day (\citealp{Yagli2019}, p. 394).
Benchmark forecasts at any level of the spatial hierarchy and for any temporal granularity are obtained through cross-temporal bottom-up of the 318 hourly bottom time series, that is:
\[
\widetilde{\textbf{y}}_{PERS_{bu}} = {\textbf{F}} \widehat{\textbf{b}}_{PERS}^{[1]} = \text{vec}\left(\widetilde{\textbf{Y}}_{PERS_{bu}}'\right).
\]
It is worth noting that the benchmark forecasts are always non-negative, and both spatially and temporally coherent.
Furthermore, these important properties are valid also for the forecasts obtained by cross-temporal bottom-up reconciliation of the 318 hourly bottom time series' NWP forecasts 3TIER:
\[
\widetilde{\textbf{y}}_{3TIER_{bu}} = {\textbf{F }} \widehat{\textbf{b}}_{3TIER}^{[1]} = \text{vec}\left(\widetilde{\textbf{Y}}_{3TIER_{bu}}'\right).
\]

In light of the results provided in section \ref{sec:techdigr}, the eight STR and TSR sequential reconciliations proposed by \cite{Yagli2019} reduce to the following four approaches: oct($ols$), oct($ols_{cs}, struc_{te}$), oct($struc_{cs}, ols_{te}$), oct($struc$).
In addition, \cite{Yagli2019} consider other four Temporal-then-Spatial-Reconciliation approaches, called TSR$_{\mathcal{L}_2}$, where the temporal reconciliation is applied only to the 318 plant level series' base forecasts. In this case, although constant matrices are used in either reconciliation steps, theorem \ref{thm:ite} no longer holds, and the obtained forecasts are temporally incoherent, as we show in the following.
In order to distinguish these approaches from the conventional sequential techniques, we call them
te$_{ols_2}$+cs$_{ols}$, 
te$_{struc_2}$+cs$_{ols}$,
te$_{ols_2}$+cs$_{struc}$, 
te$_{struc_2}$+cs$_{ols}$, respectively.

\subsection{Non-negativity and aggregation consistency issues}
\label{nn_issues}
Unlike PERS$_{bu}$ and 3TIER$_{bu}$, the other approaches considered by \cite{Yagli2019} produce some negative reconciled forecasts. Looking at table \ref{table_nn_Yagli_original}, we note that negative hourly forecasts have been produced in all 350 replications of the forecasting experiment, and there are some cases (in a range within 4 and 11 replications out of 350) where negative daily forecasts are produced as well\footnote{The relatively low number of negative hourly base forecasts, is due to the fact that the hourly  base forecasts of the 318 bottom time series are NWP 3TIER forecasts, that are always non-negative. Negative values are instead present in the ETS base forecasts for the aggregated series. Details can be found in the on-line appendix.}. 
Furthermore, the base forecasts are generally cross-temporally incoherent (i.e., in space and/or in time), and the sequential TSR$_{\mathcal{L}_2}$ approaches proposed by \cite{Yagli2019} 
are temporally incoherent.
This can be visually appreciated from figure \ref{fig:discrepancy_Yagli_original_new}, showing the boxplots from the distribution of the gross cross-sectional 
and temporal discrepancies registered in the 350 replications of the forecasting experiment, computed as (\citealp{DiFonzoGiro2021}):
\[
\begin{array}{rl}
	\text{Cross-sectional gross discrepancy:} & d_{\text{cs}} = ||\Uvet'\widehat{\Yvet}||_{1}\\
	\text{Temporal gross discrepancy:} & d_{\text{te}} = ||\Zvet'\widehat{\Yvet}'||_{1}
\end{array} ,
\]
where $||\Xvet||_{1} = \sum_{i,j}|x_{i,j}|$. For truly cross-temporally reconciled forecasts, neither cross-sectional nor temporal discrepancies are present, that is $d_{cs}=d_{te}=0$.

\begin{table}[H]
\centering
	\caption{Summary informations on the negative forecasts produced by the procedures considered by Yagli et al. (2019) in the forecasting experiment. Replications with at least a negative forecast (\# rep), number of series (\# series) with at least a negative forecast in a single replication (min and max), and min and max negative values found in all replications (values). Hourly and daily forecasts, forecast horizon: operating day.}
	\label{table_nn_Yagli_original}
		\begin{footnotesize}
		\setlength{\tabcolsep}{4pt}

\begin{tabular}[t]{l|ccccc|ccccc}
\toprule
\multirow{2}{*}{Approach}  & \multirow{2}{*}{\# rep (350)} & \multicolumn{2}{c}{\# series (324)} & \multicolumn{2}{c|}{values} & \multirow{2}{*}{\# rep (350)} & \multicolumn{2}{c}{\# series (324)} & \multicolumn{2}{c}{values}\\
         &              & min & max                           & min     & max &              & min & max                           & min     & max \\
\midrule
& \multicolumn{5}{c|}{Hourly forecasts} & \multicolumn{5}{c}{Daily forecasts}\\
base$^+$         & 350          & 4   & 6                             & -15.617 & -0.000& 11 & 0 & 35 & -51.205 & -0.006\\
oct$(ols$)       & 350          & 109 & 324                           & -69.165 & -0.000& 11 & 0 & 60 & -29.615 & -0.001\\
oct$(ols_{cs},struc_{te})$ & 350          & 137 & 324                           & -20.739 & -0.000& 11 & 0 & 33 & -10.363 & -0.000\\
oct$(struc_{cs},ols_{te})$ & 350          & 91  & 324                           & -17.961 & -0.000& 6 & 0 & 28 & -8.915 & -0.008\\
oct$(struc)$     & 350          & 89  & 324                           & -14.292 & -0.000& 4 & 0 & 11 & -1.020 & -0.026\\
te$(ols_{2})$+cs$(ols)$$^*$     & 350  & 42  & 324                           & -12.993 & -0.000& 10 & 0 & 72 & -15.037 & -0.018\\
te$(struc_{2})$+cs$(ols)$$^*$   & 350  & 36  & 324                           & -12.994 & -0.000& 10 & 0 & 71 & -14.791 & -0.005\\
te$(ols_{2})$+cs$(struc)$$^*$   & 350  & 144 & 324                           & -7.526 & -0.000 & 10 & 0 & 43 & -12.969 & -0.005\\
te$(struc_{2})$+cs$(struc)$$^*$ & 350  & 124 & 324                           & -7.642 & -0.000 & 5 & 0 & 31 & -7.042 & -0.002\\
\bottomrule
\multicolumn{11}{l}{\scriptsize $^+$ The approach produces spatial and temporal incoherent forecasts.} \\
\multicolumn{11}{l}{\scriptsize $^*$ The approach produces temporal incoherent forecasts.}
\end{tabular}

		\end{footnotesize}
\end{table}

\vspace{-.6cm}

\begin{figure}[h]
	\centering
	\includegraphics[width=0.9\linewidth]{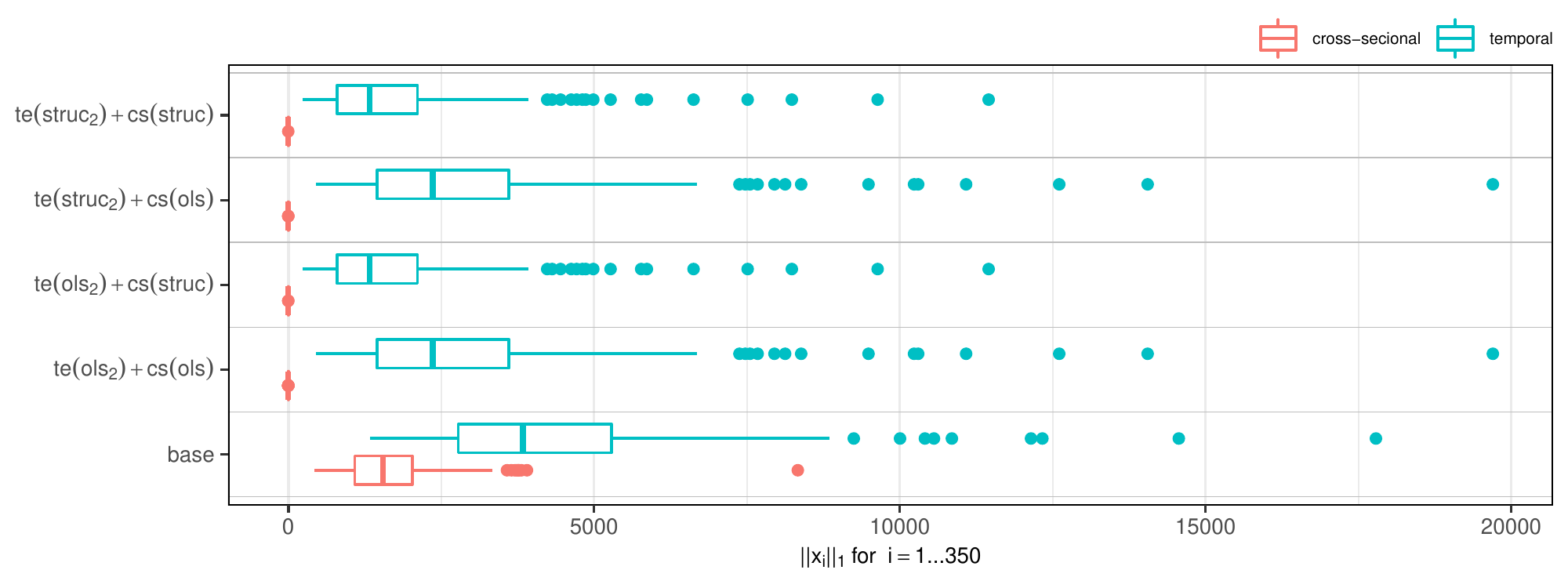}
	\caption{Base forecasts and TSR$_{\mathcal{L}_2}$ reconciled forecasts: cross-sectional (cs, in red) and temporal (te, in blue) discrepancies.}
	\label{fig:discrepancy_Yagli_original_new}
\end{figure}

In this case, forecast reconciliation may thus generate physically unreasonable values. Furthermore, if coherency is wished, the apparently innocuous practice of setting possible negative forecasts to zero is not advisable, since incoherence in sectional and/or time dimensions would be produced.
To overcome the above limitations, in section \ref{extanalysis} we propose a simple operational strategy, able to generate fully reconciled non-negative forecasts.


\subsection{Forecast evaluation}
Following \citealp{Yagli2019}, the accuracy of the considered approaches is measured in terms of normalised Root Mean Square Error (nRMSE), and Forecast Skill score:
\begin{equation}
	\label{nRMSE}
	\text{nRMSE}_{i,j}^{[k]} = \displaystyle
	\frac{\sqrt{\displaystyle\frac{1}{L}\displaystyle\sum_{l=1}^{L}\left(\widetilde{y}_{i,j,l}^{[k]} - y_{i,l}^{[k]}\right)^2}}{\displaystyle\frac{1}{L}\displaystyle\sum_{l=1}^{L}y_{i,l}^{[k]}}  , \qquad
	\text{Forecast Skill}_{i, j}^{[k]} = 1 - \frac{\text{nRMSE}_{i, j}^{[k]}}{\text{nRMSE}_{i, 0}^{[k]}},
\end{equation}
where $i=1,\ldots,n$, denotes the series, $k \in \mathcal{K}$, $j=0,\ldots,J$, denotes the forecasting approach ($j=0$ for the reference model PERS), and $L = nrep \cdot \displaystyle\frac{m}{k}$, where $nrep = 350$ is the number of the forecasting experiment replications.
Forecast Skill can be either negative (approach is worse than the reference model) or positive (approach is better than the reference model).
From table \ref{table_Yagli_original} it appears that on average all the considered forecasting approaches improve on the benchmark PERS$_{bu}$, in a range between 9.2\% (oct$(ols)$ for the 5 $\mathcal{L}_1$ level series' hourly forecasts) and 38.3\% (3TIER$_{bu}$ for the $\mathcal{L}_0$ total series' daily forecasts)\footnote{The results for all temporal aggregation orders, $k \in \{24,12,8,6,4,3,2,1\}$, are available in the on-line appendix.
}.
The forecasting accuracy indices for each Transmission Zone forecasts
are reported in table \ref{table_Yagli_original_comp_L2}.

\begin{table}[H]
\centering
	\caption{Forecast accuracies (nRMSE \%) and skills over the PERS$_{bu}$ benchmark of base forecasts and sequential reconciliation approaches. Free reconciliation procedures considered by \cite{Yagli2019}, Tables 2, 3, p. 395.
		Hourly (H) and Daily (D) forecasts, forecast horizon: operating day.
	Bold entries identify the best performing approach.}
	\label{table_Yagli_original}
		\begin{footnotesize}
		\setlength{\tabcolsep}{4pt}
		\begin{tabular}[t]{l|cc|cc|cc|cc|cc|cc}
\toprule
 &  \multicolumn{2}{c|}{$\mathcal{L}_0$} & \multicolumn{2}{c|}{$\mathcal{L}_1$} & \multicolumn{2}{c|}{$\mathcal{L}_2$}  &  \multicolumn{2}{c|}{$\mathcal{L}_0$} & \multicolumn{2}{c|}{$\mathcal{L}_1$} & \multicolumn{2}{c}{$\mathcal{L}_2$}\\
 \cmidrule(lr){2-7}  \cmidrule(lr){8-13}
 & H & D & H & D & H & D& H & D & H & D & H & D\\
\midrule
 & \multicolumn{6}{c|}{$nRMSE(\%)$} & \multicolumn{6}{c}{\textit{forecast skill}}\\
PERS$_{bu}$  & 34.62 & 20.23 & 43.15 & 24.57 & 59.75 & 30.65 & 0 & 0 & 0 & 0 & 0 & 0\\
3TIER$_{bu}$ & \textbf{26.03} & \textbf{12.48} & 33.95 & \textbf{16.75} & 53.46 & 25.19 & \textbf{0.248} & \textbf{0.383} & 0.213 & \textbf{0.318} & 0.105 & 0.178\\
base$^+$         & 27.85 & 18.17 & 34.24 & 20.94 & 53.46 & 25.82 & 0.196 & 0.102 & 0.206 & 0.148 & 0.105 & 0.158\\
oct($ols$)       & 30.69 & 17.91& 39.17 & 21.74 & 51.54 & 26.73& 0.113 & 0.115 & 0.092 & 0.115 & 0.137 & 0.128\\
oct$(ols_{cs},struc_{te})$ & 28.74 & 17.33 & 35.48 & 20.82 & 49.32 & 26.14& 0.170 & 0.143 & 0.178 & 0.152 & 0.175 & 0.147\\
oct$(struc_{cs},ols_{te})$ & 28.26 & 16.96 & 35.19 & 20.61 & 48.20 & 25.46& 0.184 & 0.162 & 0.184 & 0.161 & 0.193 & 0.169\\
oct$(struc)$     & 26.71 & 16.24 & 33.11 & 19.64 & 46.74 & \textbf{24.73} & 0.228 & 0.197 & 0.233 & 0.201 & 0.218 & \textbf{0.193}\\
te$(ols_2)$+cs$(ols)$$^*$  & 27.80 & 17.96 & 34.36 & 21.80 & 48.52 & 26.71& 0.197 & 0.112 & 0.204 & 0.113 & 0.188 & 0.129\\
te$(struc_2)$+cs$(ols)$$^*$  & 27.80 & 17.95 & 34.34 & 21.79 & 47.77 & 26.36& 0.197 & 0.113 & 0.204 & 0.113 & 0.201 & 0.140\\
te$(ols_2)$+cs$(struc)$$^*$ &  27.02 & 17.24 & 33.40 & 20.62 & 47.76 & 25.98& 0.219 & 0.148 & 0.226 & 0.161 & 0.201 & 0.152\\
te$(struc_2)$+ cs$(struc)$$^*$ &  26.17 & 16.68 & \textbf{32.44} & 19.96 & \textbf{46.25} & 25.03& 0.244 & 0.176 & \textbf{0.248} & 0.188 & \textbf{0.226} & 0.183\\
\bottomrule
\multicolumn{13}{l}{\scriptsize $^+$ The approach produces spatial and temporal incoherent forecasts.} \\
\multicolumn{13}{l}{\scriptsize $^*$ The approach produces temporal incoherent forecasts.}
\end{tabular}

		\end{footnotesize}
\end{table}


\begin{table}[H]
	\centering
	\caption{Forecast accuracies (nRMSE \%) and skills over the PERS$_{bu}$ benchmark of base forecasts and sequential reconciliation approaches for the series at $\mathcal{L}_1$ level. Free reconciliation procedures considered by \cite{Yagli2019}, Tables 2, 3, p. 395.
		Hourly (H) and Daily (D) forecasts, forecast horizon: operating day.
		Bold entries identify the best performing approach.}
	\label{table_Yagli_original_comp_L2}
	\begin{footnotesize}
		\setlength{\tabcolsep}{4pt}
		\begin{tabular}[t]{l|cc|cc|cc|cc|cc}
\toprule
 &  \multicolumn{2}{c|}{TZ$_1$} & \multicolumn{2}{c|}{TZ$_2$} & \multicolumn{2}{c|}{TZ$_3$}  &  \multicolumn{2}{c|}{TZ$_4$} & \multicolumn{2}{c}{TZ$_5$} \\
 \cmidrule(lr){2-7}  \cmidrule(lr){8-11}
 & H & D & H & D & H & D& H & D & H & D \\
\midrule
 & \multicolumn{10}{c}{$nRMSE(\%)$}\\
PERS$_{bu}$  			& 28.72 & 16.12 & 40.27 & 22.40 & 46.48 & 26.96 & 52.82 & 29.75 & 47.44 & 27.62 \\ 
3TIER$_{bu}$ 			& 22.81 & \textbf{10.43} & 33.34 & \textbf{16.72} & 34.01 & \textbf{17.06} & 46.40 & \textbf{22.80} & \textbf{33.16} & \textbf{16.75} \\ 
base$^+$         		& 22.41 & 13.55 & 32.05 & 18.70 & 35.14 & 21.95 & 44.94 & 26.89 & 36.67 & 23.58 \\ 
oct($ols$)       		& 29.27 & 14.89 & 34.58 & 19.25 & 39.55 & 22.91 & 47.30 & 26.50 & 45.16 & 25.15 \\ 
oct$(ols_{cs},struc_{te})$ 	& 24.26 & 13.93 & 32.32 & 18.47 & 36.69 & 21.78 & 43.89 & 25.30 & 38.81 & 23.59 \\
oct$(struc_{cs},ols_{te})$ 	& 23.44 & 13.50 & 32.89 & 18.83 & 37.42 & 22.33 & 44.76 & 25.97 & 38.90 & 23.49 \\ 
oct$(struc)$     		& 22.00 & 12.81 & 30.92 & 17.82 & 34.79 & 20.94 & 42.02 & 24.50 & 35.83 & 22.12 \\
te$(ols_2)$+cs$(ols)$$^*$  	& 22.95 & 15.46 & 32.06 & 19.23 & 35.12 & 22.38 & 44.93 & 26.91 & 36.73 & 25.02 \\ 
te$(struc_2)$+cs$(ols)$$^*$  	& 22.93 & 15.46 & 32.05 & 19.22 & 35.11 & 22.37 & 44.92 & 26.90 & 36.70 & 25.00 \\ 
te$(ols_2)$+cs$(struc)$$^*$ 	& 22.11 & 13.36 & 31.40 & 18.70 & 34.70 & 21.93 & 42.97 & 26.06 & 35.82 & 23.06 \\ 
te$(struc_2)$+ cs$(struc)$$^*$ 	& \textbf{21.58} & 13.00 & \textbf{30.48} & 18.09 & \textbf{33.58} & 21.18 & \textbf{41.86} & 25.21 & 34.69 & 22.32 \\ 
\midrule
 & \multicolumn{10}{c}{\textit{forecast skill}}\\
PERS$_{bu}$  			& 0 & 0 & 0 & 0 & 0 & 0 & 0 & 0 & 0 & 0 \\ 
3TIER$_{bu}$ 			& 0.206 & \textbf{0.353} & 0.172 & \textbf{0.254} & 0.268 & \textbf{0.367} & 0.121 & \textbf{0.234} & \textbf{0.301} & \textbf{0.394} \\ 
base$^+$         		& 0.220 & 0.159 & 0.204 & 0.165 & 0.244 & 0.186 & 0.149 & 0.096 & 0.227 & 0.146 \\ 
oct($ols$)       		& -0.019 & 0.076 & 0.141 & 0.141 & 0.149 & 0.150 & 0.104 & 0.109 & 0.048 & 0.089 \\ 
oct$(ols_{cs},struc_{te})$ 	& 0.155 & 0.136 & 0.197 & 0.176 & 0.211 & 0.192 & 0.169 & 0.150 & 0.182 & 0.146 \\ 
oct$(struc_{cs},ols_{te})$ 	& 0.184 & 0.163 & 0.183 & 0.159 & 0.195 & 0.172 & 0.153 & 0.127 & 0.180 & 0.149 \\ 
oct$(struc)$     		& 0.234 & 0.205 & 0.232 & 0.205 & 0.251 & 0.223 & 0.204 & 0.177 & 0.245 & 0.199 \\ 
te$(ols_2)$+cs$(ols)$$^*$  	& 0.201 & 0.041 & 0.204 & 0.141 & 0.244 & 0.170 & 0.149 & 0.095 & 0.226 & 0.094 \\ 
te$(struc_2)$+cs$(ols)$$^*$  	& 0.202 & 0.041 & 0.204 & 0.142 & 0.245 & 0.170 & 0.150 & 0.096 & 0.227 & 0.095 \\ 
te$(ols_2)$+cs$(struc)$$^*$ 	& 0.230 & 0.171 & 0.220 & 0.165 & 0.253 & 0.187 & 0.187 & 0.124 & 0.245 & 0.165 \\ 
te$(struc_2)$+ cs$(struc)$$^*$ 	& \textbf{0.249} & 0.194 & \textbf{0.243} & 0.192 & \textbf{0.277} & 0.215 & \textbf{0.208} & 0.153 & 0.269 & 0.192 \\ 
\bottomrule
\multicolumn{11}{l}{\scriptsize $^+$ The approach produces spatial and temporal incoherent forecasts.} \\
\multicolumn{11}{l}{\scriptsize $^*$ The approach produces temporal incoherent forecasts.}
\end{tabular}
	\end{footnotesize}
\end{table}

\noindent We observe that:
\begin{itemize}[nosep]
\item 3TIER$_{bu}$  gives the best forecasting performance for the total series ($\mathcal{L}_0$) at any temporal granularity, the second best being te$(struc_2)$+cs$(struc)$.
\item The approach te$(struc_2)$+cs$(struc)$ ranks first for 4 out of the 5 hourly upper series at level $\mathcal{L}_1$, whereas at daily level 3TIER$_{bu}$ `wins' again. The performance of oct$(struc)$ results very similar to that of te$(struc_2)$+cs$(struc)$.
\item te$(struc_2)$+cs$(struc)$ and oct$(struc)$ show the best performance for the 318 bottom time series at any temporal granularity, with a slight prevalence of oct$(struc)$ (for $ k \ge 3$), plus oct$(struc)$ forecasts are cross-temporally coherent.
\item Unlike te$(struc_2)$+ cs$(struc)$, 3TIER$_{bu}$ forecasts are always non-negative, and coherent both in space and time with the forecasts at any granularity.
\item For these reasons, 3TIER$_{bu}$ should be considered as a challenging competitor in the evaluation of the newly proposed procedures (section \ref{extanalysis}).
\end{itemize}

To give a complete picture of the evaluation results for hourly and daily forecasts, in figure \ref{fig:mcb_yang} the Multiple Comparison with the Best (MCB) Nemenyi tests are shown (\citealp{Koning2005}, \citealp{Kourentzes2019}, \citealp{Makridakis2022}). This allows to establish if the forecasting performances of the considered techniques are significantly different.
At daily level, oct$(struc)$ ranks first, and is significantly better than the other forecasting approaches, with  te$(struc_2)$+cs$(struc)$ at the second place. This result is reversed at hourly level: te$(struc_2)$+cs$(struc)$ ranks first and is significantly better than all the other approaches, with oct$(struc)$ at the second place. However, it should be recalled that the forecasts produced by te$(struc_2)$+cs$(struc)$ are not temporally coherent, which means that the sum of the hourly forecasts does not match with the corresponding daily forecasts (see figure \ref{fig:discrepancy_Yagli_original_new}).

\begin{figure}[H]
	\centering
	\includegraphics[width=0.9\linewidth]{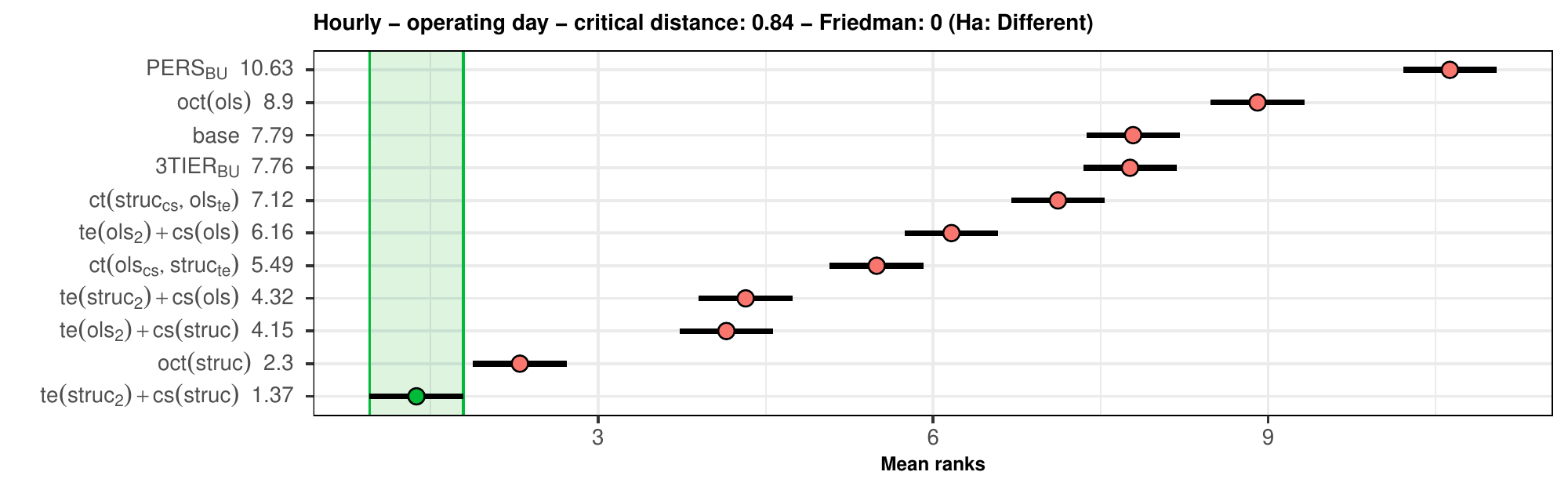}
	\includegraphics[width=0.9\linewidth]{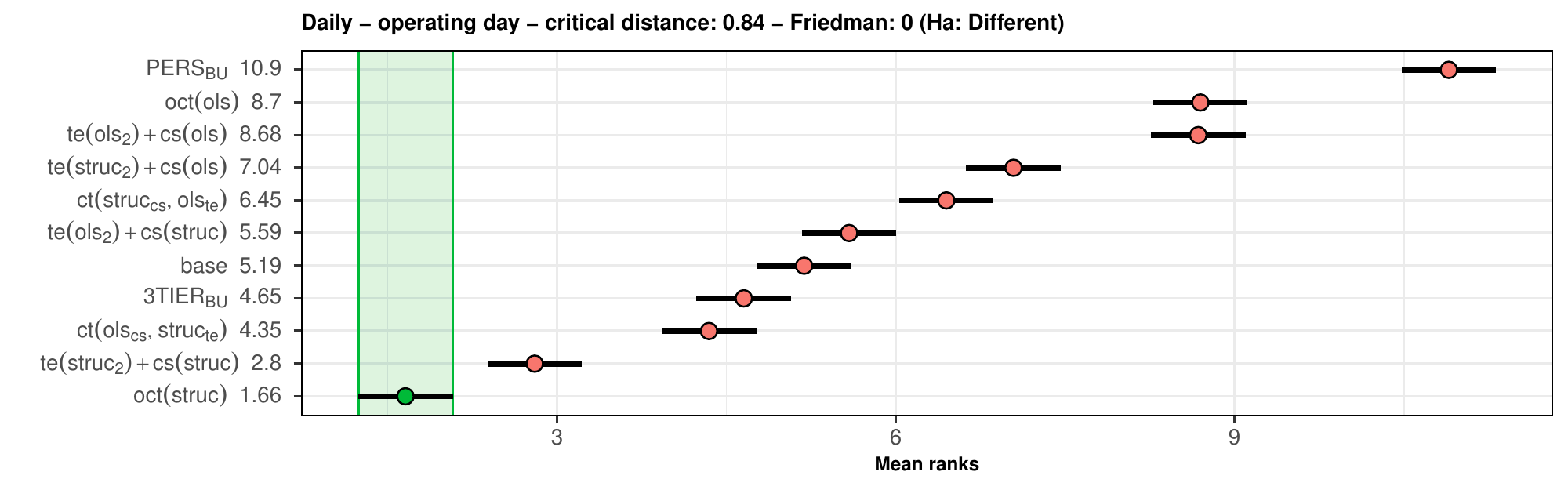}
	\caption{MCB Nemenyi test results: average ranks and 95\% confidence intervals. The free reconciliation approaches considered by \cite{Yagli2019} are sorted vertically according to the nRMSE\% mean rank. Hourly (top panel), and Daily (bottom panel) forecasts for $\mathcal{L}_0, \mathcal{L}_1, \mathcal{L}_2$ levels (324 series). Forecast horizon: operating day.
	The mean rank of each approach is displayed to the right of their names. If the intervals of two forecast reconciliation approaches do not overlap, this indicates a statistically different performance. Thus, approaches that do not overlap with the green interval are considered significantly worse than the best, and vice-versa.}
	\label{fig:mcb_yang}
\end{figure}

Limiting ourselves to consider fully coherent reconciled forecasts, the scatter plots of the 324 couples of nRMSE(\%) for oct$(struc)$ vs., respectively, PERS$_{bu}$ and 3TIER$_{bu}$ (figure \ref{fig:comparison_ct-struc_free_vs_PERS_3TIER}), show that the most performing regression-based cross-temporal reconciliation approach improves uniformly on the benchmark, and in the majority of cases on 3TIER$_{bu}$, particularly at hourly level ($k=1$), where the 89\% of the variables register an improvement in nRMSE. However, it is worth noting that at daily level, this value decreases to 49\%, which means that the NWP forecasts may still play a role at lower time granularity.

\begin{figure}[H]
	\centering
	\includegraphics[width=0.9\linewidth]{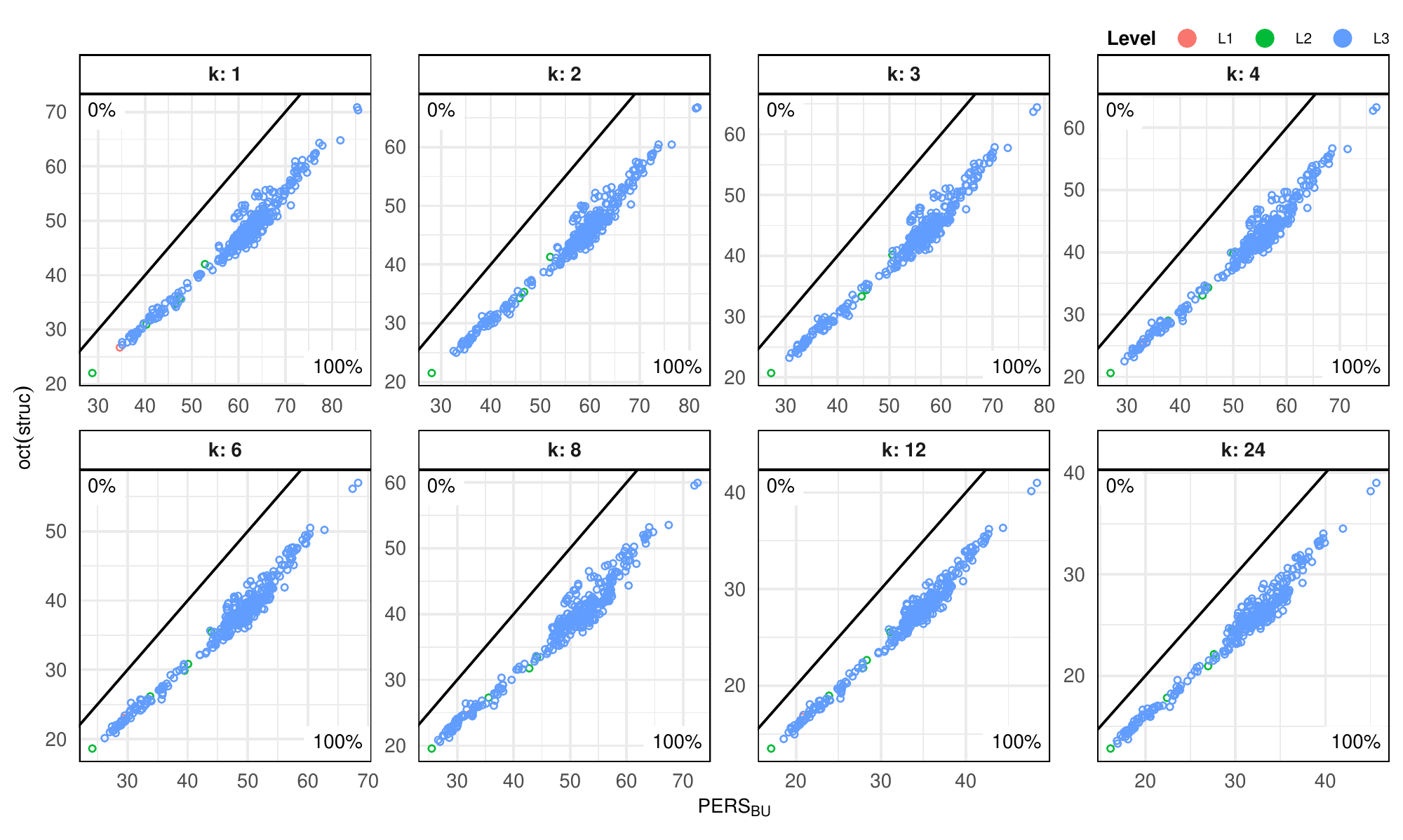}
	\includegraphics[width=0.9\linewidth]{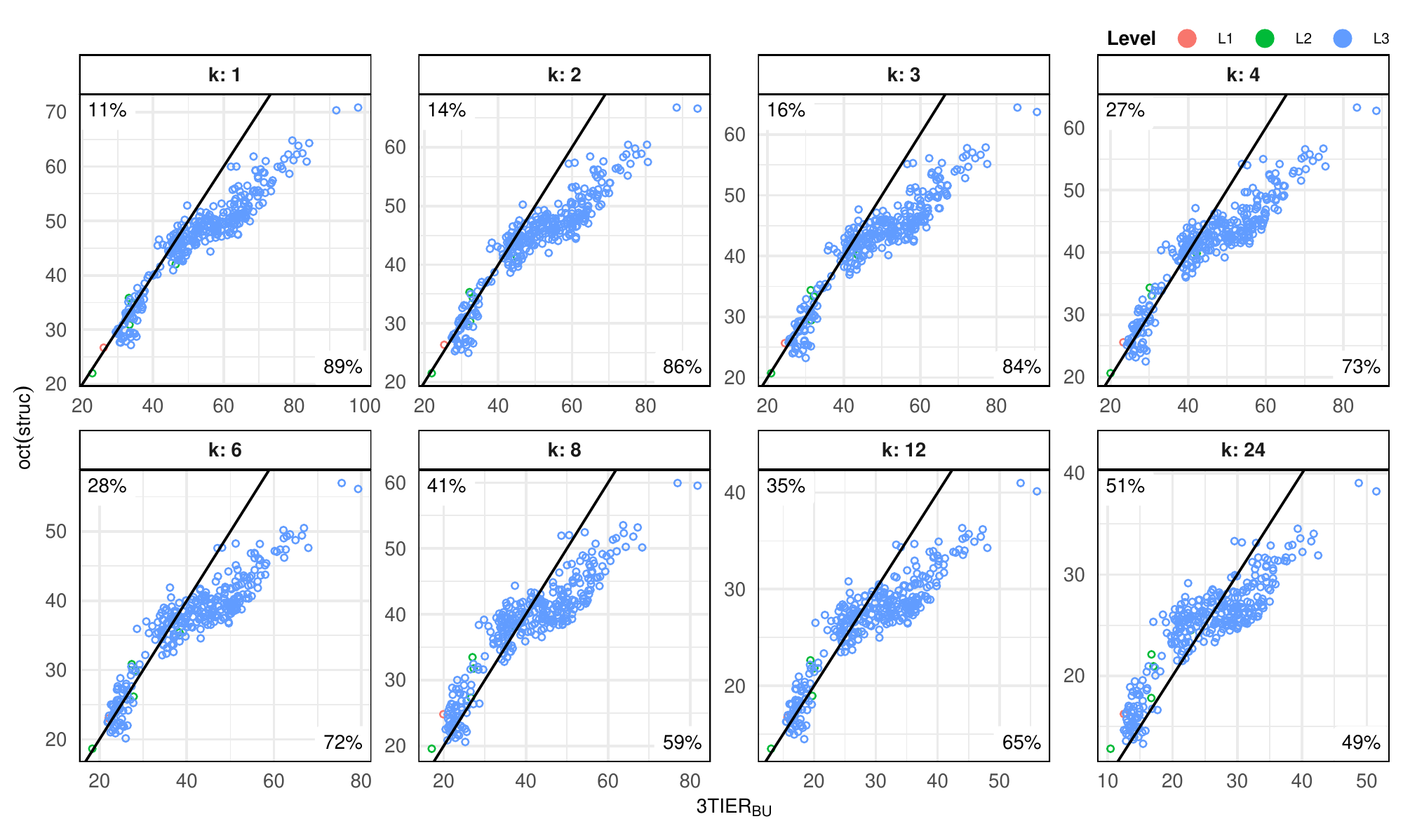}
	\caption{Comparison of nRMSE(\%) between PERS$_{bu}$ and ct$(struc)$ (top panel), and between 3TIER$_{bu}$ and ct$(struc)$ (bottom panel). The black line represents the bisector, where the nRMSE's for both approaches are equal. On the top-left (bottom-right) corner of each graph, the percentage of points above (below) the bisector is reported.}
	\label{fig:comparison_ct-struc_free_vs_PERS_3TIER}
\end{figure}

\section{Extended analysis: non-negative cross-temporal reconciliation}
\label{extanalysis}

In this section, we explore the performance of forecast reconciliation approaches able to produce
non-negative PV forecasts, both temporally and spatially coherent. For this reason, among the approaches proposed by \cite{Yagli2019}, we consider the reference benchmark PERS$_{bu}$, the NWP base forecasts 3TIER$_{bu}$, and oct$(struc)$. 
These approaches are then compared with the following 7 cross-temporal forecasting procedures:
\begin{itemize}
	\item KA$(wlsv_{te},wls_{cs})$: the heuristic approach by \cite{Kourentzes2019}, using te$(wlsv)$ in the first step, and cs$(wls)$ in the second, respectively;
	\item ct$(wlsv_{te},bu_{cs})$, ct$(struc_{te},bu_{cs})$, ct$(wls_{cs},bu_{te})$, and ct$(struc_{cs},bu_{te})$: partly bottom-up cross temporal reconciliation (see section \ref{sec:ctbu}) of 
	the high frequency bottom time series' reconciled forecasts  according to, respectively, te$(wlsv)$, te$(struc)$, cs$(wls)$, and cs$(struc)$;
	\item oct$(wlsv)$ and oct$(bdshr)$: optimal (in least squares sense) cross-temporal forecast reconciliation approaches using the in-sample forecast errors (\citealp{DiFonzoGiro2021}).
\end{itemize}


\subsection{Non negative forecast reconciliation: sntz}
Each approach, when used in its `free' version, i.e., without considering non-negative constraints in the linearly constrained quadratic program (\ref{glsobj}), is not guaranteed to always produce non-negative reconciled forecasts (see section \ref{nn_issues}). This fact may be an issue for the analyst, since in many practical situations negative forecasts could have no meaning, thus undermining the quality of the results found and the conclusions thereof.
In what follows, we consider a simple heuristic strategy to avoid negative reconciled forecasts, without using any sophisticated, and time consuming, numerical optimization procedure\footnote{Recent contributions on this topic in the hierarchical forecasting field are \cite{Wickramasuriya2020} and \cite{FoReco2022}.}.
More precisely, possible negative values of the free reconciled high-frequency bottom time series forecasts are set to zero. Denote $\widetilde{\Bvet}^{[1]}_0$ the matrix containing the non-negative reconciled forecasts produced by a `free' approach,  and the `zeroed' ones. The complete vector of non-negative cross-temporal reconciled forecasts is computed as the cross-temporal bottom-up aggregation of $\widetilde{\bvet}^{[1]}_0 = \text{vec}\left[\left(\widetilde{\Bvet}^{[1]}_0\right)'\right]$, that is, according to expression (\ref{ct-bu_formula_pap}):
$$
\widetilde{\yvet}_0 = \Fvet\widetilde{\bvet}^{[1]}_0 .
$$

\noindent We call set-negative-to-zero (sntz) this simple, and quick device to obtain non negative reconciled forecasts. While it certainly increases the forecasting accuracy of the high-frequency bottom time series forecasts wrt the `free' counterparts, this does not hold true in general for the upper level series forecasts. In addition, even if the originally reconciled forecasts are obtained according to an unbiased approach, the sntz-reconciled forecasts are no more unbiased, like the non-negative forecasts obtained through numerical optimization procedures (\citealp{Wickramasuriya2020}). However, in practical situations the differences between the results produced by the sntz heuristic, and those obtained through a numerical optimization procedure could be negligible. For example, for the dataset and the forecasting experiment in hand, we have found that the forecasting performance of the non-negative reconciliation approaches through the non-linear optimization procedure \texttt{osqp} (\citealp{Stellato2020}) implemented in \texttt{FoReco} (\citealp{FoReco2022}), is basically the same as the one obtained with the sntz heuristic.
Table \ref{table_octstrucnn} shows the nRMSE(\%) of the reconciled forecasts through the oct($struc$) approach in both free and non-negative (sntz and osqp) variants. It is worth noting that the simple sntz heuristic always gives the lowest nRMSE(\%), independently of the temporal granularity of the forecasts\footnote{Similar results have been found for all the considered reconciliation approaches. The details are available at request from the authors.}.

\begin{table}[H]
	\centering
	\caption{Forecast accuracy (nRMSE\%) of free and non-negative reconciled forecasts using the oct($struc$) approach. All temporal aggregation orders are considered, from hourly ($k=1$) to daily ($k=24$). Forecast horizon: operating day.
	Bold entries identify the best approach.}
	\label{table_octstrucnn}
	\begin{footnotesize}
	\setlength{\tabcolsep}{4pt}
\begin{tabular}[t]{ccccccccc}
	\toprule
non-negative	      & \multicolumn{8}{c}{$k$: temporal aggregation order}\\
reconciliation & 1 & 2 & 3 & 4 & 6 & 8 & 12 & 24\\
	\midrule
\multicolumn{9}{c}{$\mathcal{L}_0$ (1 series)}\\
free &  26.71 & 26.35 & 25.67 & 25.57 & 23.03 & 24.80 & 16.74 & 16.24\\
sntz &  \textbf{26.64} & \textbf{26.27} & \textbf{25.56} & \textbf{25.48} & \textbf{22.86} & \textbf{16.25} & \textbf{24.71} & \textbf{15.73}\\
osqp &  26.74 & 26.36 & 25.67 & 25.55 & 23.01 & 24.77 & 16.36 & 15.85\\
	\midrule
\multicolumn{9}{c}{$\mathcal{L}_1$ (5 series)}\\
free &  33.11 & 32.54 & 31.60 & 31.37 & 28.19 & 30.06 & 20.49 & 19.64\\
sntz &  \textbf{32.99} & \textbf{32.41} & \textbf{31.45} & \textbf{31.23} & \textbf{27.96} & \textbf{29.91} & \textbf{19.88} & \textbf{19.00}\\
osqp &  33.11 & 32.52 & 31.58 & 31.32 & 28.15 & 29.99 & 20.01 & 19.14\\
	\midrule
\multicolumn{9}{c}{$\mathcal{L}_2$ (318 series)}\\
free & 46.74 & 44.02 & 42.05 & 41.07 & 36.67 & 38.16 & 26.56 & 24.73\\
sntz &  \textbf{46.51} & \textbf{43.80} & \textbf{41.80} & \textbf{40.83} & \textbf{36.34} & \textbf{37.90} & \textbf{25.85} & \textbf{23.97}\\
osqp & 46.63 & 43.93 & 41.94 & 40.93 & 36.53 & 37.99 & 25.97 & 24.10\\
	\bottomrule
\end{tabular}

\end{footnotesize}
\end{table}

\noindent This result is visually confirmed by the graphs in figure \ref{fig:octstrucnn}, showing 
the scatter plots of the 324 couples of nRMSE(\%) for oct$(struc)_{osqp}$ \textit{vs}. oct$(struc)_{sntz}$. For this dataset, we found that oct$(struc)_{sntz}$  improves on oct$(struc)_{osqp}$ in not less than the 94\% of the 324 series for any temporal granularity.

\begin{figure}[htb]
	\centering
	\includegraphics[width=1.00\linewidth]{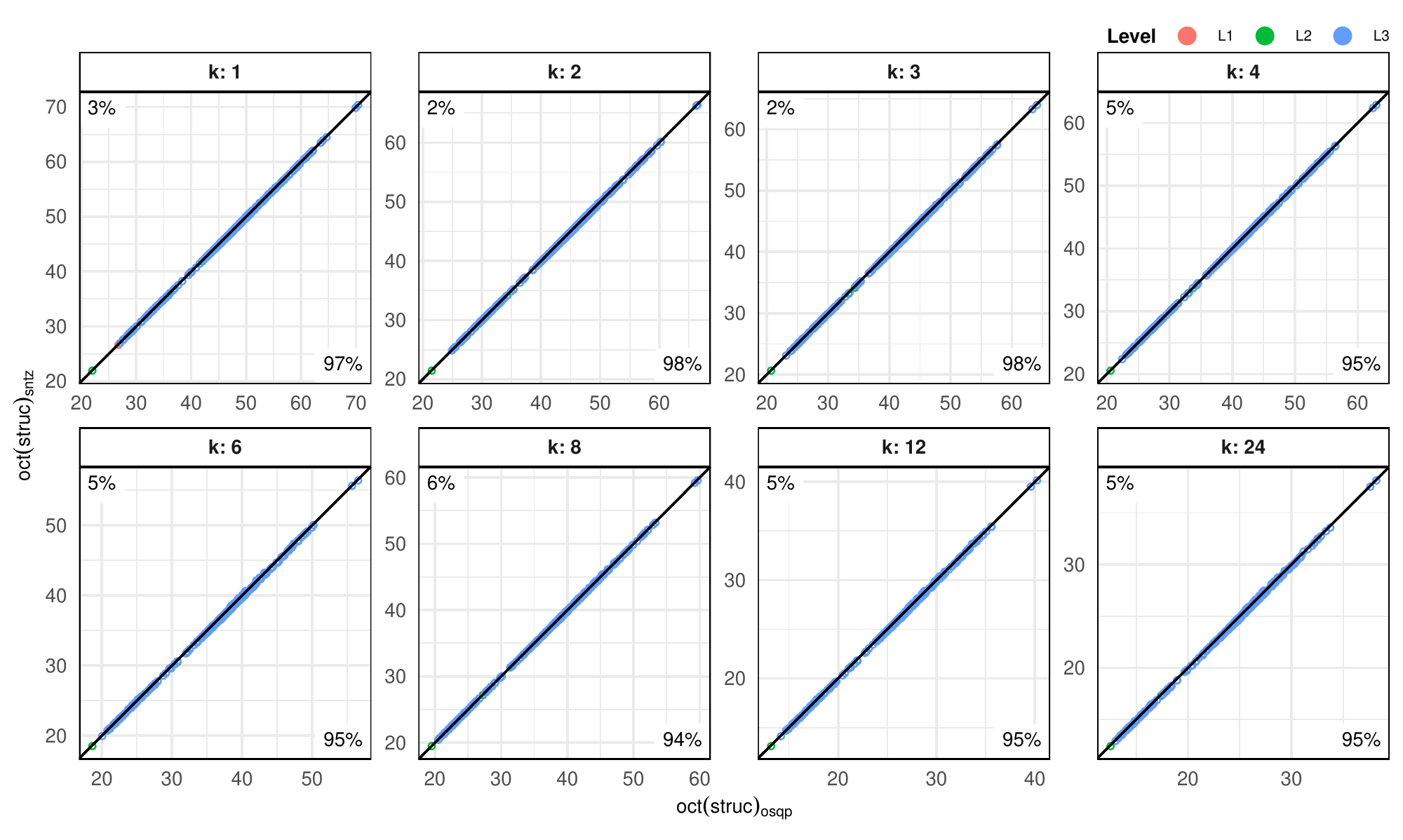}
	\caption{Comparison of nRMSE(\%) between sntz and osqp non-negative forecast reconciliation using the oct($struc$) approach. The black line represents the bisector, where the nRMSE's for oct($struc$)$_{osqp}$ and oct($struc$)$_{sntz}$ are equal. On the top-left (bottom-right) corner of each graph, the percentage of points above (below) the bisector is reported.}
	\label{fig:octstrucnn}
\end{figure}

\subsection{Forecast accuracy of the selected approaches}

Table \ref{table_final} shows the nRMSE(\%)  of the considered non-negative cross-temporal forecast reconciliation approaches, using the sntz heuristic, and the corresponding forecast skills over the benchmark forecasts PERS$_{bu}$. Overall, the accuracy improvements of the newly proposed approaches over the persistence model are in the range 16.5\%–46.5\%, whereas $3TIER_{bu}$ improvements are in the range 10.5\%-38.3\%.
Partly bottom-up approaches, with cross-sectional reconciliation at the first step (i.e., ct$(wls_{cs}, bu_{te})$ and
ct$(struc_{cs}, bu_{te})$) show a good performance for the six series at levels $\mathcal{L}_0$ and $\mathcal{L}_1$. For the 318 most disaggregated series at the bottom level, ct$(wlsv_{te}, bu_{cs})$ and oct($wlsv$) rank first and second, respectively, and their accuracy indices are very close each other.

\begin{table}[H]
\centering
	\caption{Forecast accuracy of selected non-negative, cross-temporal reconciliation approaches and base forecasts in terms of nRMSE(\%). Hourly and daily forecasts, forecast horizon: operating day.
	Bold entries and italic entries identify the best and the second best performing approaches, respectively.}
	\label{table_final}
	\begin{footnotesize}
		\setlength{\tabcolsep}{4pt}
	    \begin{tabular}[t]{l|cc|cc|cc|cc|cc|cc}
\toprule
 &  \multicolumn{2}{c|}{$\mathcal{L}_0$} & \multicolumn{2}{c|}{$\mathcal{L}_1$} & \multicolumn{2}{c|}{$\mathcal{L}_2$}  &  \multicolumn{2}{c|}{$\mathcal{L}_0$} & \multicolumn{2}{c|}{$\mathcal{L}_1$} & \multicolumn{2}{c}{$\mathcal{L}_2$}\\
 \cmidrule(lr){2-7}  \cmidrule(lr){8-13}
 & H & D & H & D & H & D& H & D & H & D & H & D\\
\midrule
 & \multicolumn{6}{c|}{$nRMSE(\%)$} & \multicolumn{6}{c}{\textit{forecast skill}}\\
PERS$_{bu}$ & 34.62 & 20.23 & 43.15 & 24.57 & 59.75 & 30.65 & 0.000 & 0.000 & 0.000 & 0.000 & 0.000 & 0.000\\
3TIER$_{bu}$ & 26.03 & 12.48 & 33.95 & 16.75 & 53.46 & 25.19 & 0.248 & 0.383 & 0.213 & 0.318 & 0.105 & 0.178\\
oct$(struc)$ & 26.64 & 15.73 & 32.99 & 19.00 & 46.51 & 23.97 & 0.231 & 0.222 & 0.235 & 0.227 & 0.222 & 0.218\\
\hline
KA$(wlsv_{te},wls_{cs})$ & 23.59 & 13.55 & 30.23 & 16.91 & 44.48 & 22.21 & 0.319 & 0.330 & 0.299 & 0.312 & 0.256 & 0.275\\
ct$(struc_{cs},bu_{te})$ & \textit{21.99} & \textit{12.38} & \textbf{28.80} & \textit{15.82} & 49.76 & 24.33 & \textit{0.365} & \textit{0.388} & \textbf{0.333} & 0.356 & 0.167 & 0.206\\
ct$(wls_{cs},bu_{te})$ & \textbf{21.23} & \textbf{10.82} & \textit{28.97} & \textbf{15.01} & 49.88 & 23.66 & \textbf{0.387} & \textbf{0.465} & \textit{0.329} & \textbf{0.389} & 0.165 & 0.228\\
ct$(struc_{te},bu_{cs})$ & 24.81 & 14.43 & 31.41 & 17.80 & 45.30 & 22.93 & 0.283 & 0.287 & 0.272 & 0.276 & 0.242 & 0.252\\
ct$(wlsv_{te},bu_{cs})$ & 23.50 & 13.42 & 30.13 & 16.78 & \textbf{44.42} & \textbf{22.12} & 0.321 & 0.337 & 0.302 & 0.317 & \textbf{0.257} & \textbf{0.278}\\
oct$(wlsv)$ & 23.63 & 13.64 & 30.21 & 17.01 & \textit{44.44} & \textit{22.27} & 0.317 & 0.326 & 0.300 & 0.308 & \textit{0.256} & \textit{0.273}\\
oct$(bdshr)$ & 24.13 & 13.92 & 30.45 & 17.12 & 45.24 & 22.54 & 0.303 & 0.312 & 0.294 & 0.303 & 0.243 & 0.265\\
\bottomrule
\end{tabular}
	\end{footnotesize}
\end{table}

\noindent From table \ref{table_L1L2L3_k1-24_final_l2}, we see that even at distinct $\mathcal{L}_1$ series (Transmission Zones), the newly proposed approaches perform better than all the approaches considered by \cite{Yagli2019}, as reported in table \ref{table_Yagli_original_comp_L2}.

\begin{table}[H]
	\centering
	\caption{Forecast accuracy of selected non-negative, cross-temporal reconciliation approaches and base forecasts in terms of nRMSE(\%) for the series at $\mathcal{L}_1$ level. Hourly and daily forecasts, forecast horizon: operating day.
		Bold entries and italic entries identify the best and the second best performing approaches, respectively.}
	\label{table_L1L2L3_k1-24_final_l2}
	\begin{footnotesize}
		\setlength{\tabcolsep}{4pt}
		\begin{tabular}[t]{l|cc|cc|cc|cc|cc}
\toprule
 &  \multicolumn{2}{c|}{TZ$_1$} & \multicolumn{2}{c|}{TZ$_2$} & \multicolumn{2}{c|}{TZ$_3$}  &  \multicolumn{2}{c|}{TZ$_4$} & \multicolumn{2}{c}{TZ$_5$} \\
 \cmidrule(lr){2-7}  \cmidrule(lr){8-11}
 & H & D & H & D & H & D& H & D & H & D \\
\midrule
 & \multicolumn{10}{c}{$nRMSE(\%)$}\\
PERS$_{bu}$ 				& 28.72 & 16.12 & 40.27 & 22.40 & 46.48 & 26.96 & 52.82 & 29.75 & 47.44 & 27.62 \\ 
3TIER$_{bu}$ 				& 22.81 & \textit{10.43} & 33.34 & 16.72 & 34.01 & 17.06 & 46.40 & 22.80 & 33.16 & \textit{16.75} \\ 
oct$(struc)$ 				& 21.92 & 12.41 & 30.84 & 17.30 & 34.67 & 20.21 & 41.84 & 23.84 & 35.66 & 21.23 \\
\hline
KA$(wlsv_{te},wls_{cs})$ 	& 20.23 & 11.10 & 28.57 & 15.43 & 30.98 & 17.45 & 39.72 & 22.09 & 31.64 & 18.51 \\ 
ct$(struc_{cs},bu_{te})$ 	& \textit{20.15} & 10.92 & \textbf{27.01} & \textit{14.37} & \textbf{28.74} & \textit{15.74} & \textbf{38.12} & \textit{20.89} & \textit{29.99} & 17.17 \\ 
ct$(wls_{cs},bu_{te})$ 		& \textbf{19.33} & \textbf{9.72} & \textit{28.11} & \textbf{14.18} & \textit{29.55} & \textbf{15.10} & \textit{39.20} & \textbf{20.17} & \textbf{28.68} & \textbf{15.88} \\ 
ct$(struc_{te},bu_{cs})$ 	& 21.04 & 11.63 & \textit{29.55} & 16.24 & 32.57 & 18.73 & 40.51 & 22.67 & 33.37 & 19.75 \\ 
ct$(wlsv_{te},bu_{cs})$ 		& 20.20 & 10.98 & 28.53 & 15.35 & 30.89 & 17.37 & 39.69 & 22.01 & 31.32 & 18.19 \\ 
oct$(wlsv)$ 				& 20.18 & 11.18 & 28.53 & 15.50 & 30.97 & 17.53 & 39.56 & 22.16 & 31.84 & 18.69 \\ 
oct$(bdshr)$ 				& 21.00 & 11.74 & 29.25 & 16.02 & 31.63 & 17.85 & 39.73 & 22.27 & 30.65 & 17.74 \\
\midrule
 & \multicolumn{10}{c}{\textit{forecast skill}}\\
PERS$_{bu}$ 				& 0 & 0 & 0 & 0 & 0 & 0 & 0 & 0 & 0 & 0 \\ 
3TIER$_{bu}$ 				& 0.206 & \textit{0.353} & 0.172 & 0.254 & 0.268 & 0.367 & 0.121 & 0.234 & 0.301 & \textit{0.394} \\ 
oct$(struc)$ 				& 0.237 & 0.230 & 0.234 & 0.228 & 0.254 & 0.250 & 0.208 & 0.199 & 0.248 & 0.231 \\ 
\hline
KA$(wlsv_{te},wls_{cs})$ 	& 0.296 & 0.312 & 0.291 & 0.311 & 0.334 & 0.353 & 0.248 & 0.258 & 0.333 & 0.330 \\ 
ct$(struc_{cs},bu_{te})$ 	& \textit{0.299} & 0.323 & \textbf{0.329} & \textit{0.359} & \textbf{0.382} & \textit{0.416} & \textbf{0.278} & \textit{0.298} & \textit{0.368} & 0.378 \\ 
ct$(wls_{cs},bu_{te})$ 		& \textbf{0.327} & \textbf{0.397} & \textit{0.302} & \textbf{0.367} & \textit{0.364} & \textbf{0.440} & \textit{0.258} & \textbf{0.322} & \textbf{0.396} & \textbf{0.425} \\
ct$(struc_{te},bu_{cs})$ 	& 0.268 & 0.279 & 0.266 & 0.275 & 0.299 & 0.305 & 0.233 & 0.238 & 0.297 & 0.285 \\ 
ct$(wlsv_{te},bu_{cs})$ 		& 0.297 & 0.319 & 0.292 & 0.315 & 0.335 & 0.356 & 0.249 & 0.260 & 0.340 & 0.341 \\ 
oct$(wlsv)$ 				& 0.297 & 0.307 & 0.292 & 0.308 & 0.334 & 0.350 & 0.251 & 0.255 & 0.329 & 0.323 \\  
oct$(bdshr)$ 				& 0.269 & 0.272 & 0.274 & 0.285 & 0.320 & 0.338 & 0.248 & 0.251 & 0.354 & 0.358 \\ 
\bottomrule
\end{tabular}
	\end{footnotesize}
\end{table}

\noindent In addition, almost all the new approaches significantly outperform $PERS_{bu}$, $3TIER_{bu}$, and oct$(struc)$ for both hourly and daily forecasts (figure \ref{fig:bigvertical2}).

\begin{figure}[H]
	\centering
	\includegraphics[width=1.00\linewidth]{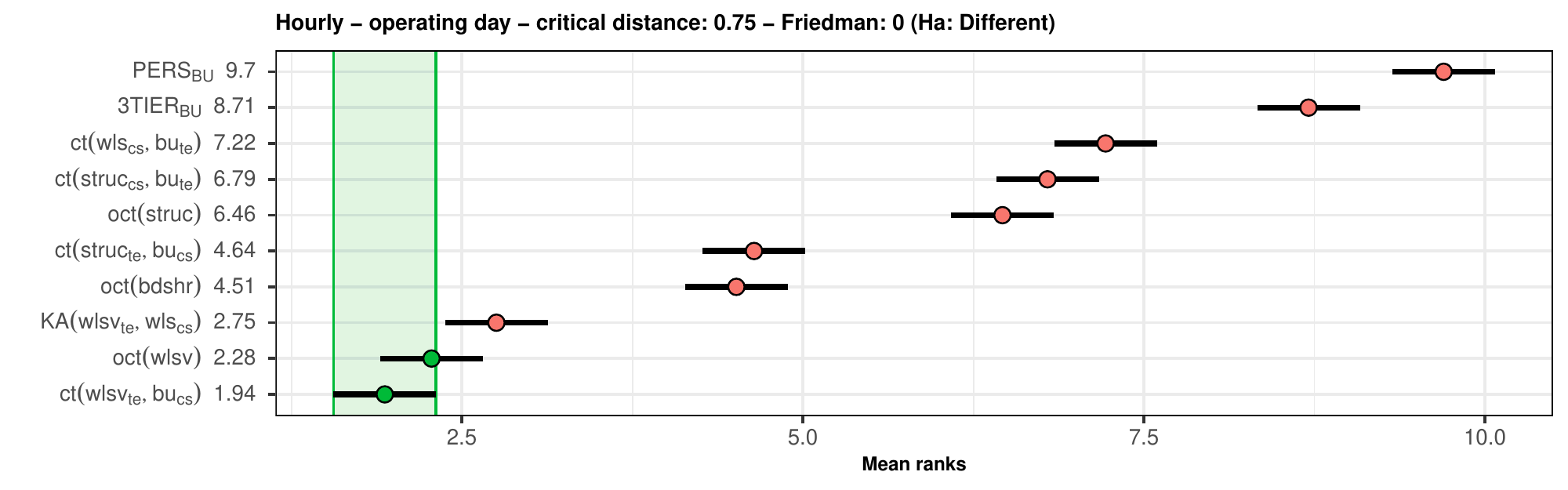}
	\includegraphics[width=1.00\linewidth]{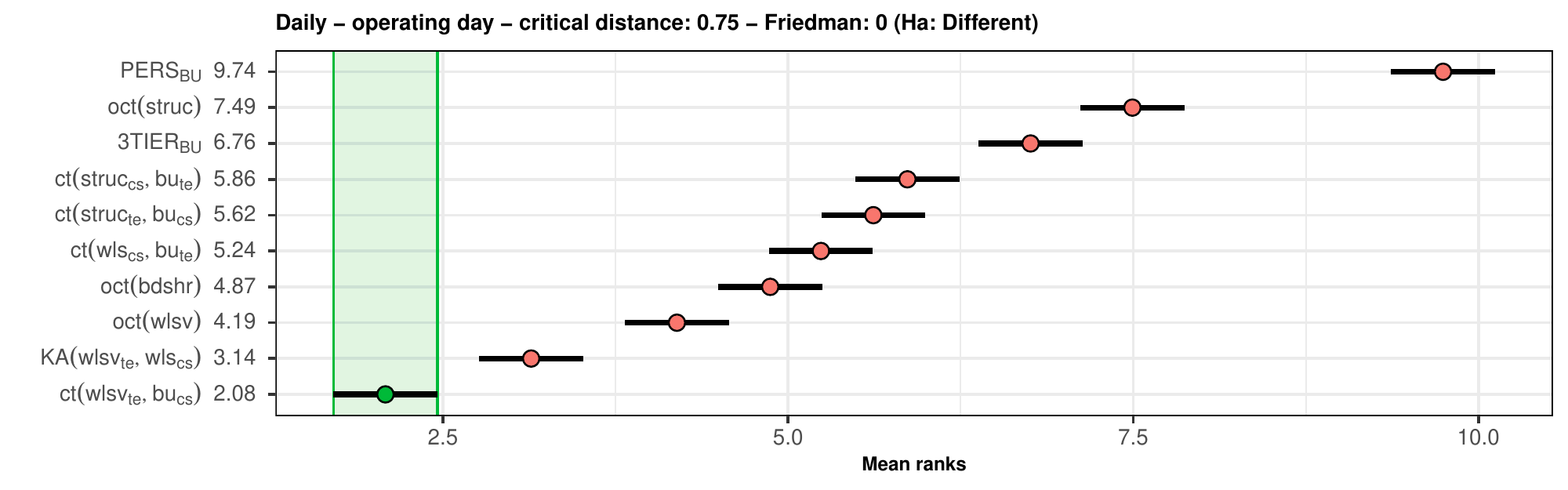}
	\caption{MCB-Nemenyi test on selected non-negative cross-temporal reconciliation approaches with operating day forecast horizon. $\mathcal{L}_0, \mathcal{L}_1, \mathcal{L}_2$ levels (324 series). Top panel: hourly forecasts; Bottom panel: daily forecasts.
		The mean rank of each approach is displayed to the right of their names. If the intervals of two forecast reconciliation approaches do not overlap, this indicates a statistically different performance. Thus, approaches that do not overlap with the green interval are considered significantly worse than the best, and vice-versa.}
	\label{fig:bigvertical2}
\end{figure}

\noindent Finally, focusing on the two best performing approaches of the forecasting experiment (ct$(wlsv_{te},bu_{cs})$ and oct($wlsv)$), and looking at the nRMSEs of the individual series, it is worth noting that:
\begin{itemize}
\item ct$(wlsv_{te},bu_{cs})$ always produces more accurate forecasts than oct$(struc)$, for any series and granularity (figure \ref{fig:comparison_ct-struc0-ctbu-t-wlsv0}),

\item the accuracy increases of ct$(wlsv_{te},bu_{cs})$ over the NWP approach 3TIER$_{bu}$ are still clear (figure \ref{fig:comp2}), even though for daily forecasts 3TIER$_{bu}$ performs better in about 1 case out 4;

\item the accuracy of ct$(wlsv_{te},bu_{cs})$ and oct($wlsv)$ are practically indistinguishable (figure \ref{fig:comparison_ct-wlsv0-ctbu-t-wlsv0});
\end{itemize}

\begin{figure}[H]
	\centering
	\includegraphics[width=0.9\linewidth]{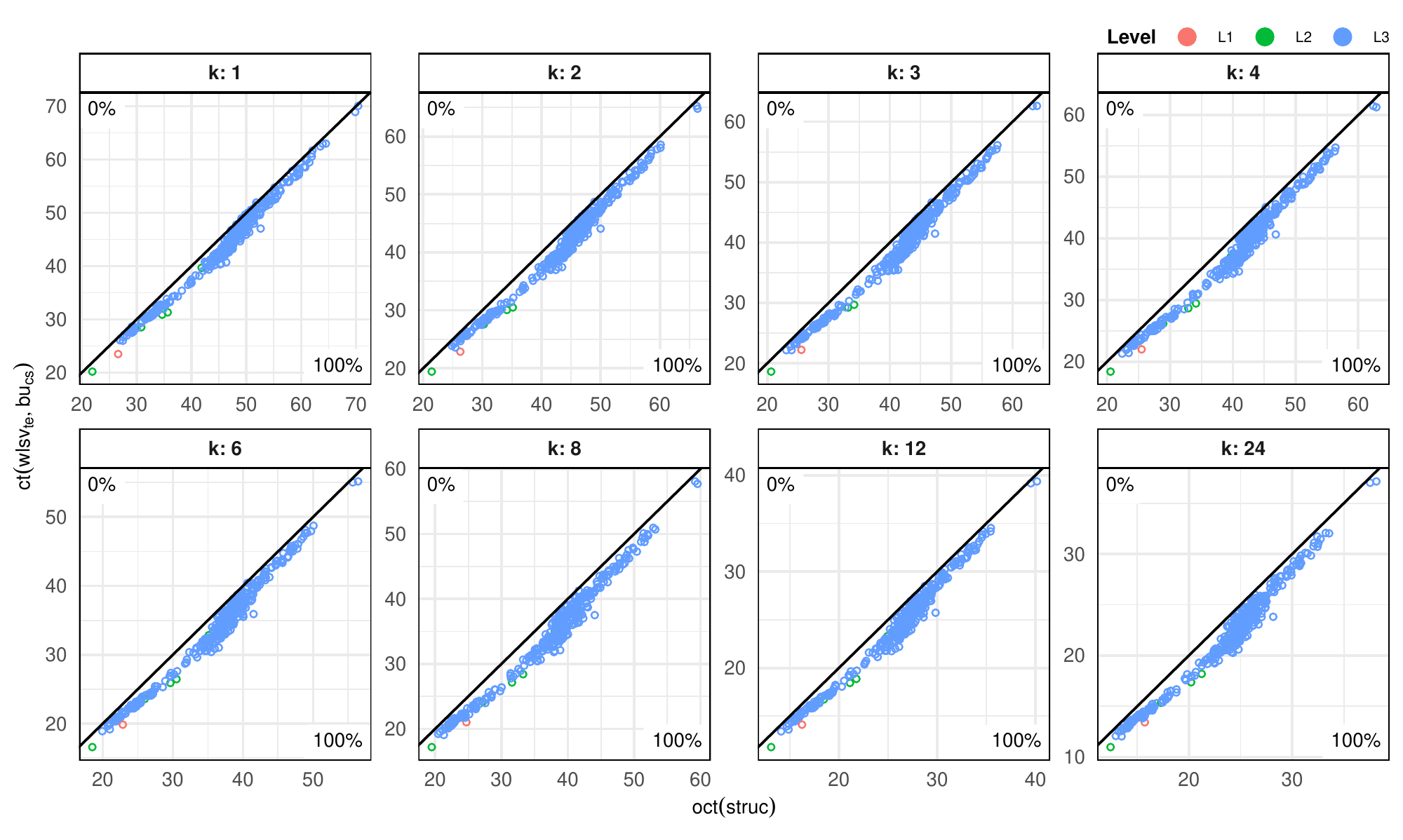}
	\caption{Comparison of nRMSE(\%) between non-negative reconciliation approaches: oct$(struc)$ and ct$(wlsv_{te},bu_{cs})$, forecast horizon: operating day. The black line represents the bisector, where the nRMSE's for both approaches are equal. On the top-left (bottom-right) corner of each graph, the percentage of points above (below) the bisector is reported.}
	\label{fig:comparison_ct-struc0-ctbu-t-wlsv0}
\end{figure}

\begin{figure}[H]
	\centering
	\includegraphics[width=0.9\linewidth]{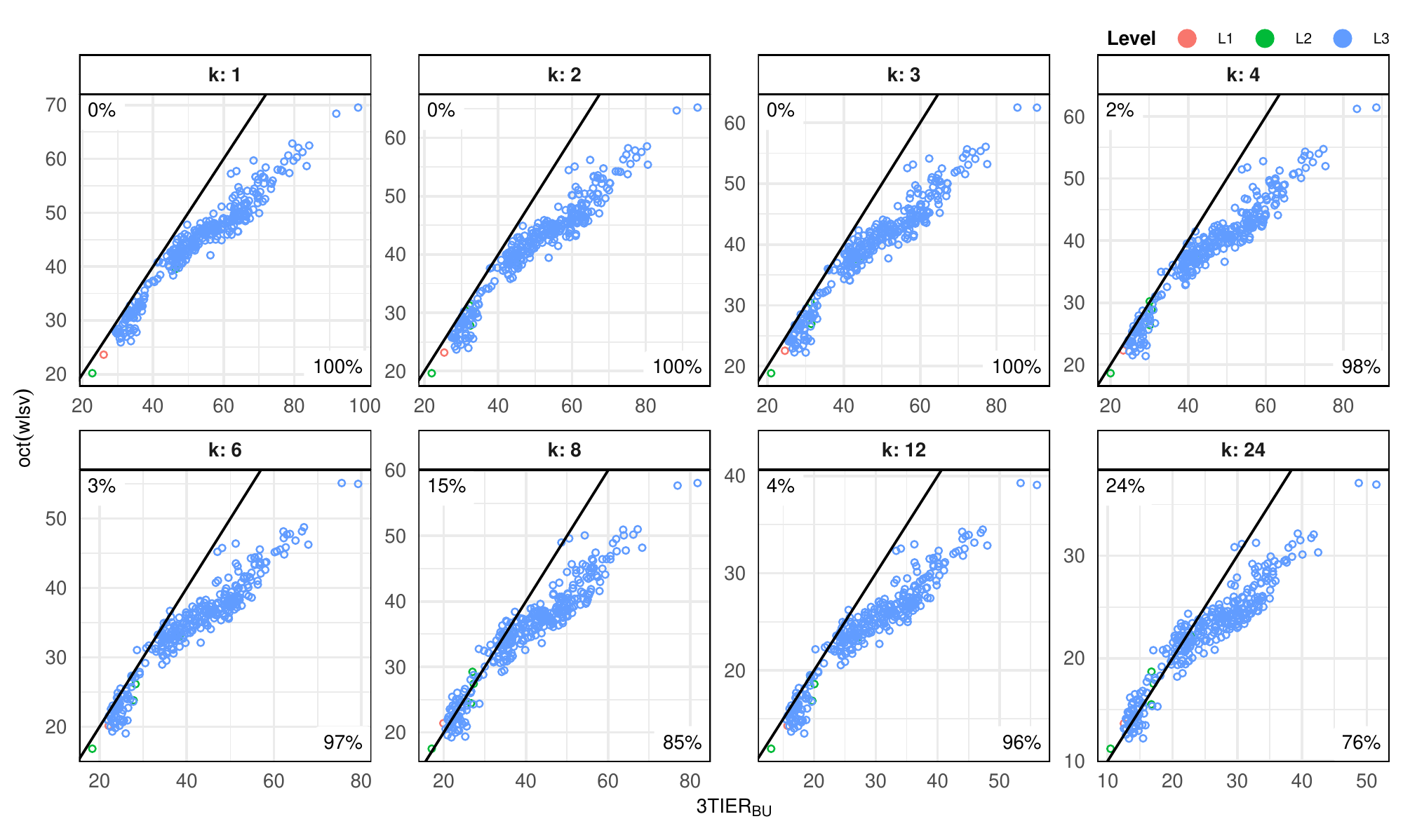}
	\caption{Comparison of nRMSE(\%) between non-negative reconciliation approaches: $3TIER_{bu}$ and ct$(wlsv_{te},bu_{cs})$. Forecast horizon: operating day. The black line represents the bisector, where the nRMSE's for both approaches are equal. On the top-left (bottom-right) corner of each graph, the percentage of points above (below) the bisector is reported.}
	\label{fig:comp2}
\end{figure}

\begin{figure}[H]
	\centering
	\includegraphics[width=0.9\linewidth]{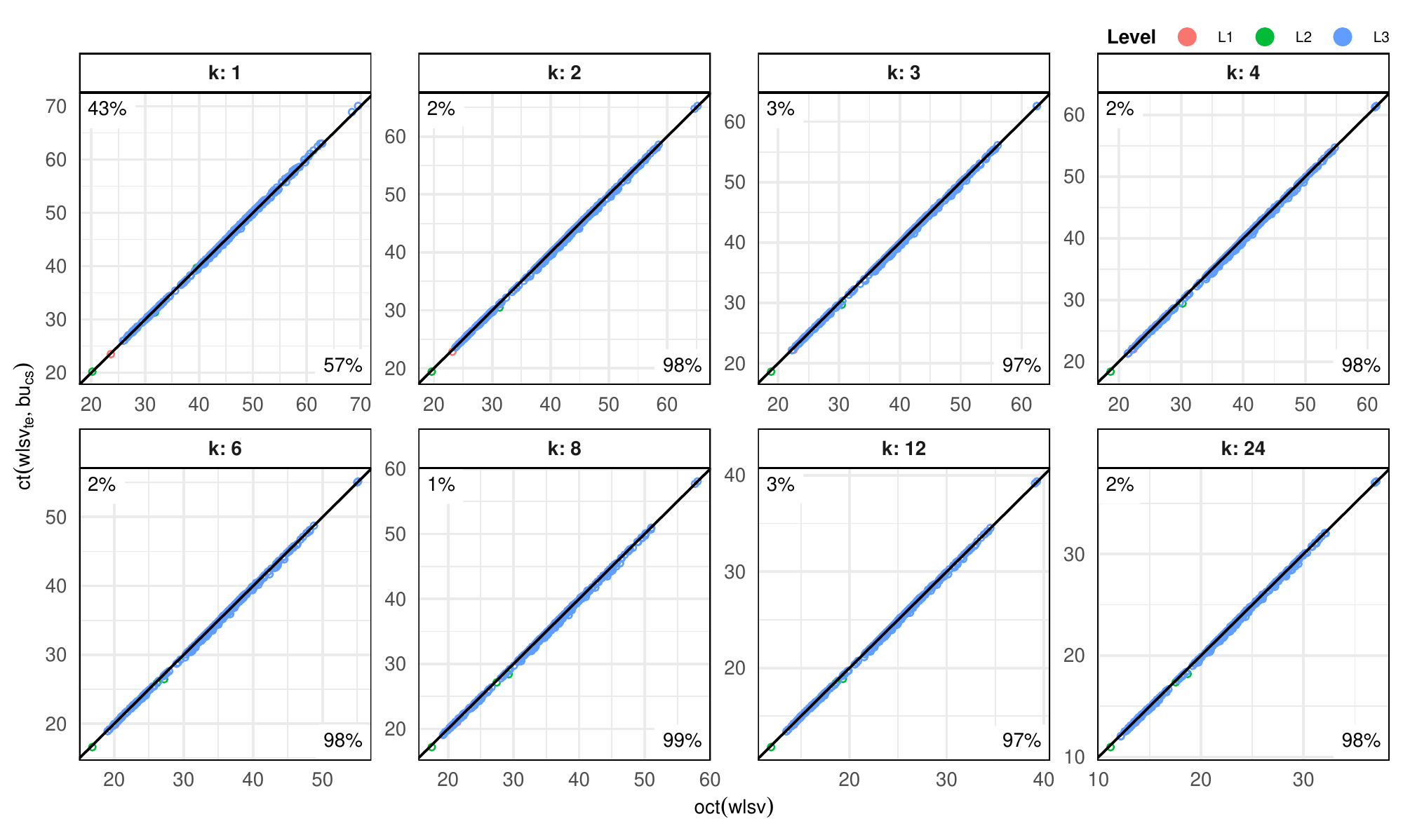}
	\caption{Comparison of nRMSE(\%) between non-negative reconciliation approaches: oct$(wlsv)$ and ct$(wlsv_{te},bu_{cs})$, forecast horizon: operating day. The black line represents the bisector, where the nRMSE's for both approaches are equal. On the top-left (bottom-right) corner of each graph, the percentage of points above (below) the bisector is reported.}
	\label{fig:comparison_ct-wlsv0-ctbu-t-wlsv0}
\end{figure}

We may thus conclude that, for the PV324 dataset considered in this work, a thorough exploitation of  cross-temporal hierarchies significantly improves the forecasting accuracy as compared to the approaches considered in \cite{Yagli2019}. In particular, the use of the in-sample base forecast errors, even in the simple diagonal versions of te($wlsv$) (first step of ct($wlsv_{te},bu_{cs}$)), and oct($wlsv$), increases the forecasting accuracy at different time granularities.

\section{Conclusion}
\label{sec:conclusion}

Renewable energy is providing increasingly more energy to the grid all over the world. But grid operators must carefully manage the balance between the generation and consumption of energy to make the best use of abundant renewable energy.
For an ISO, this provides greater grid stability, higher revenue and better use of what sun is available at any one time.
Better short-term solar energy forecasts mean lower-emissions, cheaper energy and a more stable electricity grid.
Solar forecasting is thus a key tool to achieve these results.
In this paper, cross-temporal point forecast reconciliation has been applied to generate non-negative, fully coherent (both in space and time) forecasts of PV generated power. Both methodological and practical issues have been tackled, in order to develop effective and easy-to-handle cross-temporal forecasting approaches.
Besides assuring both spatial and temporal coherence, and non-negativity of the reconciled forecasts, the results show that for the considered dataset, cross-temporal forecast reconciliation significantly improves on 
the sequential procedures proposed by \cite{Yagli2019}, at any cross-sectional level of the hierarchy and for any temporal aggregation order.
However, these findings should be considered neither conclusive, nor valid in general. Other forecasting experiments, through simulations and using other datasets, as well as in solar forecasting and in other application fields, should be performed to empirically robustify the results shown so far.
In addition, other research queries could raise in the field of cross-temporal PV forecast reconciliation, like
using reduced temporal hierarchies, as well as forecasting through appropriate Machine Learning end-to-end approaches (\citealp{Stratigakos2022}), and probabilistic instead of deterministic (point) forecasting (\citealp{Panamtash2018}, \citealp{Jeon2019}, \citealp{Panagiotelis2022}, \citealp{Yang2020probnonpar}, \citealp{Yang2020probhomosc}, \citealp{BenTaieb2021}).
All these topics are left for future research.



\appendix
\setcounter{table}{0}
\setcounter{figure}{0}

\section{Cross-temporal bottom-up forecast reconciliation}
\label{sec:ctbu-appendix}
A visual representation of cross-sectional, temporal, and cross-temporal bottom-up reconciliation is shown in figure \ref{fig:big}, where the `bottom time series' base forecasts are highlighted on a purple background, and the `upper time series' bottom-up reconciled forecasts are highlighted on a pink background. The arrows show the direction(s) along with bottom-up reconciliation is performed. The deponent abbreviations `$cs$', `$te$', `$ct$', are alternatively used to denote reconciled forecasts according to cross-sectional, temporal, and cross-temporal bottom-up, respectively.

\begin{figure}[H]
     \centering
     \begin{subfigure}[b]{0.49\textwidth}
	\centering
	\caption{$\widetilde{\Yvet}_{cs(bu)}$: cross-sectional bottom-up}
	\begin{tikzpicture}[>=latex, line width=1pt,
		Matrix/.style={
            matrix of nodes,
            font=\large,
            align=center,
            text width = 1.15cm, 
            text height = 0.5cm, 
            column sep=4pt,
            row sep=7pt,
            nodes in empty cells,
            left delimiter={[},
            right delimiter={]}
		}]
		\matrix[Matrix] (Mcs){ 
			$\widetilde{\Avet}_{cs(bu)}^{[m]}$ & \dots & $\widetilde{\Avet}_{cs(bu)}^{[k_2]}$ & $\widetilde{\Avet}_{cs(bu)}^{[1]}$ \\
			$\widehat{\Bvet}^{[m]}$ & \dots & $\widehat{\Bvet}^{[k_2]}$ & $\widehat{\Bvet}^{[1]}$ \\
		};
		\draw[<-] (Mcs.north east)++(0.4,0) coordinate (temp) -- (temp |- Mcs.south) node [midway,label={[label distance=0.1cm,rotate=-90, xshift = 1.5mm, font=\footnotesize]Cross-sectional}]{};
		\draw[<-, opacity = 0] (Mcs.south west)++(0,-0.15) coordinate (temp) -- (temp -| Mcs.east) node [midway,label={[label distance=0cm,xshift = 1.5mm, font=\footnotesize]below:Temporal}]{};
		\node[opacity=0.2,
		rounded corners,
		inner sep=0pt, fill = blue, fit=(Mcs-2-1)(Mcs-2-4)](Bcs){};
		\node[opacity=0.2,
		rounded corners,
		inner sep=0pt, fill = red, fit=(Mcs-1-1)(Mcs-1-4)](Acs){};
	\end{tikzpicture}
	\label{fig:csbu}
	\end{subfigure}
     \hfill
     \begin{subfigure}[b]{0.49\textwidth}
	         \centering
	         \caption{$\widetilde{\Yvet}_{te\;bu}$: temporal bottom-up}
	         \begin{tikzpicture}[>=latex, line width=1pt,
        Matrix/.style={
            matrix of nodes,
            font=\large,
            align=center,
            text width = 1.15cm, 
            text height = 0.5cm, 
            column sep=4pt,
            row sep=7pt,
            nodes in empty cells,
            left delimiter={[},
            right delimiter={]}
        }]
    \matrix[Matrix] (Mt){ 
        $\widetilde{\Avet}_{te(bu)}^{[m]}$ & \dots & $\widetilde{\Avet}_{te(bu)}^{[k_2]}$ & $\widehat{\Avet}^{[1]}$ \\
        $\widetilde{\Bvet}_{te(bu)}^{[m]}$ & \dots & $\widetilde{\Bvet}_{te(bu)}^{[k_2]}$ & $\widehat{\Bvet}^{[1]}_{\phantom{t}}$ \\
    };
    \draw[<-] (Mt.south west)++(0,-0.15) coordinate (temp) 
        -- (temp -| Mt.east) node [midway,label={[label distance=0cm,xshift = 1.5mm, font=\footnotesize]below:Temporal}]{};
        \node[opacity=0.2,
            rounded corners,
            inner sep=0pt, fill = blue, fit=(Mt-1-4)(Mt-2-4)](Bt){};
        \node[opacity=0.2,
            rounded corners,
            inner sep=0pt, fill = red, fit=(Mt-1-1)(Mt-2-3)](At){}; 
    \end{tikzpicture}
	         \label{fig:tbu}
	     \end{subfigure}
     \begin{subfigure}[b]{0.49\textwidth}
	         \centering
	          \caption{$\widetilde{\Yvet}_{ct\;bu}$: cross-temporal bottom-up}
	          \begin{tikzpicture}[>=latex, line width=1pt,
        Matrix/.style={
            matrix of nodes,
            font=\large,
            align=center,
            text width = 1.15cm, 
            text height = 0.5cm, 
            column sep=4pt,
            row sep=7pt,
            nodes in empty cells,
            left delimiter={[},
            right delimiter={]}
        }]
    \matrix[Matrix] (M){ 
        $\widetilde{\Avet}_{ct(bu)}^{[m]}$ & \dots & $\widetilde{\Avet}_{ct(bu)}^{[k_2]}$ & $\widetilde{\Avet}_{ct(bu)}^{[1]}$ \\
        $\widetilde{\Bvet}_{ct(bu)}^{[m]}$ & \dots & $\widetilde{\Bvet}_{ct(bu)}^{[k_2]}$ & $\widehat{\Bvet}^{[1]}_{\phantom{.}}$ \\
    };
    \draw[<-] (M.north east)++(0.4,0) coordinate (temp) -- (temp |- M.south) node [midway,label={[label distance=0.1cm,rotate=-90, xshift = 1.5mm, font=\footnotesize]Cross-sectional}]{};
    \draw[<-] (M.south west)++(0,-0.15) coordinate (temp) 
        -- (temp -| M.east) node [midway,label={[label distance=0cm,xshift = 1.5mm, font=\footnotesize]below:Temporal}]{};
        \node[opacity=0.2,
            rounded corners,
            inner sep=0pt, fill = blue, fit=(M-2-4)(M-2-4)](Bct){};
        \path[rounded corners,inner sep=2pt, fill = red, opacity = 0.2]
    (M-2-3.south) -| (M-2-3.north east) -- (M-1-4.south west)  -| (M-1-4.north east) -| (M-2-1.west) |- (M-2-3.south);  
    \end{tikzpicture}
	     \end{subfigure}
        \caption{A `visual' representation of cross-sectional ($\widetilde{\Yvet}_{cs(bu)}$), temporal ($\widetilde{\Yvet}_{te(bu)}$), and cross-temporal ($\widetilde{\Yvet}_{ct(bu)}$) bottom-up reconciled forecasts.}
        \label{fig:big}
\end{figure}

\subsection*{Cross-sectional bottom-up reconciliation of $k^*+m$ hierarchical time series' base forecasts}
The cross-sectional bottom-up ($cs(bu)$) reconciliation of hierarchical time series' base forecasts at different temporal granularities (temporal aggregation levels) $k \in \mathcal{K}$, may be represented as follows:
\[
\widetilde{\yvet}_{\tau,cs(bu)}^{[k]} = \Svet\widehat{\bvet}_{\tau}^{[k]}, \quad \tau = 1, \ldots, \frac{m}{k}, \quad k \in \mathcal{K} \quad \rightarrow \quad \widetilde{\Yvet}_{cs(bu)}^{[k]} = \Svet \widehat{\Bvet}^{[k]}, \quad k \in \mathcal{K}.
\]
In compact notation we have:
\[
\widetilde{\Yvet}_{cs(bu)} = \Svet \widehat{\Bvet} \quad \rightarrow  \quad
\widetilde{\Yvet}_{cs(bu)}' = \widehat{\Bvet}'\Svet' \quad \rightarrow  \quad
\widetilde{\yvet}_{cs(bu)} = \left(\Svet \otimes \Ivet_{k^* + m}\right)\widehat{\bvet} ,
\]
where $\widetilde{\yvet}_{cs(bu)} = \text{vec}\left(\widetilde{\Yvet}_{cs(bu)}'\right)$, and $\widehat{\bvet} = \text{vec}\left(\widehat{\Bvet}'\right)$.

The $cs(bu)$ reconciled forecasts are obviously cross-sectionally coherent, but in general they are not temporally coherent:
\[
\begin{array}{rcl}
\Uvet' \widetilde{\Yvet}_{cs(bu)} = \Zerovet_{[n_a \times (k^*+m)]} & \leftrightarrow & \left(\Uvet' \otimes \Ivet_{k^* + m}\right) \widetilde{\yvet}_{cs(bu)} = \Zerovet_{[n_a (k^*+m) \times 1]} \\
\Zvet' \widetilde{\Yvet}_{cs(bu)}' \ne \Zerovet_{[k^* \times n]} & \leftrightarrow & \left(\Ivet_n \otimes \Zvet'\right) \widetilde{\yvet}_{cs(bu)} \ne \Zerovet_{[k^* n \times 1]}
\end{array} .
\]

\subsection*{Temporal bottom-up reconciliation of $n$ individual time series' base forecasts}
The temporal bottom-up ($te(bu)$) reconciliation of $n$ individual time series' base forecasts at different temporal granularities (temporal aggregation levels), may be represented as follows:
\[
\begin{array}{c}
\widetilde{\avet}_{i,te(bu)} = \Rvet\widehat{\avet}_i^{[1]}, \quad i=1, \ldots, n_a \\
\widetilde{\bvet}_{j,te(bu)} = \Rvet\widehat{\bvet}_j^{[1]}, \quad j=1, \ldots, n_b \\
\end{array} \quad 
\rightarrow \quad 
\begin{array}{c}
\widetilde{\Avet}_{te(bu)} = \Rvet\widehat{\Avet}^{[1]\prime}\\
\widetilde{\Bvet}_{te(bu)} = \Rvet\widehat{\Bvet}^{[1]\prime} \\
\end{array}
\]
and, finally,
\[
\widetilde{\Yvet}_{te(bu)}' = \Rvet \widehat{\Yvet}^{[1]\prime}
\quad \rightarrow\quad 
\widetilde{\yvet}_{te(bu)} = \left(\Ivet_n \otimes \Rvet\right)\widehat{\yvet}^{[1]} .
\]
\noindent The $te(bu)$ reconciled forecasts are obviously temporally coherent, but in general they are not cross-sectionally coherent:
\[
\begin{array}{rcl}
	\Zvet' \widetilde{\Yvet}_{te(bu)}' = \Zerovet_{[k^* \times n]} & \leftrightarrow & \left(\Ivet_n \otimes \Zvet'\right) \widetilde{\yvet}_{te(bu)} = \Zerovet_{[k^* n \times 1]} \\
	\Uvet' \widetilde{\Yvet}_{te(bu)} \ne \Zerovet_{[n_a \times (k^*+m)]} & \leftrightarrow & \left(\Uvet' \otimes \Ivet_{k^* + m}\right) \widetilde{\yvet}_{te(bu)} \ne \Zerovet_{[n_a (k^*+m) \times 1]}
\end{array} .
\]

\subsection*{Cross-temporal bottom-up reconciliation of $n$ individual time series' base forecasts for different temporal granularities}
The cross-temporal bottom-up ($ct(bu)$) reconciliation of $k^*+m$ hierarchical time series' base forecasts at different temporal granularities, 
may be represented as follows:
\begin{equation}
	\label{ct-bu_formula}
\text{vec}\left(\widetilde{\Yvet}_{ct(bu)}'\right) = \left(\Svet \otimes \Rvet\right) \text{vec}\left(\widehat{\Bvet}^{[1]\prime}\right)
\quad \leftrightarrow  \quad
\widetilde{\yvet}_{cs(bu)} \left(\Svet \otimes \Rvet\right) \widehat{\bvet}^{[1]} ,
\end{equation}
where $\widehat{\bvet}^{[1]} = \text{vec}\left(\widehat{\Bvet}^{[1]\prime}\right)$.
Let us consider now the coherency checks for the $ct(bu)$ reconciled forecasts:
\begin{itemize}
\item cross-sectional coherency
\[
\left(\Uvet' \otimes \Ivet_{k^* + m}\right) \left(\Svet \otimes \Rvet\right) \widehat{\bvet}^{[1]} = \left(\Uvet'\Svet \otimes \Rvet\right) \widehat{\bvet}^{[1]} = \Zerovet_{[n_a(k^*+m) \times 1]},
\]
since
$\Uvet'\Svet = \left[\Ivet_{n_a} \; -\Cvet\right] \left[\begin{array}{c} \Cvet \\ \Ivet_{n_b}\end{array} \right] = \Cvet - \Cvet = \Zerovet_{[n_a \times n_b]}$;
\item temporal coherency
\[
\left(\Ivet_{n} \otimes \Zvet' \right) \left(\Svet \otimes \Rvet\right) \widehat{\bvet}^{[1]} =
\left(\Svet \otimes \Zvet'\Rvet\right) \widehat{\bvet}^{[1]} = \Zerovet_{[k^* n \times 1]},
\]
since
$\Zvet'\Rvet = \left[\Ivet_{k^*} \; -\Kvet\right] \left[\begin{array}{c} \Kvet \\ \Ivet_{m}\end{array} \right] = \Kvet - \Kvet = \Zerovet_{[k^* \times m]}$.
\end{itemize}

\noindent In the end, cross temporal bottom-up reconciliation can be tought of as a two-step sequential cross temporal reconciliation approach, where either cross-sectional reconciliation of the high-frequency bottom time series base forecasts (the `very' bottom series of a cross-temporal hierarchy), is followed by temporal reconciliation, or \textit{vice-versa}.
For, looking at expression (\ref{ct-bu_formula}), and by exploiting the properties of the vec operator (\citealp{Harville2008}, p. 345), it is:
\[
\left(\Svet \otimes \Rvet\right) \text{vec}\left(\widehat{\Bvet}^{[1]\prime}\right) =
\text{vec}\left[\Rvet \left(\widehat{\Bvet}^{[1]\prime}\right)\Svet'\right] ,
\]
which may be interpretated in two different, equivalent ways:
\begin{itemize}
	\item looking at the expression $\Rvet \left(\widehat{\Bvet}^{[1]\prime}\right)\Svet'$ from the right, one first computes
	the cross-sectional bottom-up reconciled forecasts of all the $n$ high-frequency time series, $\widetilde{\Yvet}_{cs(bu)}^{[1]} = \Svet \widehat{\Bvet}^{[1]}$. These forecasts, that are not temporally coherent with the base forecasts of the same series at lower temporal granularity (i.e., $\widehat{\Yvet}^{[k]}$, $k \ne 1$), are then reconciled \textit{via} temporal bottom-up,
	$\widetilde{\Yvet}_{ct(bu)}' = \Rvet \widetilde{\Yvet}_{cs(bu)}^{[1]\prime}$:
\[
\widehat{\Bvet}^{[1]} \quad \xrightarrow{\quad cs(bu) \quad} \quad \widetilde{\Yvet}_{cs(bu)}^{[1]} \quad \xrightarrow{\quad te(bu) \quad} \quad \widetilde{\Yvet}_{ct(bu)} ;
\]
	\item moving in turn from the left, one first computes the temporally reconciled forecasts of the $n_b$ bottom time series of the hierarchy,
	$\widetilde{\Bvet}_{te(bu)}' = \Rvet \left(\widehat{\Bvet}^{[1]\prime}\right)$.
	These forecasts, that are not cross-sectionally coherent with the base forecasts of the $n_a$ upper time series of the hierarchy, are then reconciled \textit{via} cross-sectional bottom-up,
	$\widetilde{\Yvet}_{ct(bu)} = \Svet \left(\widetilde{\Bvet}_{te(bu)}\right)$:
	\[
	\widehat{\Bvet}^{[1]} \quad \xrightarrow{\quad te(bu) \quad} \quad \widetilde{\Bvet}_{te(bu)} \quad \xrightarrow{\quad cs(bu) \quad}\quad \widetilde{\Yvet}_{ct(bu)} .
	\]
\end{itemize}
The two equivalent interpretations of the cross-temporal bottom-up approach are represented in Figure \ref{fig:ctbu_rep}, where $\widetilde{\Bvet}_{ct(bu)}^{[1]} \equiv \widehat{\Bvet}^{[1]}$,
$\widetilde{\Avet}_{ct(bu)}^{[1]} \equiv \widetilde{\Avet}_{cs(bu)}^{[1]}$, and
$\widetilde{\Bvet}_{ct(bu)}^{[k]} \equiv \widetilde{\Bvet}_{te(bu)}^{[k]}$, $k>1$.

\begin{figure}[H]
     \centering
     \begin{subfigure}[b]{1\textwidth}
	         \centering
	          \begin{tikzpicture}[>=latex, line width=1pt,
        Matrix/.style={
            matrix of nodes,
            font=\normalsize,
            align=center,
            text width = 1cm, 
            text height = 0.5cm, 
            column sep=4pt,
            row sep=7pt,
            nodes in empty cells,
            left delimiter={[},
            right delimiter={]}
        }]
    \matrix[Matrix] (M){ 
        |[opacity = 0.35]|$\widehat{\Avet}^{[m]}$ & |[opacity = 0.35]|\dots & |[opacity = 0.35]|$\widehat{\Avet}^{[k_2]}$ & $\widetilde{\Avet}_{cs(bu)}^{[1]}$ \\
        |[opacity = 0.35]|$\widehat{\Bvet}^{[m]}$ & |[opacity = 0.35]|\dots & |[opacity = 0.35]|$\widehat{\Bvet}^{[k_2]}$ & $\widehat{\Bvet}^{[1]}_{\phantom{.}}$ \\
    };
    \matrix[Matrix,  right = of M, xshift = 1.5cm] (M2){ 
        $\widetilde{\Avet}_{ct(bu)}^{[m]}$ & \dots & $\widetilde{\Avet}_{ct(bu)}^{[k_2]}$ & $\widetilde{\Avet}_{cs(bu)}^{[1]}$ \\
        $\widetilde{\Bvet}_{ct(bu)}^{[m]}$ & \dots & $\widetilde{\Bvet}_{ct(bu)}^{[k_2]}$ & $\widehat{\Bvet}^{[1]}_{\phantom{.}}$ \\
    };
    \draw[<-] (M.north east)++(0.4,0) coordinate (temp) -- (temp |- M.south) node [midway,label={[label distance=0.1cm,rotate=-90, xshift = 1.5mm, font=\footnotesize]Cross-sectional}]{};
        
    \draw[<-, opacity = 0.3] (M2.north east)++(0.4,0) coordinate (temp) -- (temp |- M2.south) node [midway,label={[label distance=0.1cm,rotate=-90, xshift = 1.5mm, font=\footnotesize]Cross-sectional}]{};
    \draw[<-] (M2.south west)++(0,-0.15) coordinate (temp) 
        -- (temp -| M2.east) node [midway,label={[label distance=0cm,xshift = 1.5mm, font=\footnotesize]below:Temporal}]{};
        
        \node[opacity=0.2,
            rounded corners,
            inner sep=0pt, fill = blue, fit=(M-2-4)(M-2-4)](Bct){};
        \node[opacity=0.2,
            rounded corners,
            inner sep=0pt, fill = green, fit=(M-1-4)(M-1-4)](Acs){};
        
        \node[opacity=0.2,
            rounded corners,
            inner sep=0pt, fill = blue, fit=(M2-2-4)(M2-2-4)](Bct){};
        \node[opacity=0.2,
            rounded corners,
            inner sep=0pt, fill = green, fit=(M2-1-4)(M2-1-4)](Acs){};
        \node[opacity=0.2,
            rounded corners,
            inner sep=0pt, fill = red, fit=(M2-1-1)(M2-2-3)](A){};
    \path (M) -- node[xshift = 0.5cm] {$\longrightarrow$}
              (M2);
    \end{tikzpicture}
	         \caption{$cs(bu)$ on $\widehat{\Yvet}^{[1]}$ $+$  $te(bu)$ on $\widetilde{\Yvet}_{cs(bu)}$.}
	         \label{fig:ctbu_rep1}
	     \end{subfigure}\vspace{0.5cm}
	 \begin{subfigure}[b]{1\textwidth}
	         \centering
	          \begin{tikzpicture}[>=latex, line width=1pt,
        Matrix/.style={
            matrix of nodes,
            font=\normalsize,
            align=center,
            text width = 1cm, 
            text height = 0.5cm, 
            column sep=4pt,
            row sep=7pt,
            nodes in empty cells,
            left delimiter={[},
            right delimiter={]}
        }]
    \matrix[Matrix] (M){ 
        |[opacity = 0.35]|$\widehat{\Avet}^{[m]}$ & |[opacity = 0.35]|\dots & |[opacity = 0.35]|$\widehat{\Avet}^{[k_2]}$ &  |[opacity = 0.35]|$\widehat{\Avet}^{[1]}$ \\
        $\widetilde{\Bvet}_{te(bu)}^{[m]}$ & \dots & $\widetilde{\Bvet}_{te(bu)}^{[k_2]}$ & $\widehat{\Bvet}^{[1]}_{\phantom{.}}$ \\
    };
    \matrix[Matrix,  right = of M, xshift = 1cm] (M2){ 
        $\widetilde{\Avet}_{ct(bu)}^{[m]}$ & \dots & $\widetilde{\Avet}_{ct(bu)}^{[k_2]}$ & $\widetilde{\Avet}_{ct(bu)}^{[1]}$ \\
        $\widetilde{\Bvet}_{te(bu)}^{[m]}$ & \dots & $\widetilde{\Bvet}_{te(bu)}^{[k_2]}$ & $\widehat{\Bvet}^{[1]}_{\phantom{.}}$ \\
    };
    \draw[<-] (M.south west)++(0,-0.15) coordinate (temp) 
        -- (temp -| M.east) node [midway,label={[label distance=0cm,xshift = 1.5mm, font=\footnotesize]below:Temporal}]{};
        
    \draw[<-] (M2.north east)++(0.4,0) coordinate (temp) -- (temp |- M2.south) node [midway,label={[label distance=0.1cm,rotate=-90, xshift = 1.5mm, font=\footnotesize]Cross-sectional}]{};
    \draw[<-, opacity = 0.3] (M2.south west)++(0,-0.15) coordinate (temp) 
        -- (temp -| M2.east) node [midway,label={[label distance=0cm,xshift = 1.5mm, font=\footnotesize]below:Temporal}]{};
        
        \node[opacity=0.2,
            rounded corners,
            inner sep=0pt, fill = blue, fit=(M-2-4)(M-2-4)](Bct){};
        \node[opacity=0.2,
            rounded corners,
            inner sep=0pt, fill = green, fit=(M-2-1)(M-2-3)](Acs){};
        
        \node[opacity=0.2,
            rounded corners,
            inner sep=0pt, fill = blue, fit=(M2-2-4)(M2-2-4)](Bct){};
        \node[opacity=0.2,
            rounded corners,
            inner sep=0pt, fill = green, fit=(M2-2-1)(M2-2-3)](Acs){};
        \node[opacity=0.2,
            rounded corners,
            inner sep=0pt, fill = red, fit=(M2-1-1)(M2-1-4)](A){};
    \path (M) -- node[xshift = 0cm] {$\longrightarrow$}
              (M2);
    \end{tikzpicture}
	         \caption{$te(bu)$ on $\widehat{\Bvet}^{[1]}$ $+$ $cs(bu)$ on $\widetilde{\Bvet}_{te(bu)}$.}
	         \label{fig:ctbu_rep2}
	     \end{subfigure}
        \caption{A `visual' interpretation of $\widetilde{\Yvet}_{ct(bu)}$ in two different, equivalent ways.}
        \label{fig:ctbu_rep}
\end{figure}
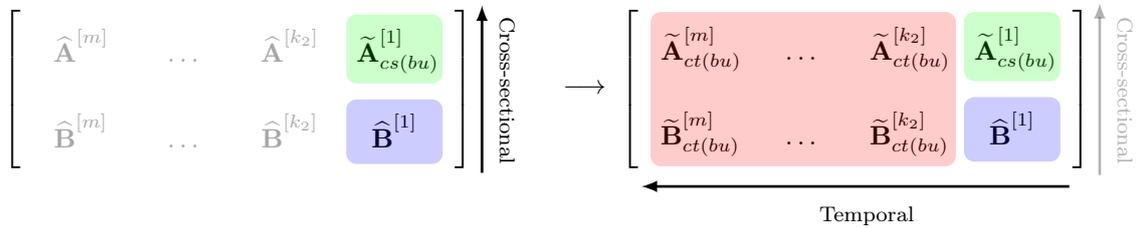
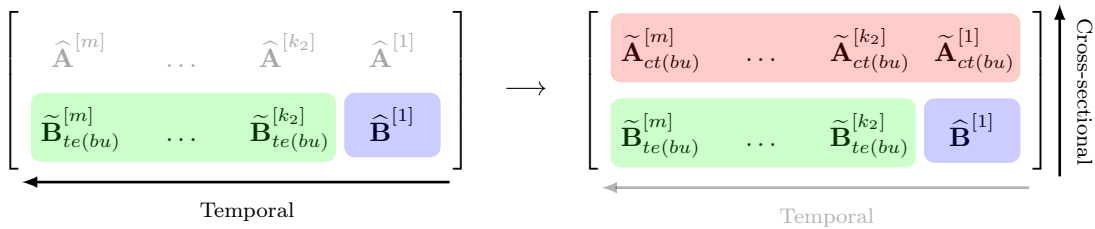

\section{Proof of the theorem \ref{thm:ite}}
\label{sec:proof}

	\begin{enumerate}
		\item Let $\Uvet'_{n_a \times n}$  and $\Zvet'_{m \times (k^\ast + m)}$ be the cross-sectional and temporal, respectively, zero constraints kernel matrix. Let
		$$
		\Mvet_{cs}^{[k]} = \Mvet_{cs} = \left[\Ivet_n - \Wvet\Uvet\left(\Uvet'\Wvet\Uvet\right)^{-1}\Uvet'\right], \quad k \in \mathcal{K}
		$$
		and
		$$
		\Mvet_{te,i} = \Mvet_{te} = \left[\Ivet_{k^\ast+m} - \Omegavet\Zvet\left(\Zvet'\Omegavet\Zvet\right)^{-1}\Zvet'\right], \quad i=1,\ldots,n ,
		$$
		be the cross-sectional and temporal projection matrices, respectively, such that:
		$$
		\begin{array}{rcll}
			\widetilde{\Yvet}_{cs}^{[k]} & = & \Mvet_{cs}\widehat{\Yvet}^{[k]}, & k \in \mathcal{K}, \\
			\widetilde{\yvet}_{te,i}      & = & \Mvet_{te}\widehat{\yvet}_i,      & i=1,\ldots,n,
		\end{array}
		$$
		where
$\widehat{\Yvet}^{[k]}$ and $\widetilde{\Yvet}_{cs}^{[k]}$ are the $\left( n \times \frac{m}{k}\right)$ matrices with, respectively, the base and the cross-sectional reconciled forecasts of the $n$ series at time granularity $k$, and
			 $\widehat{\yvet}_i$ and $\widetilde{\yvet}_{te,i}$ are the $\left[ (k^*+m) \times 1\right]$ vectors with, respectively, the base and the temporal reconciled forecasts at all time granularities for the $i$-th series.
		In compact matrix form we get:
		$$
		\begin{array}{rcl}
			\widetilde{\Yvet}_{cs} & = & \left[ \widetilde{\Yvet}_{cs}^{[m]} \; \widetilde{\Yvet}_{cs}^{[k_{p-1}]} \ldots \widetilde{\Yvet}_{cs}^{[k_2]} \; \widetilde{\Yvet}_{cs}^{[1]}
			\right] = \Mvet_{cs} \left[ \widehat{\Yvet}^{[m]} \; \widehat{\Yvet}^{[k_{p-1}]} \ldots \widehat{\Yvet}^{[k_2]} \; \widehat{\Yvet}^{[1]}
			\right] = \Mvet_{cs}\widehat{\Yvet}, \\
			\widetilde{\Yvet}_{te} & = & \left\{\Mvet_{te} \left[ \widehat{\yvet}_1 \ldots \widehat{\yvet}_i \ldots \widehat{\yvet}_n \right]\right\}' = \left\{\Mvet_{te}\widehat{\Yvet}'\right\}' .
		\end{array}
		$$
		
		Then, $\widetilde{\Yvet}_{cst}$ is obtained by first operating a cross-sectional reconciliation (step 1, $\widetilde{\Yvet}_{cs}$), and then reconciling via temporal hierarchies the forecasts obtained in the previous step (step 2, $\widetilde{\Yvet}_{cst}$). In matrix terms:
		$$
		\begin{array}{rcl}
			\text{step 1 } &
			\widetilde{\Yvet}_{cs} &= \Mvet_{cs}\widehat{\Yvet} \\
			\text{step 2 } &
			\widetilde{\Yvet}_{cst} &= \left\{\Mvet_{te} \widetilde{\Yvet}_{cs}'\right\}' = \widetilde{\Yvet}_{cs}\Mvet_{te}' = \Mvet_{cs}\widehat{\Yvet}\Mvet_{te}'
		\end{array}
		$$
		where $\widehat{\Yvet}$ is the ($n\times (k^\ast+m)$) base forecasts matrix.
		Viceversa, $\widetilde{\Yvet}_{tcs}$ is obtained by first operating a temporal reconciliation (step 1, $\widetilde{\Yvet}_{te}$), and finally a cross-sectional reconciliation is applied (step 2):
		$$
		\begin{array}{rcl}
			\text{step 1 } & \widetilde{\Yvet}_{te} & = \left\{\Mvet_{te} \widehat{\Yvet}'\right\}' =\widehat{\Yvet}\Mvet_{te}'\\
			\text{step 2 } & \widetilde{\Yvet}_{tcs} & = \Mvet_{cs} \widetilde{\Yvet}_{te} =\Mvet_{cs} \widehat{\Yvet}\Mvet_{te}' = \widetilde{\Yvet}_{cst} .
		\end{array}
		$$
The forecasts $\widetilde{\Yvet}_{tcs}$ (and $\widetilde{\Yvet}_{cst}$) (a single two-step iteration) are cross-temporally reconciled because:
		$$
		\begin{aligned}
			\Uvet'\widetilde{\Yvet}_{tcs} &= \Uvet'\Mvet_{cs} \widehat{\Yvet}\Mvet_{te}'= \left[\Uvet' - \Uvet'\Wvet\Uvet\left(\Uvet'\Wvet\Uvet\right)^{-1}\Uvet'\right]  \widehat{\Yvet}\Mvet_{te}'=\left[\Uvet' - \Uvet'\right]  \widehat{\Yvet}\Mvet_{te}' = \Zerovet_{n_a \times (k^\ast+ m)}
		\end{aligned}
		$$
		and
		$$
		\begin{aligned}
			\Zvet'\widetilde{\Yvet}_{tcs}' &= \Zvet'\Mvet_{te} \widehat{\Yvet}'\Mvet_{cs}' = \left[\Zvet' - \Zvet'\Omegavet\Zvet\left(\Zvet'\Omegavet\Zvet\right)^{-1}\Zvet'\right] \Mvet_{cs}'=\left[\Zvet' - \Zvet'\right] \widehat{\Yvet}'\Mvet_{cs}'= \Zerovet_{m \times n}.
		\end{aligned}
		$$
		
		\item Using matrix vectorization and structural representation for the cross-temporal reconciliation problem we have: 
		$$
		\begin{aligned}
			\widetilde{\yvet}_{oct} = \Fvet \left(\Fvet' \Omegavet_{ct}^{-1}\Fvet\right)^{-1}\Fvet'\widehat{\yvet}
		\end{aligned} ,
		$$
		where $\widetilde{\yvet}_{oct} = \mathrm{vec}\left(\widetilde{\Yvet}_{oct}'\right)$, $\widehat{\yvet} = \mathrm{vec}\left(\widehat{\Yvet}_{oct}'\right)$, 
		and $\Fvet = \Svet \otimes \Rvet$ is the [$n(k^\ast +m)\times mn_b$] cross-temporal summing matrix. Let $\Omegavet_{ct} = \Wvet \otimes \Omegavet$ be the cross-temporal error covariance matrix, by exploiting the properties of the Kronecker product (\citealp{Henderson1981}; \citealp{Harville2008}), we obtain:
		$$
		\begin{aligned}
			\widetilde{\yvet}_{oct} &= \Fvet \left(\Fvet' \Omegavet_{ct}^{-1}\Fvet\right)^{-1}\Fvet'\Omegavet_{ct}^{-1}\widehat{\yvet} \\
			&= (\Svet \otimes \Rvet)\left[(\Svet \otimes \Rvet)' (\Wvet \otimes \Omegavet)^{-1} (\Svet \otimes \Rvet)\right]^{-1}(\Svet \otimes \Rvet)'(\Wvet \otimes \Omegavet)^{-1}\widehat{\yvet}\\
			& = (\Svet \otimes \Rvet)\left[(\Svet' \otimes \Rvet') (\Wvet^{-1}  \otimes \Omegavet^{-1} )(\Svet \otimes \Rvet)\right]^{-1}(\Svet' \otimes \Rvet')(\Wvet^{-1} \otimes \Omegavet^{-1})\widehat{\yvet}\\
			& = (\Svet \otimes \Rvet)\left(\Svet'\Wvet^{-1}\Svet \otimes \Rvet'\Omegavet^{-1}\Rvet\right)^{-1}(\Svet'\Wvet^{-1} \otimes \Rvet'\Omegavet^{-1})\widehat{\yvet}\\
			& = \left[\Svet\left(\Svet'\Wvet^{-1}\Svet\right)^{-1}\Svet'\Wvet^{-1} \otimes \Rvet\left(\Rvet'\Omegavet^{-1}\Rvet\right)^{-1}\Rvet'\Omegavet^{-1}\right]\widehat{\yvet}.
		\end{aligned}
		$$
		Exploiting the structural notation also in the \textit{temporal then cross-sectional} iterative procedure and the vectorization $\widetilde{\yvet}_{tcs} = \mathrm{vec}\left(\widehat{\Yvet}_{tcs}'\right)$, we obtain:
		$$
		\begin{aligned}
			\widetilde{\yvet}_{tcs} &= \Pvet \left[\Ivet_{k^\ast + m} \otimes \Svet\left(\Svet'\Wvet^{-1}\Svet\right)^{-1}\Svet'\Wvet^{-1}\right] \Pvet' \left[\Ivet_{n} \otimes \Rvet \left(\Rvet'\Omegavet^{-1}\Rvet\right)^{-1}\Rvet'\Omegavet^{-1}\right]\widehat{\yvet}\\
			&= \left[\Svet\left(\Svet'\Wvet^{-1}\Svet\right)^{-1}\Svet'\Wvet^{-1} \otimes \Ivet_{k^\ast + m}\right] \left[\Ivet_{n} \otimes \Rvet \left(\Rvet'\Omegavet^{-1}\Rvet\right)^{-1}\Rvet'\Omegavet^{-1}\right]\widehat{\yvet} \\
			&=\left[\Svet\left(\Svet'\Wvet^{-1}\Svet\right)^{-1}\Svet'\Wvet^{-1} \otimes \Rvet \left(\Rvet'\Omegavet^{-1}\Rvet\right)^{-1}\Rvet'\Omegavet^{-1}\right]\widehat{\yvet} = \widetilde{\yvet}_{oct}.
		\end{aligned}
		$$
		Finally, $\widetilde{\Yvet}_{tcs} = \widetilde{\Yvet}_{oct}$ and, following from point 1, $\widetilde{\Yvet}_{tcs} = \widetilde{\Yvet}_{cst} = \widetilde{\Yvet}_{oct}$.
	\end{enumerate}
\hfill $\Box$

\subsection*{The heuristic KA with constant cross-sectional and temporal error covariance matrices}
In the \cite{Kourentzes2019} approach (section \ref{sec:heuite}), the temporally reconciled predictions at step 1, $\widehat{\Yvet}\Mvet_{te}'$, are cross-sectionally reconciled via premultiplication by the matrix $\overline{\Mvet} = \frac{1}{p} \sum_{k \in \mathcal{K}} \Mvet^{[k]}_{cs}$.
 As  $\Wvet^{[k]} = \Wvet$, $k \in \mathcal{K}$, 
 $\Mvet^{[k]}_{cs} = \Mvet_{cs}$, and  $\overline{\Mvet}=\Mvet_{cs}$.
 Then,
$$
\widetilde{\Yvet}_{KA} = \overline{\Mvet}\widehat{\Yvet}\Mvet_{te}' = 
\Mvet_{cs}\widehat{\Yvet}\Mvet_{te}' =
\widetilde{\Yvet}_{cst}= \widetilde{\Yvet}_{tcs}.
$$

%


\bibliography{mybibfile}

\section*{Acknowldgements}
\vspace{-.3cm}

\noindent \textit{Parts of this research were carried in the frame of the PRIN2017 project “HiDEA: Advanced Econometrics for High-frequency Data”, 2017RSMPZZ.}

\clearpage
\setcounter{figure}{0}
\setcounter{table}{0}
\setcounter{equation}{0}
\renewcommand\theequation{OA.\arabic{equation}}  
\renewcommand\thefigure{OA.\arabic{figure}}  
\renewcommand\thetable{OA.\arabic{table}}

{\huge On-line appendix: supplementary tables and graphs}
\addcontentsline{toc}{section}{On-line appendix: supplementary tables and graphs}
\label{AuxTab}


\begin{center}
\begin{footnotesize}
\begin{longtable}{rlccccc}
\caption{Number of series with at least one negative reconciled forecast.
		All temporal aggregation orders, forecast horizon: operating day.}
	\label{table_nn_Yagli_original_appendix} \\
\toprule
k & comb & n & min & max & fmin & fmax\\
\midrule
\endfirsthead

\multicolumn{7}{c}%
{{\bfseries \tablename\ \thetable{} -- continued from previous page}} \\
\toprule
k & comb & n & min & max & fmin & fmax\\
\midrule
\endhead

\bottomrule \multicolumn{7}{r}{{Continued on next page}} \\ 
\endfoot

\bottomrule
\endlastfoot

1 & base & 350 & 4 & 6 & -15.617 & 0.000\\*
1 & oct($ols$) & 350 & 109 & 324 & -69.165 & 0.000\\*
1 & ite($struc_{cs}, ols_{te}$) & 350 & 137 & 324 & -20.739 & 0.000\\*
1 & ite($ols_{cs}, struc_{te}$) & 350 & 91 & 324 & -17.961 & 0.000\\*
1 & oct($struc$) & 350 & 89 & 324 & -14.292 & 0.000\\*
1 & te($ols_2$)+cs($ols$) & 350 & 42 & 324 & -12.993 & 0.000\\*
1 & te($struc_2$)+cs($ols$) & 350 & 36 & 324 & -12.994 & 0.000\\*
1 & te($ols_2$)+cs($struc$) & 350 & 144 & 324 & -7.526 & 0.000\\*
1 & te($struc_2$)+cs($struc$) & 350 & 124 & 324 & -7.642 & 0.000\\*
\midrule
2 & base & 350 & 171 & 320 & -32.639 & 0\\*
2 & oct($ols$) & 350 & 60 & 324 & -32.147 & 0\\*
2 & ite($struc_{cs}, ols_{te}$) & 350 & 56 & 324 & -31.867 & 0\\*
2 & ite($ols_{cs}, struc_{te}$) & 350 & 26 & 324 & -18.728 & 0\\*
2 & oct($struc$) & 350 & 33 & 324 & -22.719 & 0\\*
2 & te($ols_2$)+cs($ols$) & 350 & 29 & 324 & -27.023 & 0\\*
2 & te($struc_2$)+cs($ols$) & 350 & 24 & 324 & -27.025 & 0\\*
2 & te($ols_2$)+cs($struc$) & 350 & 82 & 321 & -20.451 & 0\\*
2 & te($struc_2$)+cs($struc$) & 350 & 58 & 322 & -17.812 & 0\\*
\midrule
3 & base & 350 & 60 & 323 & -45.120 & 0\\*
3 & oct($ols$) & 350 & 6 & 324 & -31.640 & 0\\*
3 & ite($struc_{cs}, ols_{te}$) & 350 & 9 & 324 & -35.059 & 0\\*
3 & ite($ols_{cs}, struc_{te}$) & 350 & 20 & 324 & -28.085 & 0\\*
3 & oct($struc$) & 350 & 27 & 323 & -32.377 & 0\\*
3 & te($ols_2$)+cs($ols$) & 350 & 3 & 324 & -43.853 & 0\\*
3 & te($struc_2$)+cs($ols$) & 350 & 1 & 324 & -43.855 & 0\\*
3 & te($ols_2$)+cs($struc$) & 350 & 40 & 324 & -42.255 & 0\\*
3 & te($struc_2$)+cs($struc$) & 350 & 18 & 324 & -42.445 & 0\\*
\midrule
4 & base & 350 & 14 & 321 & -88.776 & 0\\*
4 & oct($ols$) & 350 & 11 & 322 & -42.291 & 0\\*
4 & ite($ols_{cs}, struc_{te}$) & 350 & 15 & 322 & -35.129 & 0\\*
4 & ite($struc_{cs}, ols_{te}$) & 350 & 25 & 318 & -44.256 & 0\\*
4 & oct($struc$) & 350 & 14 & 324 & -40.562 & 0\\*
4 & te($ols_2$)+cs($ols$) & 350 & 2 & 324 & -85.947 & 0\\*
4 & te($struc_2$)+cs($ols$) & 350 & 1 & 324 & -85.956 & 0\\*
4 & te($ols_2$)+cs($struc$) & 350 & 2 & 324 & -63.026 & 0\\*
4 & te($struc_2$)+cs($struc$) & 350 & 1 & 324 & -63.927 & 0\\*
\midrule
6 & base & 350 & 19 & 321 & -146.702 & 0\\*
6 & oct($ols$) & 350 & 3 & 324 & -61.415 & 0\\*
6 & ite($struc_{cs}, ols_{te}$) & 350 & 4 & 323 & -67.327 & 0\\*
6 & ite($ols_{cs}, struc_{te}$) & 350 & 17 & 323 & -53.857 & 0\\*
6 & oct($struc$) & 350 & 16 & 322 & -63.281 & 0\\*
6 & te($ols_2$)+cs($ols$) & 347 & 0 & 324 & -133.572 & 0\\*
6 & te($struc_2$)+cs($ols$) & 347 & 0 & 324 & -133.571 & 0\\*
6 & te($ols_2$)+cs($struc$) & 347 & 0 & 324 & -89.955 & 0\\*
6 & te($struc_2$)+cs($struc$) & 349 & 0 & 324 & -88.924 & 0\\*
\midrule
8 & base & 212 & 0 & 238 & -866.818 & 0.000\\*
8 & oct($ols$) & 300 & 0 & 305 & -78.762 & 0.000\\*
8 & oct($struc$) & 276 & 0 & 324 & -51.806 & 0.000\\*
8 & ite($struc_{cs}, ols_{te}$) & 296 & 0 & 272 & -53.688 & 0.000\\*
8 & ite($ols_{cs}, struc_{te}$) & 274 & 0 & 319 & -42.734 & 0.000\\*
8 & te($ols_2$)+cs($ols$) & 307 & 0 & 320 & -680.766 & 0.000\\*
8 & te($struc_2$)+cs($ols$) & 285 & 0 & 324 & -680.488 & 0.000\\*
8 & te($ols_2$)+cs($struc$) & 300 & 0 & 241 & -56.011 & 0.000\\*
8 & te($struc_2$)+cs($struc$) & 284 & 0 & 268 & -52.341 & 0.000\\*
\midrule
12 & base & 46 & 0 & 57 & -59.405 & -0.001\\*
12 & oct($ols$) & 14 & 0 & 78 & -15.328 & -0.001\\*
12 & ite($struc_{cs}, ols_{te}$) & 13 & 0 & 45 & -7.949 & -0.004\\*
12 & ite($ols_{cs}, struc_{te}$) & 8 & 0 & 31 & -6.672 & -0.001\\*
12 & oct($struc$) & 6 & 0 & 18 & -0.537 & -0.009\\*
12 & te($ols_2$)+cs($ols$) & 35 & 0 & 95 & -32.259 & 0.000\\*
12 & te($struc_2$)+cs($ols$) & 26 & 0 & 100 & -32.191 & 0.000\\*
12 & te($ols_2$)+cs($struc$) & 19 & 0 & 75 & -9.799 & 0.000\\*
12 & te($struc_2$)+cs($struc$) & 12 & 0 & 69 & -6.544 & -0.002\\*
\midrule
24 & base & 11 & 0 & 35 & -51.205 & -0.006\\*
24 & oct($ols$) & 11 & 0 & 60 & -29.615 & -0.001\\*
24 & ite($struc_{cs}, ols_{te}$) & 11 & 0 & 33 & -10.363 & 0.000\\*
24 & ite($ols_{cs}, struc_{te}$) & 6 & 0 & 28 & -8.915 & -0.008\\*
24 & oct($struc$) & 4 & 0 & 11 & -1.020 & -0.026\\*
24 & te($ols_2$)+cs($ols$) & 10 & 0 & 72 & -15.037 & -0.018\\*
24 & te($struc_2$)+cs($ols$) & 10 & 0 & 71 & -14.791 & -0.005\\*
24 & te($ols_2$)+cs($struc$) & 10 & 0 & 43 & -12.969 & -0.005\\*
24 & te($struc_2$)+cs($struc$) & 5 & 0 & 31 & -7.042 & -0.002\\*

\end{longtable}
\end{footnotesize}
\end{center}


\begin{center}
\begin{footnotesize}
\begin{longtable}{llcccccccc}
	\caption{Number of replications with at least one negative reconciled forecast.
		All temporal aggregation orders, forecast horizon: operating day.}
	\label{table_nnrep_Yagli_original_appendix} \\

\toprule
$k$ & comb & \#$\mathcal{L}_1$ & \#$\mathcal{L}_2$ & \#$\mathcal{L}_3$ & \#tot & \#$\mathcal{L}_1$ & \#$\mathcal{L}_2$ & \#$\mathcal{L}_3$ & \%tot\\
\midrule
\endfirsthead

\multicolumn{10}{c}%
{{\bfseries \tablename\ \thetable{} -- continued from previous page}} \\
\toprule
$k$ & comb & \#$\mathcal{L}_1$ & \#$\mathcal{L}_2$ & \#$\mathcal{L}_3$ & \#tot & \#$\mathcal{L}_1$ & \#$\mathcal{L}_2$ & \#$\mathcal{L}_3$ & \%tot\\
\midrule
\endhead

\bottomrule \multicolumn{10}{r}{{Continued on next page}} \\ 
\endfoot

\bottomrule
\endlastfoot

1 & base & 1 & 5 & 0 & 6 & 100.00\% & 100.00\% & 0.00\% & 1.85\%\\*
1 & oct($ols$) & 1 & 5 & 318 & 324 & 100.00\% & 100.00\% & 100.00\% & 100.00\%\\*
1 & ite($struc_{cs}, ols_{te}$ & 1 & 5 & 318 & 324 & 100.00\% & 100.00\% & 100.00\% & 100.00\%\\*
1 & ite($ols_{cs}, struc_{te}$) & 1 & 5 & 318 & 324 & 100.00\% & 100.00\% & 100.00\% & 100.00\%\\*
1 & oct($struc$) & 1 & 5 & 318 & 324 & 100.00\% & 100.00\% & 100.00\% & 100.00\%\\*
1 & te($ols_2$)+cs($ols$) & 1 & 5 & 318 & 324 & 100.00\% & 100.00\% & 100.00\% & 100.00\%\\*
1 & te($struc_2$)+cs($ols$) & 1 & 5 & 318 & 324 & 100.00\% & 100.00\% & 100.00\% & 100.00\%\\*
1 & te($ols_2$)+cs($struc$) & 1 & 5 & 318 & 324 & 100.00\% & 100.00\% & 100.00\% & 100.00\%\\*
1 & te($struc_2$)+cs($struc$) & 1 & 5 & 318 & 324 & 100.00\% & 100.00\% & 100.00\% & 100.00\%\\*
\midrule
2 & base & 1 & 5 & 318 & 324 & 100.00\% & 100.00\% & 100.00\% & 100.00\%\\*
2 & oct($ols$) & 1 & 5 & 318 & 324 & 100.00\% & 100.00\% & 100.00\% & 100.00\%\\*
2 & ite($struc_{cs}, ols_{te}$ & 1 & 5 & 318 & 324 & 100.00\% & 100.00\% & 100.00\% & 100.00\%\\*
2 & ite($ols_{cs}, struc_{te}$) & 1 & 5 & 318 & 324 & 100.00\% & 100.00\% & 100.00\% & 100.00\%\\*
2 & oct($struc$) & 1 & 5 & 318 & 324 & 100.00\% & 100.00\% & 100.00\% & 100.00\%\\*
2 & te($ols_2$)+cs($ols$) & 1 & 5 & 318 & 324 & 100.00\% & 100.00\% & 100.00\% & 100.00\%\\*
2 & te($ols_2$)+cs($struc$) & 1 & 5 & 318 & 324 & 100.00\% & 100.00\% & 100.00\% & 100.00\%\\*
2 & te($struc_2$)+cs($ols$) & 1 & 5 & 318 & 324 & 100.00\% & 100.00\% & 100.00\% & 100.00\%\\*
2 & te($struc_2$)+cs($struc$) & 1 & 5 & 318 & 324 & 100.00\% & 100.00\% & 100.00\% & 100.00\%\\*
\midrule
3 & base & 1 & 5 & 318 & 324 & 100.00\% & 100.00\% & 100.00\% & 100.00\%\\*
3 & oct($ols$) & 1 & 5 & 318 & 324 & 100.00\% & 100.00\% & 100.00\% & 100.00\%\\*
3 & ite($struc_{cs}, ols_{te}$ & 1 & 5 & 318 & 324 & 100.00\% & 100.00\% & 100.00\% & 100.00\%\\*
3 & ite($ols_{cs}, struc_{te}$) & 1 & 5 & 318 & 324 & 100.00\% & 100.00\% & 100.00\% & 100.00\%\\*
3 & oct($struc$) & 1 & 5 & 318 & 324 & 100.00\% & 100.00\% & 100.00\% & 100.00\%\\*
3 & te($ols_2$)+cs($ols$) & 1 & 5 & 318 & 324 & 100.00\% & 100.00\% & 100.00\% & 100.00\%\\*
3 & te($struc_2$)+cs($ols$) & 1 & 5 & 318 & 324 & 100.00\% & 100.00\% & 100.00\% & 100.00\%\\*
3 & te($ols_2$)+cs($struc$) & 1 & 5 & 318 & 324 & 100.00\% & 100.00\% & 100.00\% & 100.00\%\\*
3 & te($struc_2$)+cs($struc$) & 1 & 5 & 318 & 324 & 100.00\% & 100.00\% & 100.00\% & 100.00\%\\*
\midrule
4 & base & 1 & 5 & 318 & 324 & 100.00\% & 100.00\% & 100.00\% & 100.00\%\\*
4 & oct($ols$) & 1 & 5 & 318 & 324 & 100.00\% & 100.00\% & 100.00\% & 100.00\%\\*
4 & ite($struc_{cs}, ols_{te}$ & 1 & 5 & 318 & 324 & 100.00\% & 100.00\% & 100.00\% & 100.00\%\\*
4 & ite($ols_{cs}, struc_{te}$) & 1 & 5 & 318 & 324 & 100.00\% & 100.00\% & 100.00\% & 100.00\%\\*
4 & oct($struc$) & 1 & 5 & 318 & 324 & 100.00\% & 100.00\% & 100.00\% & 100.00\%\\*
4 & te($ols_2$)+cs($ols$) & 1 & 5 & 318 & 324 & 100.00\% & 100.00\% & 100.00\% & 100.00\%\\*
4 & te($struc_2$)+cs($ols$) & 1 & 5 & 318 & 324 & 100.00\% & 100.00\% & 100.00\% & 100.00\%\\*
4 & te($ols_2$)+cs($struc$) & 1 & 5 & 318 & 324 & 100.00\% & 100.00\% & 100.00\% & 100.00\%\\*
4 & te($struc_2$)+cs($struc$) & 1 & 5 & 318 & 324 & 100.00\% & 100.00\% & 100.00\% & 100.00\%\\*
\midrule
6 & base & 1 & 5 & 318 & 324 & 100.00\% & 100.00\% & 100.00\% & 100.00\%\\*
6 & oct($ols$) & 1 & 5 & 318 & 324 & 100.00\% & 100.00\% & 100.00\% & 100.00\%\\*
6 & ite($struc_{cs}, ols_{te}$ & 1 & 5 & 318 & 324 & 100.00\% & 100.00\% & 100.00\% & 100.00\%\\*
6 & ite($ols_{cs}, struc_{te}$) & 1 & 5 & 318 & 324 & 100.00\% & 100.00\% & 100.00\% & 100.00\%\\*
6 & oct($struc$) & 1 & 5 & 318 & 324 & 100.00\% & 100.00\% & 100.00\% & 100.00\%\\*
6 & te($ols_2$)+cs($ols$) & 1 & 5 & 318 & 324 & 100.00\% & 100.00\% & 100.00\% & 100.00\%\\*
6 & te($struc_2$)+cs($ols$) & 1 & 5 & 318 & 324 & 100.00\% & 100.00\% & 100.00\% & 100.00\%\\*
6 & te($ols_2$)+cs($struc$) & 1 & 5 & 318 & 324 & 100.00\% & 100.00\% & 100.00\% & 100.00\%\\*
6 & te($struc_2$)+cs($struc$) & 1 & 5 & 318 & 324 & 100.00\% & 100.00\% & 100.00\% & 100.00\%\\*
\midrule
8 & base & 1 & 3 & 318 & 322 & 100.00\% & 60.00\% & 100.00\% & 99.38\%\\*
8 & oct($ols$) & 1 & 5 & 318 & 324 & 100.00\% & 100.00\% & 100.00\% & 100.00\%\\*
8 & ite($struc_{cs}, ols_{te}$ & 1 & 5 & 318 & 324 & 100.00\% & 100.00\% & 100.00\% & 100.00\%\\*
8 & ite($ols_{cs}, struc_{te}$) & 1 & 5 & 318 & 324 & 100.00\% & 100.00\% & 100.00\% & 100.00\%\\*
8 & oct($struc$) & 1 & 5 & 318 & 324 & 100.00\% & 100.00\% & 100.00\% & 100.00\%\\*
8 & te($ols_2$)+cs($ols$) & 1 & 5 & 318 & 324 & 100.00\% & 100.00\% & 100.00\% & 100.00\%\\*
8 & te($struc_2$)+cs($ols$) & 1 & 5 & 318 & 324 & 100.00\% & 100.00\% & 100.00\% & 100.00\%\\*
8 & te($ols_2$)+cs($struc$) & 1 & 5 & 318 & 324 & 100.00\% & 100.00\% & 100.00\% & 100.00\%\\*
8 & te($struc_2$)+cs($struc$) & 1 & 5 & 318 & 324 & 100.00\% & 100.00\% & 100.00\% & 100.00\%\\*
\midrule
12 & base & 0 & 2 & 135 & 137 & 0.00\% & 40.00\% & 42.45\% & 42.28\%\\*
12 & oct($ols$) & 0 & 1 & 162 & 163 & 0.00\% & 20.00\% & 50.94\% & 50.31\%\\*
12 & ite($struc_{cs}, ols_{te}$ & 0 & 1 & 99 & 100 & 0.00\% & 20.00\% & 31.13\% & 30.86\%\\*
12 & ite($ols_{cs}, struc_{te}$) & 0 & 1 & 67 & 68 & 0.00\% & 20.00\% & 21.07\% & 20.99\%\\*
12 & oct($struc$) & 0 & 0 & 28 & 28 & 0.00\% & 0.00\% & 8.81\% & 8.64\%\\*
12 & te($ols_2$)+cs($ols$) & 0 & 2 & 172 & 174 & 0.00\% & 40.00\% & 54.09\% & 53.70\%\\*
12 & te($struc_2$)+cs($ols$) & 0 & 2 & 156 & 158 & 0.00\% & 40.00\% & 49.06\% & 48.77\%\\*
12 & te($ols_2$)+cs($struc$) & 0 & 1 & 126 & 127 & 0.00\% & 20.00\% & 39.62\% & 39.20\%\\*
12 & te($struc_2$)+cs($struc$) & 0 & 1 & 105 & 106 & 0.00\% & 20.00\% & 33.02\% & 32.72\%\\*
\midrule
24 & base & 0 & 2 & 58 & 60 & 0.00\% & 40.00\% & 18.24\% & 18.52\%\\*
24 & oct($ols$) & 0 & 1 & 134 & 135 & 0.00\% & 20.00\% & 42.14\% & 41.67\%\\*
24 & ite($struc_{cs}, ols_{te}$ & 0 & 1 & 72 & 73 & 0.00\% & 20.00\% & 22.64\% & 22.53\%\\*
24 & ite($ols_{cs}, struc_{te}$) & 0 & 1 & 47 & 48 & 0.00\% & 20.00\% & 14.78\% & 14.81\%\\*
24 & oct($struc$) & 0 & 0 & 17 & 17 & 0.00\% & 0.00\% & 5.35\% & 5.25\%\\*
24 & te($ols_2$)+cs($ols$) & 0 & 1 & 121 & 122 & 0.00\% & 20.00\% & 38.05\% & 37.65\%\\*
24 & te($struc_2$)+cs($ols$) & 0 & 1 & 102 & 103 & 0.00\% & 20.00\% & 32.08\% & 31.79\%\\*
24 & te($ols_2$)+cs($struc$) & 0 & 1 & 80 & 81 & 0.00\% & 20.00\% & 25.16\% & 25.00\%\\*
24 & te($struc_2$)+cs($struc$) & 0 & 1 & 41 & 42 & 0.00\% & 20.00\% & 12.89\% & 12.96\%\\*

\end{longtable}	
\end{footnotesize}
\end{center}

\begin{figure}[H]
	\centering
	\includegraphics[width=0.9\linewidth]{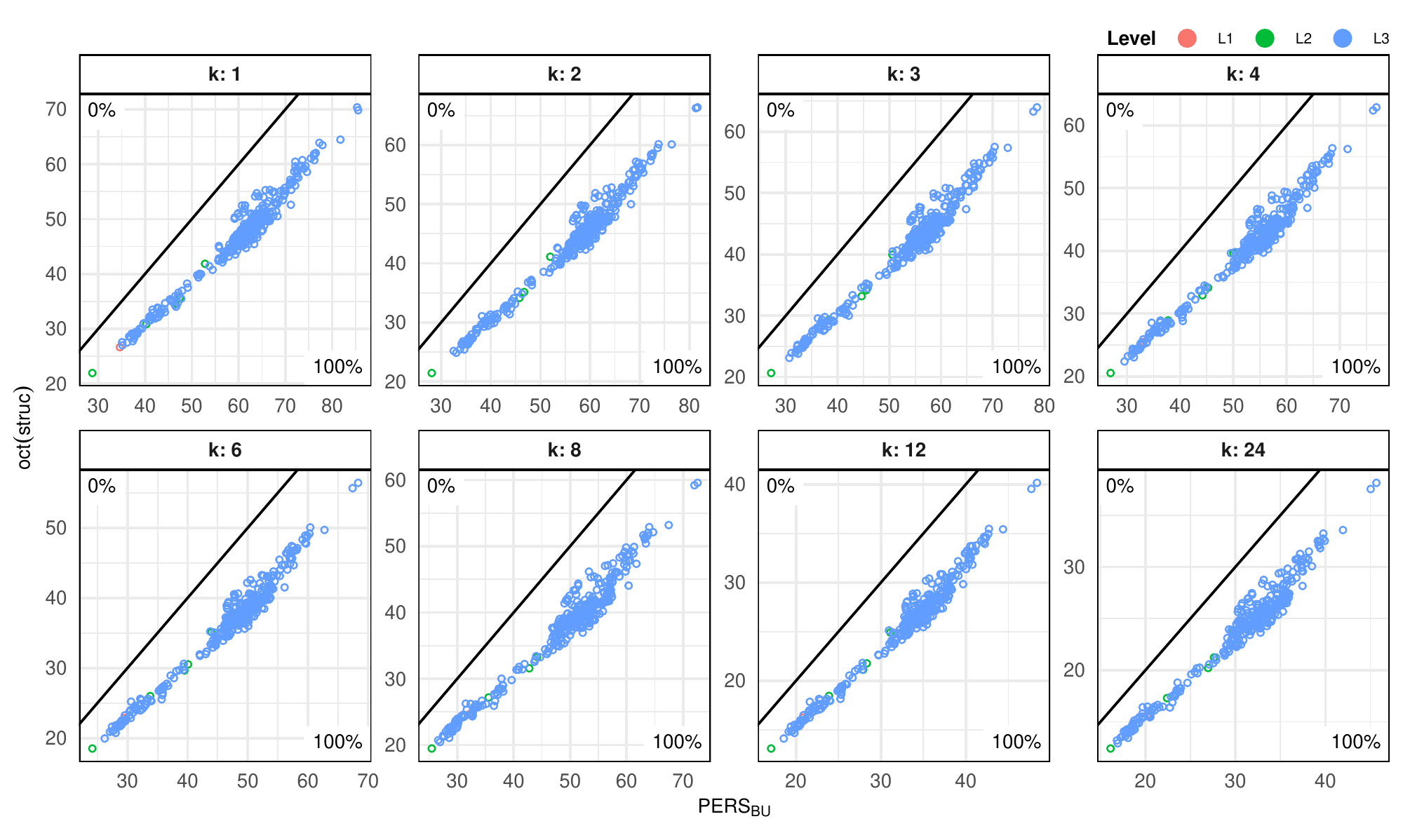}
	\includegraphics[width=0.9\linewidth]{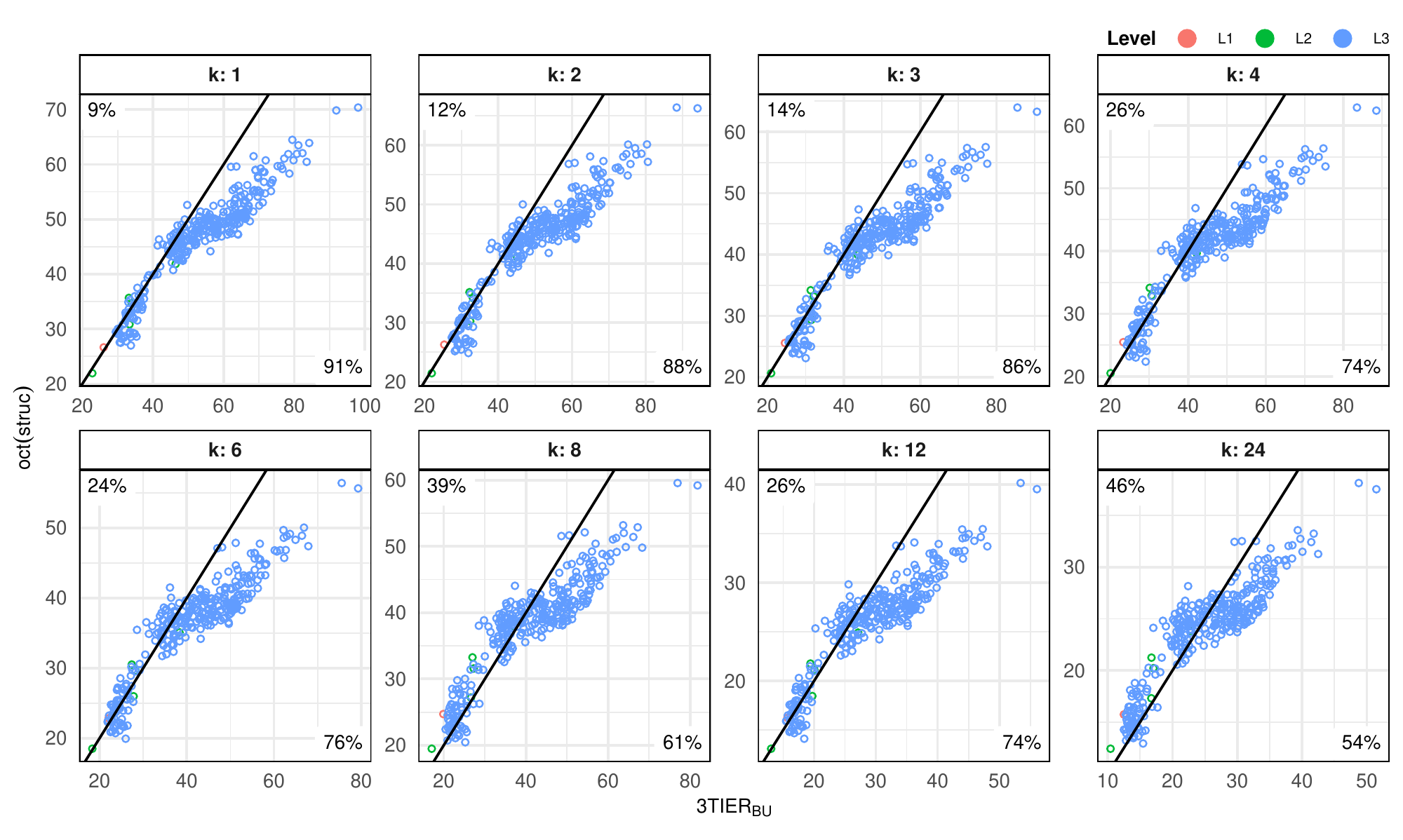}
	\caption{Comparison of nRMSE(\%) between PERS$_{BU}$ and not negative oct-struc (top panel), and between 3TIER$_{BU}$ and not negative oct-struc (bottom panel). The black line represents the bisector, where the nRMSE's for both approaches are equal. On the top-left (bottom-right) corner of each graph, the percentage of points above (below) the bisector is reported.}
	\label{fig:comparison_ct-struc_vs_PERS_3TIER}
\end{figure}

%
%
%
%

\newpage

A general overview of the accuracy of the considered procedures at any tempural granularity is given by table \ref{table_Yagli_original_appendix}, while a visual confirmation of the performance of ct-struc is shown by the MCB graphs in figure 5 in the main paper. 
The ct-struc daily reconciled forecasts show the best performance, significantly different from all the other procedures, while at hourly level it ranks second, the best being t-struc$_2$ + cs-struc. However, t-struc$_2$ + cs-struc produces temporally incoherent forecasts, and a punctual assessment of the performance of these two procedures in terms of nRMSE (figure \ref{fig:comparison_ct-struc_vs_tcs-struc2-struc}) confirms that their forecasting accuracy is about the same.

\begin{table}[H]
	\caption{Forecast accuracies of sequential reconciliation methods and base forecasts in terms of nRMSE(\%). Free reconciliation approaches considered by Yagli et al. (2019), Tables 2, 3, p. 395.
		All temporal aggregation orders, forecast horizon: operating day.}
	\label{table_Yagli_original_appendix}
	\begin{center}
		\begin{footnotesize}
		\begin{tabular}[t]{llcccccccc}
\toprule
$\mathcal{L}$ & comb & 1 & 2 & 3 & 4 & 6 & 8 & 12 & 24\\
\midrule
1 & PERS$_{bu}$ & 34.62 & 34.14 & 33.35 & 33.11 & 29.45 & 32.23 & 20.83 & 20.23\\
1 & 3TIER$_{bu}$ & \textbf{26.03} & \textbf{25.34} & \textbf{24.56} & \textbf{23.31} & \textbf{22.09} & \textbf{19.89} & \textbf{15.70} & \textbf{12.48}\\
1 & base & 27.85 & 27.55 & 28.31 & 29.69 & 28.01 & 35.32 & 21.04 & 18.17\\
1 &oct($ols$)& 30.69 & 29.96 & 28.29 & 28.86 & 25.26 & 28.12 & 18.39 & 17.91\\
1 &ite($struc_{cs}, ols_{te}$) & 28.74 & 28.25 & 27.14 & 27.38 & 24.34 & 26.64 & 17.82 & 17.33\\
1 & ite($ols_{cs}, struc_{te}$) & 28.26 & 27.84 & 27.09 & 27.00 & 24.17 & 26.24 & 17.46 & 16.96\\
1 & oct($struc$) & 26.71 & 26.35 & 25.67 & 25.57 & 23.03 & 24.80 & 16.74 & 16.24\\
1 & te($ols_2$)+cs($ols$) & 27.80 & 27.41 & 27.91 & 29.38 & 27.39 & 34.11 & 20.80 & 17.96\\
1 &  te($struc_2$)+cs($ols$) & 27.80 & 27.41 & 27.90 & 29.37 & 27.39 & 34.10 & 20.80 & 17.95\\
1 & te($ols_2$)+cs($struc$) & 27.02 & 26.73 & 26.73 & 27.75 & 25.13 & 28.93 & 19.11 & 17.24\\
1 & te($struc_2$)+cs($struc$) & 26.17 & 25.86 & 25.83 & 26.80 & 24.37 & 27.95 & 18.58 & 16.68\\
\midrule
2 & PERS$_{bu}$ & 43.15 & 42.42 & 41.27 & 40.81 & 36.29 & 39.19 & 25.66 & 24.57\\
2 & 3TIER$_{bu}$ & 33.95 & 32.94 & 31.90 & \textbf{30.62} & \textbf{28.02} & \textbf{26.99} & \textbf{19.87} & \textbf{16.75}\\
2 & base & 34.24 & 34.04 & 33.79 & 35.01 & 33.11 & 40.51 & 25.16 & 20.94\\
2 &oct($ols$)& 39.17 & 37.97 & 35.24 & 36.12 & 31.27 & 34.80 & 22.58 & 21.74\\
2 &ite($struc_{cs}, ols_{te}$) & 35.48 & 34.74 & 33.24 & 33.44 & 29.65 & 32.16 & 21.66 & 20.82\\
2 & ite($ols_{cs}, struc_{te}$) & 35.19 & 34.57 & 33.48 & 33.28 & 29.74 & 31.95 & 21.47 & 20.61\\
2 & oct($struc$) & 33.11 & 32.54 & \textbf{31.60} & 31.37 & 28.19 & 30.06 & 20.49 & 19.64\\
2 & te($ols_2$)+cs($ols$) & 34.36 & 34.52 & 35.40 & 36.29 & 35.64 & 45.42 & 26.24 & 21.80\\
2 &  te($struc_2$)+cs($ols$) & 34.34 & 34.50 & 35.38 & 36.26 & 35.62 & 45.40 & 26.23 & 21.79\\
2 & te($ols_2$)+cs($struc$) & 33.40 & 33.07 & 32.82 & 33.64 & 30.70 & 35.26 & 23.20 & 20.62\\
2 & te($struc_2$)+cs($struc$) & \textbf{32.44} & \textbf{32.07} & 31.80 & 32.54 & 29.81 & 34.14 & 22.59 & 19.96\\
\midrule
3 & PERS$_{bu}$ & 59.75 & 56.81 & 54.23 & 52.87 & 46.82 & 49.07 & 33.12 & 30.65\\
3 & 3TIER$_{bu}$ & 53.46 & 50.57 & 48.33 & 46.19 & 41.36 & 40.72 & 29.28 & 25.19\\
3 & base & 53.46 & 47.22 & 44.87 & 44.30 & 39.68 & 42.66 & 30.92 & 25.82\\
3 & oct($ols$)& 51.54 & 48.48 & 45.06 & 45.17 & 39.41 & 42.39 & 28.50 & 26.73\\
3 &ite($struc_{cs}, ols_{te}$) & 49.32 & 46.47 & 43.80 & 43.45 & 38.38 & 40.64 & 27.92 & 26.14\\
3 & ite($ols_{cs}, struc_{te}$) & 48.20 & 45.50 & 43.45 & 42.51 & 37.83 & 39.64 & 27.28 & 25.46\\
3 & oct($struc$) & 46.74 & 44.02 & \textbf{42.05} & \textbf{41.07} & \textbf{36.67} & \textbf{38.16} & \textbf{26.56} & \textbf{24.73}\\
3 & te($ols_2$)+cs($ols$) & 48.52 & 46.20 & 45.02 & 45.58 & 42.49 & 50.21 & 31.50 & 26.71\\
3 &  te($struc_2$)+cs($ols$) & 47.77 & 45.52 & 44.64 & 44.96 & 42.07 & 49.67 & 31.19 & 26.36\\
3 & te($ols_2$)+cs($struc$) & 47.76 & 45.19 & 43.49 & 43.72 & 39.27 & 43.19 & 29.26 & 25.98\\
3 & te($struc_2$)+cs($struc$) & \textbf{46.25} & \textbf{43.69} & 42.24 & 42.12 & 38.05 & 41.53 & 28.38 & 25.03\\
\bottomrule
\end{tabular}
		\end{footnotesize}
	\end{center}
\end{table}

\begin{figure}[H]
	\centering
	\includegraphics[width=0.9\linewidth]{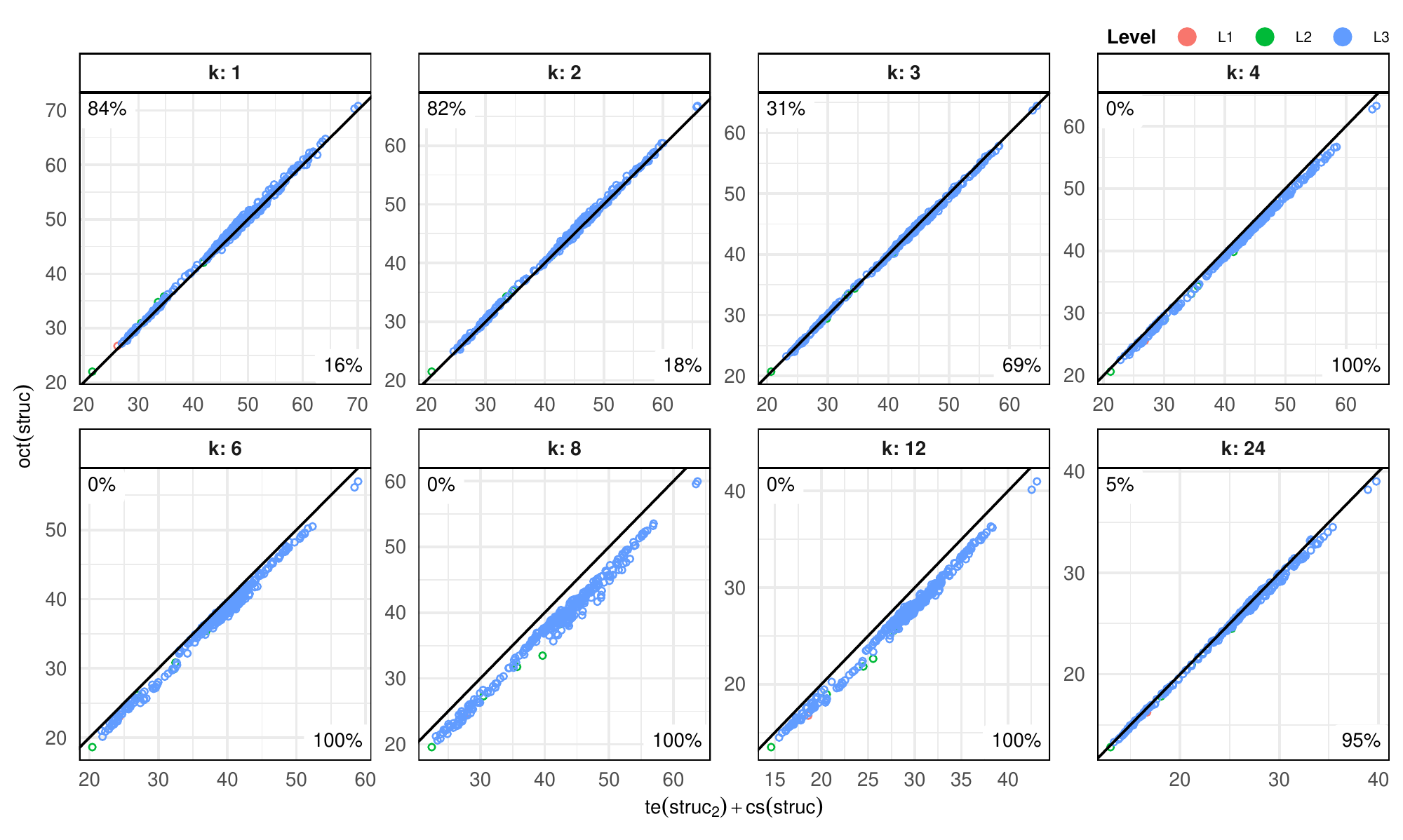}
	\includegraphics[width=0.9\linewidth]{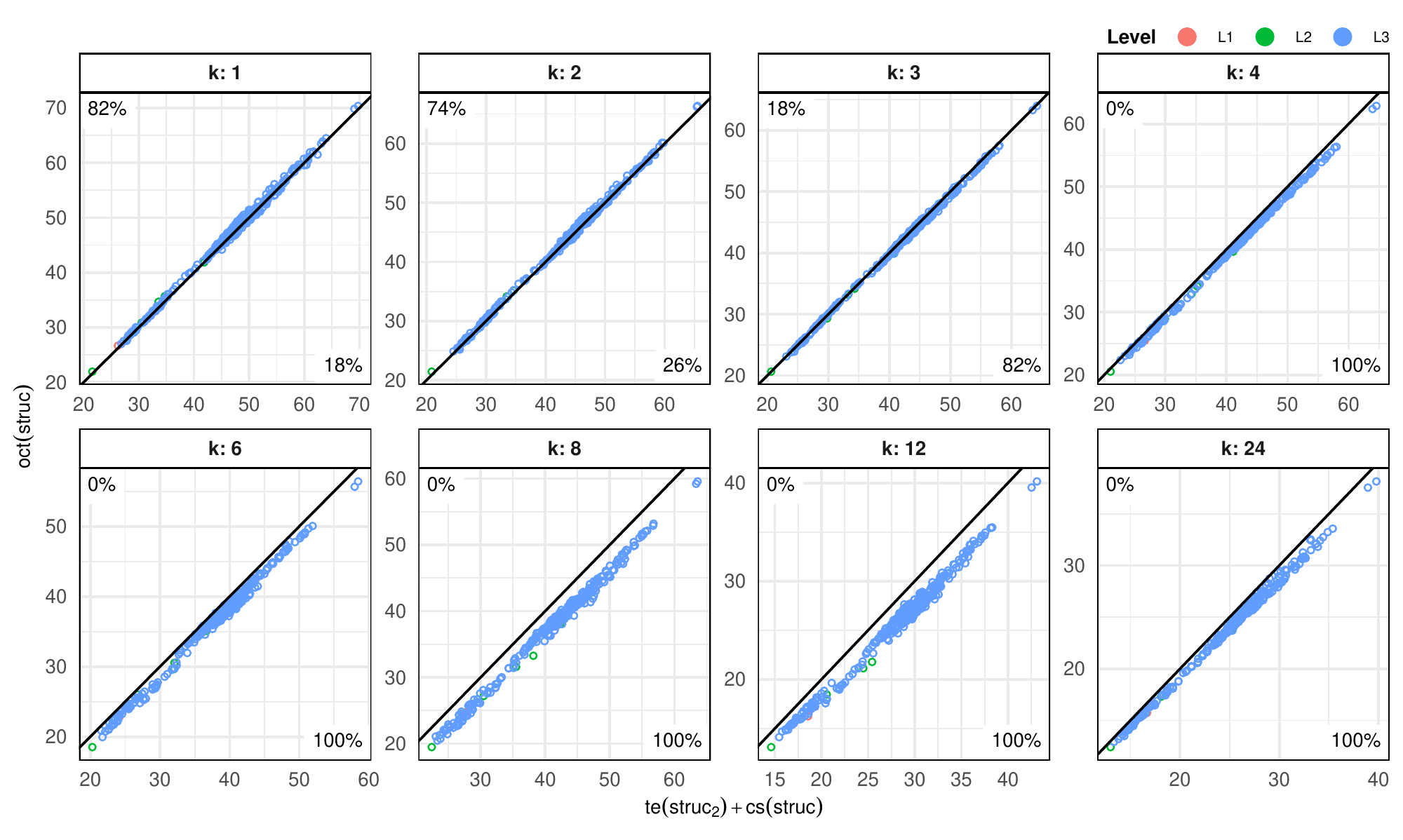}
	\caption{Comparison of nRMSE(\%) between t-struc$_2$ + cs-struc and ct-struc (top panel), and the not negative version (bottom pannel). The black line represents the bisector, where the nRMSE's for both approaches are equal.}
	\label{fig:comparison_ct-struc_vs_tcs-struc2-struc}
\end{figure}

\begin{figure}[H]
	\centering
	\includegraphics[width=1.00\linewidth]{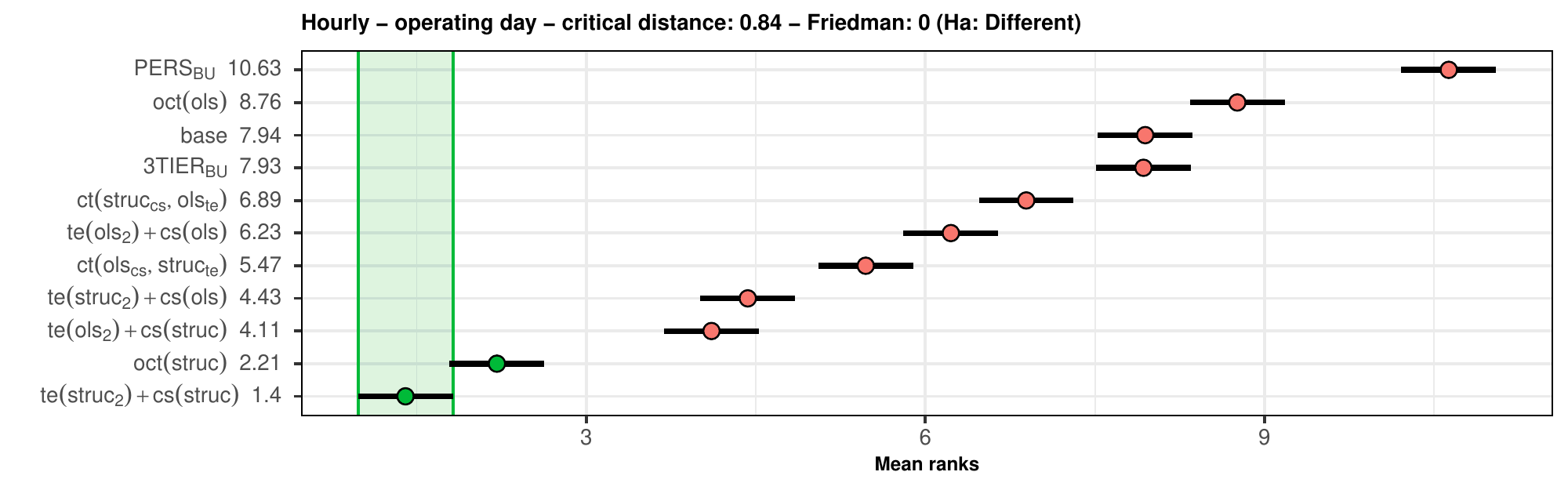}
	\includegraphics[width=1.00\linewidth]{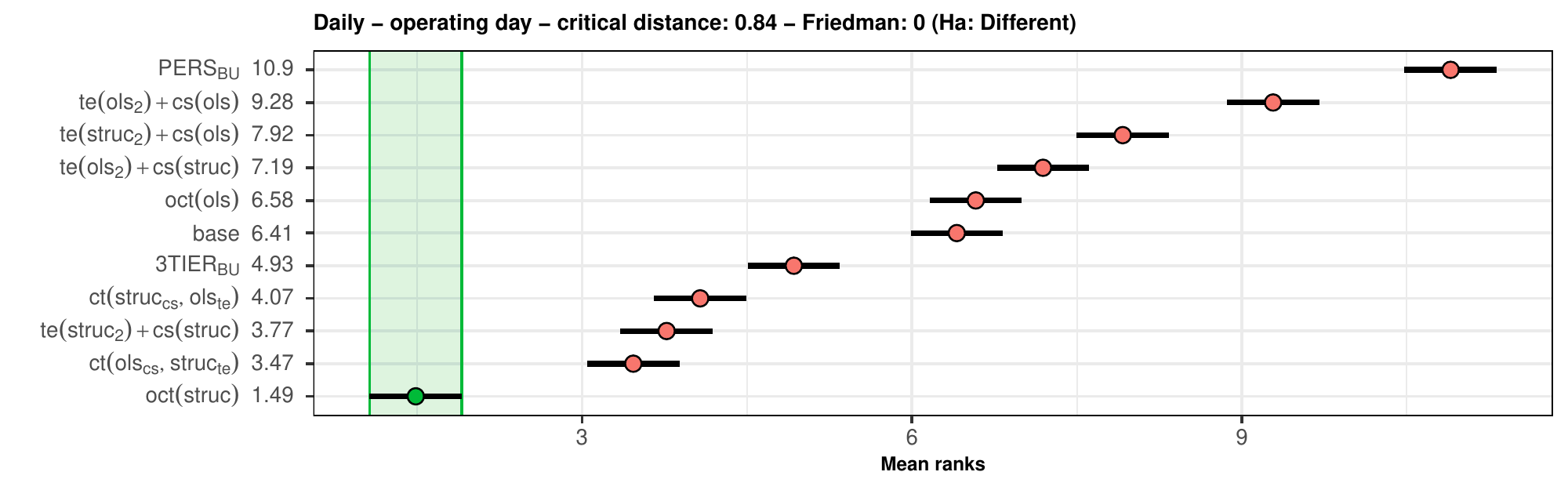}
	\caption{MCB-Nemenyi test on Yang cross-temporal reconciliation approaches with operating day forecast horizon. $\mathcal{L}_0, \mathcal{L}_1, \mathcal{L}_2$ levels (324 series). Top panel: hourly forecasts; Bottom panel: daily forecasts.
			The mean rank of each approach is displayed to the right of their names. If the intervals of two forecast reconciliation approaches do not overlap, this indicates a statistically different performance. Thus, approaches that do not overlap with the green interval are considered significantly worse than the best, and vice-versa.}
	\label{fig:mcb_yang_0}
\end{figure}

\end{document}